\documentclass{amsart}
\usepackage[utf8]{inputenc}
\usepackage{amssymb}
\usepackage{amsthm}
\usepackage{amsfonts}
\usepackage{amsmath}
\usepackage[colorlinks=true, allcolors=black]{hyperref}
\usepackage{multicol,multirow}
\usepackage{dsfont}
\usepackage{pifont}
\usepackage[text={440pt,575pt},headheight=9pt,centering]{geometry}
\usepackage{comment}
\usepackage{enumitem}
\usepackage{float}
\usepackage{pdflscape}
\usepackage{afterpage}
\usepackage{capt-of}
\usepackage{lipsum}
\usepackage{subcaption}
\usepackage{tabularx}
\usepackage{booktabs}
\usepackage{graphicx}
\usepackage{tikz}
\usetikzlibrary{arrows.meta, decorations.pathmorphing, positioning}
\usepackage{algorithm}
\usepackage{algpseudocode}
\algrenewcommand\algorithmicrequire{\textbf{Input:}}
\algrenewcommand\algorithmicensure{\textbf{Output:}}

\let\oldsection\section
\renewcommand{\section}{\vspace{0.2in}\oldsection}

\let\oldsubsection\subsection
\renewcommand{\subsection}{\vspace{0.1in}\oldsubsection}

\newcommand{\vertiii}[1]{{\left\vert\kern-0.25ex\left\vert\kern-0.25ex\left\vert #1 
    \right\vert\kern-0.25ex\right\vert\kern-0.25ex\right\vert}_\infty}

\DeclareMathOperator*{\argmin}{argmin}
\usepackage{tikz}
\usetikzlibrary{shapes, positioning}

\newtheorem{theorem}{Theorem}
\newtheorem{lemma}[theorem]{Lemma}

\newtheorem{definition}{Definition}

\newtheorem{example}{Example}

\newtheorem{assumption}{Assumption}
\newtheorem{remark}{Remark}

\usepackage{graphicx}

\usepackage{placeins}
\usepackage[comma, sort&compress]{natbib} % has a nice set of citation styles and commands

 \usepackage[foot]{amsaddr}

\usetikzlibrary{calc}
\usepackage{fontawesome5}

\title[Integer programming for invariant predictions]{Convex Mixed-Integer Programming for Invariant Predictions from Multi-environment Data}

\date{}

\author[X. Zhang]{Xiaozhu Zhang $^1$}
\address{$^1$ Department of Statistics, University of Washington}
\email{xzzhang@uw.edu}

\author[R. Mazumder]{Rahul Mazumder $^2$}
\address{$^2$ Sloan School of Management,
Operations Research Center
and MIT Center for Statistics,
Massachusetts Institute of Technology}
\email{rahulmaz@mit.edu}

\author[A. Taeb]{Armeen Taeb $^{1}$}
\email{ataeb@uw.edu}

\begin{document}
\begin{abstract}
Heterogeneous data provide an opportunity to identify causal predictors by revealing which predictive relationships remain invariant.  Under a linear model, we introduce a least-squares coefficient-and-risk (LSCR) invariance principle, requiring both ordinary least squares coefficients and prediction risks to remain constant across environments. Our model neither specifies an intervention mechanism nor assumes a structural causal model for the entire system. We develop an invariance score that quantifies departures from LSCR invariance and use its optimization to construct two complementary procedures for causal predictor identification. First, we propose an estimator that jointly minimizes prediction risk and invariance penalty, and establish non-asymptotic guarantees for estimation and recovery of the causal parents under suitable identifiability conditions. Second, without requiring full identifiability, we develop a testing procedure that minimizes the invariance score over all predictor sets excluding a given variable, enabling valid causal discoveries with asymptotic error control via a multiplier bootstrap. Both procedures admit convex mixed-integer optimization formulations that avoid exhaustive subset enumeration while providing certificates of optimality. Simulations and real data experiments demonstrate substantial statistical and computational improvements over existing methods.
%
%We develop statistically rigorous and computationally tractable procedures for causal predictor identification, enabling full recovery under suitable identifiability conditions and valid partial recovery otherwise. First, we propose an estimator that recovers the causal parents under suitable identifiability conditions, with non-asymptotic guarantees for estimation and support recovery. Second, without requiring full identifiability, we develop a testing procedure that identifies predictors shared by all invariant sets, yielding valid causal discoveries with asymptotic error control via a multiplier bootstrap. 

%We address two key challenges: identifying as many causal predictors as possible while controlling erroneous discoveries, and efficiently searching for invariant structures without exhaustive subset enumeration. 
\end{abstract}
\maketitle
\vspace{-0.1in}

%\vspace{-0.1in}
\section{Introduction}

Identifying the causal predictors of a target response is an important problem in causal inference. However, data from a single observational distribution generally cannot distinguish causal predictors from noncausal variables that are associated with the response. Data collected across 
varying settings and conditions, broadly called \emph{environments}, offer an additional source of information for causal discovery. Although interventions or distributional changes may alter the distributions of covariates and their associations with the response, the causal mechanism governing the response remains invariant across environments, provided it is not directly intervened upon.

%that is not directly intervened upon remains invariant across environments.

This idea underlies the invariant causal prediction framework of \citet{peters2016causal}. Let $S^\star$ denote the causal parent set of a target response $Y$. If the causal mechanism for $Y$ is unchanged, then the conditional distribution of $Y$ given $X_{S^\star}$ remains invariant across environments. This establishes a connection between invariance and causality: the true causal parent set must yield an invariant conditional distribution, and any predictor set exhibiting such invariance is a natural candidate for the causal parent set. Motivated by this connection, \citet{peters2016causal} identify invariant predictor sets across environments and take their intersection to recover a subset of $S^\star$. 
%searches for such invariant predictor sets across environments and uses their intersection to discover a subset of the parent set $S^\star$.

Since the work by \citet{peters2016causal}, exploiting invariance across environments has become an important strategy for causal discovery. One line of research estimates invariant predictive models through optimization, typically by incorporating a tractable property of invariance into the estimation objective. Under suitable identifiability conditions, these methods can recover the true causal model. \citet{arjovsky2019invariant} provide an early example with invariant risk minimization (IRM), which 
seeks a data representation such that the same prediction rule minimizes the risk in every environment. However, IRM does not generally guarantee recovery of the causal parent set, even in linear settings, and its practical implementation relies on highly heuristic nonconvex optimization that can be sensitive to initialization and tuning parameters~\citep{rojas2018invariant,Kamath2021DoesIR}.

Other approaches focus on estimating causal regression coefficients. \citet{rothenhausler2019causal} propose Causal Dantzig, which estimates causal coefficients in linear structural models using invariant moment conditions under additive interventions. More recently, \citet{wang2024causal} developed NegDRO, a nonconvex minimax approach that exploits invariance of prediction risks to identify causal coefficients. Although both Causal Dantzig and NegDRO can accommodate hidden confounding, they rely on specific additive intervention mechanisms and assume a linear structural causal model for the entire system, which may limit their applicability. 

In contrast, EILLS, introduced by \citet{fan2024environment}, and CoCo, proposed by \citet{yin2024optimization}, do not require a specific intervention mechanism. Both approaches exploit the first order optimality conditions of least squares regression across environments to identify causal coefficients. EILLS incorporates an invariance penalty into the regression objective, together with indicators for selected features, whereas CoCo relies on continuous nonconvex optimization. Another method is Bayesian invariant prediction (BIP), proposed by \citet{wu2025bayesian}, which uses latent binary indicators to select variables whose conditional relationship with the response remains invariant across environments. %\at{For negdro and causal Dantzig, you also want to mention that they assume a linear SCM on the entire sytem of variables.  }

Among these methods, only EILLS and BIP explicitly select subsets of predictors when identifying an invariant model, while sparse formulations of Causal Dantzig perform variable selection through $\ell_1$ regularization. However, the exact implementations of EILLS and BIP require exhaustive enumeration over candidate predictor subsets, making them computationally prohibitive even in moderate dimensions. Although their approximate implementations alleviate this computational burden, they lack any optimality guarantees. In contrast, $\ell_1$-regularized Causal Dantzig admits a computationally tractable convex formulation, but its reliance on a convex relaxation of variable selection may lead to shrinkage bias and inaccurate support recovery.

The methods discussed above, including IRM, Causal Dantzig, NegDRO, EILLS, CoCo, and BIP, generally require sufficient intervention heterogeneity to guarantee recovery of the causal parent set. Without such heterogeneity, the variables discovered by these methods may include noncausal predictors, and recovery of even a subset of the true causal parent set is not guaranteed.

%Furthermore, their guarantees for recovering the causal parent set require sufficient intervention heterogeneity. When this condition fails, the estimated set may include noncausal variables, with no guarantee of recovering even a subset of the true causal parent set.

An alternative line of work seeks to identify a subset of the causal parent set while controlling type-I error, without requiring full identifiability. As discussed earlier, \citet{peters2016causal} introduce invariant causal prediction (ICP), which tests invariance over all candidate predictor subsets and takes the intersection of the accepted sets. However, exhaustive enumeration of predictor subsets becomes computationally prohibitive as the dimension increases. To address this limitation, \citet{martinet2022variance} formulate causal predictor discovery as a multiple testing problem instead of exhaustive subset enumeration. These tests use Wasserstein variance minimization (WVM) to assess whether the conditional distribution of the response can remain invariant across environments when a given variable is excluded. However, WVM remains computationally expensive. Moreover, both ICP and WVM can have limited statistical power in identifying causal predictors, particularly at moderate sample sizes.
\subsection{Gaps and our contributions}

The preceding discussion highlights two related but distinct tasks in
invariant causal prediction:
\begin{enumerate}[label = (T\arabic*)]
    \item When $S^\star$ is identifiable, how can we accurately
    recover the entire causal parent set?

    \item When $S^\star$ is not fully identifiable, how can we
    identify a subset of $S^\star$ while controlling
    false causal discoveries?
\end{enumerate}

As shown in Table~\ref{tab:related-work}, existing approaches either do not explicitly perform variable selection or generally focus on only one of these two goals. To address both tasks, we introduce an invariance score that quantifies the degree of invariance of a predictor set $S$, with smaller scores indicating greater invariance. For Task~(T1), we incorporate this score as a regularization term in the risk minimization objective to encourage the selection of an invariant predictor set. We establish identifiability conditions for $S^\star$ and show that the resulting estimator recovers $S^\star$ under suitable regularity conditions. For Task~(T2), we adopt the multiple testing framework of \citet{martinet2022variance}. Specifically, for each variable $k$, we search for the minimum invariance score among all predictor sets that exclude $k$. If even this minimum score is sufficiently large, it provides evidence that no invariant predictor set can exclude $k$. Consequently, $k$ must belong to the intersection of all invariant predictor sets and hence to the true causal parent set $S^\star$. We further develop a multiplier bootstrap procedure to approximate the distribution of the minimum score and establish asymptotic type-I error control for false causal discoveries.

\begin{table}[!ht]
\centering
\footnotesize
\setlength{\tabcolsep}{6pt}
\renewcommand{\arraystretch}{1.25}
\begin{tabular}{@{}c|c|c|c|c|c@{}}
\toprule
Method & General interventions & Task 1 & Task 2 & $\ell_0$-Var. selection & Computation \\
\midrule
IRM & \ding{52} & \ding{56} & \ding{56} & \ding{56} & continuous nonconvex \\
Causal Dantzig & \ding{56}  & \ding{52}  & \ding{56} & \ding{56} & convex program \\
NegDRO & \ding{56} & \ding{56} & \ding{56} & \ding{56} & nonconvex minimax \\
EILLS & \ding{52} & \ding{52} & \ding{56} & \ding{52} & enumeration \\
CoCo & \ding{52} & \ding{56}  & \ding{56} & \ding{56} & continuous nonconvex \\
BIP & \ding{52} & \ding{52}   & \ding{56} & \ding{52} & enumeration \\
\midrule
ICP & \ding{52} & only under ident. & \ding{52} & subset testing & $t$-test and $F$-test \\
WVM & \ding{52} & only under ident. & \ding{52} & variable-wise tests & Wasserstein \\
\midrule
Ours & \ding{52} & \ding{52} & \ding{52} & \ding{52} & convex MIP \\
\bottomrule
\end{tabular}
\caption{\small Comparison of related work and our contributions. Note that ICP and WVM can recover $S^\star$ only under identifiability.}
\label{tab:related-work}
\end{table}

Notably, both tasks involve optimization over discrete predictor subsets. To avoid exhaustive subset enumeration, we develop convex mixed-integer programming (MIP) formulations for the invariance score and its subsequent optimization problems. Subset selection is naturally represented through binary variables, and the resulting formulations have convex continuous relaxations, enabling efficient optimization and certificates of global optimality. Unlike existing approaches ICP, EILLS and BIP, our MIP framework allows the use of branch-and-bound algorithms that (i) efficiently prune candidate subsets to reduce the computational burden and (ii) return solutions with guaranteed optimality gaps below a prespecified tolerance. We further develop a scalable coordinate descent heuristic that efficiently finds solutions of high quality. When incorporated into the MIP framework, these solutions can be equipped with certificates of optimality gaps.

Our invariance score is constructed based on the \textit{least-squares coefficient and risk} (LSCR) invariance principle for linear models. This principle requires both the population ordinary least squares (OLS) coefficients and the corresponding OLS risks to remain invariant across environments. Unlike Causal Dantzig and NegDRO, neither the linear structural causal model we consider nor the LSCR principle imposes specific restrictions on the intervention mechanisms, provided that the causal mechanism for the response remains unchanged. Moreover, tests based on the LSCR principle demonstrate greater statistical power than ICP and WVM in our experiments. This improvement can be attributed, in part, to the stronger invariance restrictions imposed by LSCR compared with tests on the regression residuals used in ICP, which can fail to detect changes in regression coefficients. In contrast, WVM involves a highly nonconvex optimization problem, making it susceptible to suboptimal solutions that can compromise its empirical performance. Our MIP framework can also be extended to settings with latent confounding by adopting a different invariance score.

Finally, we demonstrate the performance of our methods using both synthetic data and real data collected from a tightly controlled physical system. First, our estimator for Task~(T1) achieves substantially higher accuracy in recovering causal parents than competing methods as the dimension increases. Second, our procedure for Task~(T2) consistently achieves higher statistical power than competing methods while maintaining type-I error control. For both tasks, our methods remain computationally competitive.

The remainder of the paper is organized as follows. Section~\ref{sec:setup} introduces the problem setup and the LSCR invariance principle. Section~\ref{sec:MIP} presents the MIP formulations for the invariance score and the two complementary tasks. Sections~\ref{sec:task1} and~\ref{sec:task2} investigate Task~(T1) and Task~(T2), respectively. Sections~\ref{sec:experiments} and \ref{sec:realdata} evaluate the proposed approach through synthetic and real data experiments. Section~\ref{sec:extensions} extends our framework to settings with nonlinear causal effects or hidden confounding. Finally, Section~\ref{sec:discussion} discusses future research directions. The code for implementing our methods and reproducing the numerical experiments is available at \url{https://github.com/Xiaozhu-Zhang1998/mip_inv_pred}.

\subsection{Other related literature}
This section reviews three additional lines of related research not discussed earlier and highlights their connections to and distinctions from our proposed framework.

\vspace{0.05in}
\noindent\textbf{Distributional robustness and domain generalization.}
Beyond causal predictor selection, invariance has also been widely used to improve prediction under distributional shifts. One line of work focuses on identifying invariant predictor subsets that generalize to unseen environments. \citet{rojas2018invariant} study this problem in causal transfer learning, while \citet{pfister2021stabilizing} combine regression models based on invariant predictor subsets to improve predictive performance across environments. Another line of work focuses on distributional robustness. \citet{meinshausen2015maximin} introduce maximin effects for robust regression with heterogeneous data, while anchor regression \citep{rothenhausler2021anchor} and DRIG \citep{shen2026causality} aim to achieve robust prediction under distributional shifts caused by interventions. For probabilistic prediction, \citet{henzi2025invariant} propose invariant probabilistic prediction (IPP), which uses proper scoring rules to encourage invariance of prediction risks across environments. While these approaches primarily aim to improve prediction under distributional shifts, our work focuses on identifying causal predictors through invariance across environments.

\vspace{0.05in}
\noindent\textbf{Beyond linear invariant prediction.}
Another line of work extends invariant causal prediction beyond linear models. \citet{heinze2018invariant} develop nonlinear and nonparametric versions of ICP using statistical tests for invariance across environments. \citet{kook2025model} extend ICP to general response types, including categorical, count, and censored outcomes, through transformation models. More recently, \citet{gu2025causality} propose focused adversarial invariant regularization, which uses neural networks and adversarial optimization to identify invariant predictors in nonparametric regression models. These approaches extend invariant causal learning to more flexible model classes, whereas our work focuses on linear models to develop mixed-integer convex formulations with statistical and optimization guarantees. Nonetheless, we describe how our MIP formulations can be extended to handle nonlinear additive models in Section~\ref{sec:extensions}.

\vspace{0.05in}
\noindent\textbf{Mixed-integer optimization in statistical learning.}
Recent advances in mixed-integer programming have made combinatorial optimization increasingly practical for statistical learning. \citet{bertsimas2016best} develop a mixed-integer optimization approach to best subset selection, using initial high-quality solutions to accelerate branch-and-bound algorithms. \citet{bertsimas2020sparse} further introduce a binary convex reformulation and cutting-plane algorithms for scalable $\ell_0$-regularized regression. \citet{hazimeh2020fast} develop efficient coordinate descent and local combinatorial search algorithms for best subset selection. Similar ideas have also been applied to other statistical problems, including multi-task learning \citep{behdin2026multi} and Bayesian network structure learning \citep{manzour2021integer}. These studies demonstrate that exploiting problem structures can substantially improve the computational efficiency of discrete optimization. Convex mixed-integer programming has also been used to learn entire causal graphs \citep{xu2025integer,manzour2021integer,kucukyavuz2023consistent,Zhang2025ConvexMP}, although this is a substantially more challenging problem that typically requires modeling assumptions on the entire system of variables. %Building on these advances, we develop a convex mixed-integer framework for identifying causal predictors of a target response through invariance across environments, without imposing structural modeling assumptions on the entire system.

%\vspace{0.1in}
\subsection{Notation} 
\label{sec:notation}
For a random variable $Z$ with observations $\{Z_i\}_{i=1}^m$, we define the empirical squared norm by $\|Z\|_m^2:=m^{-1}\sum_{i=1}^m Z_i^2$. Similarly, for two random variables $Z$ and $W$ with paired observations $\{(Z_i,W_i)\}_{i=1}^m$, we define the empirical inner product by $\langle Z,W\rangle_m:=m^{-1}\sum_{i=1}^m Z_iW_i$. For any index set $S \subseteq [p]$ and random vector $U \in \mathbb{R}^p$, we use $U_S\in\mathbb{R}^{|S|}$ to denote the subvector of $U$ indexed by $S$. For a vector $w \in \mathbb{R}^p$, we use $\mathrm{supp}(w) = \{j:w_j \neq 0\}$ to denote its support. Lastly, for two sequences $a_n$ and $b_n$, we write $a_n \lesssim b_n$ if there exists a constant $C>0$ and $n_{0} \in \mathbb{N}$ such that $a_{n}\le C b_{n}$ for all $n \geq n_{0}$. We write $a_n \asymp b_n$ if both $a_{n} \lesssim b_{n}$ and $a_{n} \gtrsim b_{n}$ hold. We write $a_n = o(b_n)$ if $a_n /b_n \rightarrow 0$ as $n \rightarrow \infty$.

%We consider multi-environment data $\{(X^e,Y^e)\}_{e\in\mathcal{E}}$, where $X^e\in\mathbb{R}^p$ and $Y^e\in\mathbb{R}$. For any index set $S\subseteq[p]$, we use $X_S^e\in\mathbb{R}^{|S|}$ to denote the subvector of $X^e$ indexed by $S$. We write $X_i^e\in\mathbb{R}^p$ and $Y_i^e\in\mathbb{R}$ for the $i$th observations of $X^e$ and $Y^e$, respectively, and let $n^e$ denote the sample size in environment $e$. We define the total and minimum sample sizes as $n^\star:=\sum_{e\in\mathcal{E}}n^e$ \at{why not $n_\text{total}$ and $n_\text{min}$} and $n:=\min_{e\in\mathcal{E}}n^e$. When needed, we use $\boldsymbol{X}^e\in\mathbb{R}^{n^e\times p}$ and $\boldsymbol{Y}^e\in\mathbb{R}^{n^e}$ to denote the corresponding sample matrix and response vector. 

%\vspace{0.2in}
\section{Setup: causality and invariant predictions from heterogeneous data}
\label{sec:setup}

\vspace{-0.1in}
\subsection{Linear structural causal model}
We suppose we have access to observations of variables under different environments, such as experimental conditions in which some of the variables may have been manipulated, that is, received interventions. To represent this setting, we consider covariates $X \in \mathbb{R}^p$, and response variable $Y \in \mathbb{R}$. 
The interventions on these variables take values in the set $\mathcal{E}$; each $e \in \mathcal{E}$ represents a different environment that generates the random vector $(X^{e},Y^{e})$. We posit that for every $e \in \mathcal{E}$, the covariates $X^e$ and the response $Y^e$ are related via the following linear structural causal model (SCM) for some distribution $\mathcal{F}$:   
%interventions on these variables are generated randomly from adiscrete random variable E taking on values in the set E; eache ∈ E represents a different environment that generates therandom vectors (X e, Y e).
\begin{eqnarray}
\begin{aligned}
Y^e = {\beta^\star}^\top X^{e} + \varepsilon^{e}, \quad \mathrm{supp}(\beta^\star)=S^\star, \quad \varepsilon^{e} \perp X^{e}_{S^\star},\quad \varepsilon^{e} \sim \mathcal{F},\quad  \mathbb{E}[\varepsilon^{e}]=0. %, \quad \mathbb{E}[\varepsilon^{(e)}]=0.
\end{aligned}
\label{eqn:scm_main}
\end{eqnarray}
Here, $\beta^\star \in \mathbb{R}^p$ is a vector encoding causal relationships between the covariates and the response variable, namely $\beta^\star_{j} \neq 0$ if covariate $X^{e}_{j}$ is a causal parent of $Y^e$ in the graph among $(X^{e},Y^{e})$. The subset $S^\star \subseteq \{1,2,\dots,p\}$ represents the covariates that are causal parents of $Y^{e}$, and the coefficients $\beta^\star_{S^\star}$ represent the strength of these causal effects. The SCM \eqref{eqn:scm_main} thus posits that the structure and strength of the causal parents of $Y$ do not change across $e\in \mathcal{E}$. The noise variable $\varepsilon^{e} \in \mathbb{R}$ is assumed to be independent of the causal parents $X^{e}_{S^\star}$. This independence condition corresponds to a causal sufficiency assumption, namely that there is no unobserved confounding simultaneously affecting the covariates and the response variable. In Appendix~\ref{sec:app:confounding}, we discuss extensions to the setting with latent confounding. Importantly, $\varepsilon^{e}$ can have dependencies with other coordinates of $X^{e}$ outside the index set $S^\star$. Such dependencies are expected when some of the covariates are descendants of $Y$ in the causal graph.

The condition that the distribution of $\varepsilon^{e}$ remains the same across all $e\in \mathcal{E}$ means that the response variable does not receive a direct intervention, although its marginal distribution could still be changed due to interventions on its causal ancestors. Finally, the zero-mean condition on $\varepsilon^{e}$ can be assumed without loss of generality, as a nonzero mean can be absorbed into the term ${\beta^\star}^\top X^{{e}}$ by augmenting the covariates with a constant intercept variable. The parameter $\beta^\star$ and the distribution $\mathcal{F}$ are not assumed to be known. 

For every $e \in \mathcal{E}$, we assume that we have $n^e$ independent and identically distributed observations of $(X^e,Y^e)$ from the SCM \eqref{eqn:scm_main}. We further assume that the samples are independent across environments.

\subsection{LSCR-invariant prediction models} 
\label{sec:setup:LSCR}
The SCM \eqref{eqn:scm_main} exhibits important invariances that we highlight in this section. For any $S \subseteq [p]$ and $e\in\mathcal{E}$, define
\begin{eqnarray}
\begin{aligned}
\beta^{e,S} :=  \argmin_{\substack{\tilde\beta\in\mathbb R^p,\
\mathrm{supp}(\tilde\beta)\subseteq S}}
\mathbb E\left[\left( Y^{e}-\tilde\beta^\top X^{e}\right)^2\right], \qquad
\nu^{e,S}:=  \mathbb{E}\left[\left( Y^{e}-{\beta^{e,S}}^\top X^{e}\right)^2\right],
\end{aligned}
\label{eqn:beta_r}
\end{eqnarray}
as the population OLS solution when linearly regressing $Y^e$ on $X_S^{e}$, and the corresponding risk. We propose the following invariance principle based on the OLS coefficients and risk.

\vspace{0.05in}
\begin{definition}[LSCR-invariant prediction model] 
\label{def:LSCR}
Let $\mathrm{supp}(\beta) \subseteq S$ for $\beta\in\mathbb R^p$. We say that $\beta$ is a {least-squares coefficient and risk (LSCR) invariant prediction model} if:
$$ \beta = \beta^{e,S} \quad \text{ and }\quad \nu^{e,S} = \nu^{f,S} \quad  \text{for every }e,f\in\mathcal{E}.$$
The corresponding LSCR-invariant collection is defined as:
\begin{eqnarray*}
\begin{aligned}
%\mathcal{S}_\mathrm{LSC}(\mathcal{E})&:= \{S \subseteq \{1,2,\dots,p\}: \exists \text{ LSC-invariant prediction model }\beta \text{ with support }S\},\\
\mathcal{S}_\mathrm{LSCR}(\mathcal{E})&:=\{S \subseteq \{1,2,\dots,p\}: \beta^{e,S} = \beta^{f,S} \text{ and } \nu^{e,S} = \nu^{f,S}, \ \forall e,f\in\mathcal{E} \},
\end{aligned}
\end{eqnarray*}
which contains all subsets of covariates that result in LSCR-invariant prediction models.

%We say that a coefficient vector $\beta \in \mathbb{R}^p$ with $\mathrm{support}(\beta)=S$ is least squares coefficient (LSC) invariant prediction model if 
%We say that $\beta$ is a least squares coefficient and risk (LSCR) invariant if it is LSC-invariant and for some $\sigma > 0$ $$\text{for all }e\in\mathcal{E},\quad  \mathbb{E}[(Y^e-X^e\beta)^2]=\sigma^2.$$ 
%$$Y^e = X^e\beta + \tilde{\varepsilon}^e, \quad \tilde{\varepsilon}^e \perp X^e_S,\quad   \tilde{\var.$$
%for some distribution $\tilde{\mathcal{F}}$.
\end{definition}

An LSCR prediction model $\beta$ is an OLS solution that remains invariant across environments and whose corresponding residual risk is also constant across environments. Under this definition, the invariant coefficient vector and its support determine each other. Specifically, if $\beta$ is LSCR-invariant, then its support $S=\operatorname{supp}(\beta)$ necessarily belongs to $\mathcal{S}_{\rm LSCR}(\mathcal{E})$. Conversely, if $S\in\mathcal{S}_{\rm LSCR}(\mathcal{E})$, then the associated invariant coefficient is uniquely determined as the population OLS coefficient restricted to $S$. In particular,
\begin{align} \label{eq:model-equiv}
\beta
=
\beta^{e,S}
=
\argmin_{\substack{\tilde\beta\in\mathbb R^p,\
\mathrm{supp}(\tilde\beta)\subseteq S}} \ \sum_{e\in\mathcal{E}}  \  \mathbb{E} \left[ \left( Y^e - \tilde\beta^\top X^e \right)^2  \right].
\end{align}
Thus, the same coefficient is obtained from each individual environment as well as from the pooled population distribution. Throughout the paper, we use the term \emph{model} to refer interchangeably to the coefficient vector $\beta$ and its support $S$ whenever the correspondence is clear from context.

Under the SCM in \eqref{eqn:scm_main}, both the causal coefficient vector $\beta^\star$, which links $X^{e}_{S^\star}$ to $Y^e$, and the distribution of the noise term $\varepsilon^e$ are invariant over $e\in\mathcal{E}$. It therefore follows that $\beta^\star$ is an LSCR prediction model.
However, as illustrated in the following example, unless the data across the environments are sufficiently heterogeneous, the collection $\mathcal{S}_{\rm LSCR}(\mathcal{E})$ may also contain models that include spurious covariates. At one extreme, if there is only a single environment, then every subset $S \subseteq [p]$ belongs to $\mathcal{S}_{\rm LSCR}(\mathcal{E})$.

\vspace{0.05in}
\begin{example}\label{ex:teaser}
Consider the linear SCM in Figure~\ref{fig:illus} with $\mathcal{E}=\{1,2,3\}$. The set $S^\star=\{1,3\}$ denotes the true causal parent set of the response $Y$. We consider several candidate subsets and their corresponding population OLS solutions $\beta^{e,S}$ and risks $\nu^{e,S}$ in Table~\ref{tab:lscr-example}. First, the set $S_1=\{1,3\}$ is the true causal parent set of $Y$, and is therefore LSCR-invariant. In particular, $ \beta^{e,S_1}=\beta^\star$ and $\nu^{e,S_1}=\sigma^2$ for all $e\in\mathcal{E}$. Interestingly, the set $S_2=\{1,2\}$ is also LSCR-invariant, even though $X_2$ is not a direct parent of $Y$. This is because the invariant relationship $X_3^e=\gamma_3X_2^e+\varepsilon_3^e$ allows $X_2$ to serve as an alternative predictor for $X_3$. However, if interventions were also imposed on $\gamma_3$ or $\eta_3$, then $S_2$ would no longer be LSCR-invariant. Thus, distinguishing alternative invariant predictors from the true causal parents requires sufficiently rich intervention mechanisms. In contrast, for $S_3=\{4\}$, both $\beta^{e,S_3}$ and $\nu^{e,S_3}$ vary across environments because $V_1 \neq V_2 \neq V_3$ where $V_e=\operatorname{Var}(Y^e)$. Therefore, $S_3$ is not LSCR-invariant. By enumerating all subsets in this example, we obtain
$$
\mathcal{S}_{\rm LSCR}(\mathcal{E})
=
\left\{
\{1,2\},\ \{1,3\},\ \{1,2,3\},\
\{1,2,4\},\ \{1,3,4\},\ \{1,2,3,4\}
\right\},
$$
where every invariant set contains either $\{1,2\}$ or $\{1,3\}$ as a subset.

% $$
% V_1 = (\beta_1^\star)^2 + (\beta_3^\star\gamma_3)^2 + (\beta_3^\star)^2\eta_3^2 + \sigma^2, \quad
% V_2 = (\beta_1^\star)^2a^2 + (\beta_3^\star\gamma_3)^2 + (\beta_3^\star)^2\eta_3^2 + \sigma^2, \quad
% V_3 = (\beta_1^\star)^2 + (\beta_3^\star\gamma_3)^2a^2 + (\beta_3^\star)^2\eta_3^2 + \sigma^2.
% $$
\label{example:1}
\end{example}

\vspace{0.1in}
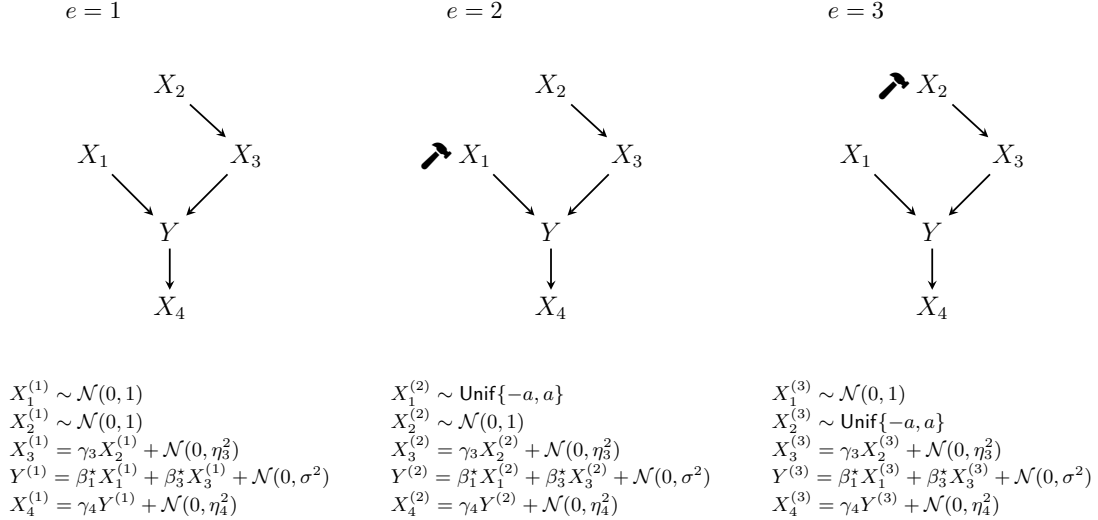
\begin{figure}[!ht]
\centering
\resizebox{0.95\textwidth}{!}{
\begin{tikzpicture}[
    >=stealth,
    var/.style={font=\Large},
    env/.style={font=\large\bfseries},
    eq/.style={font=\small, align=left}
]
%%%%%%%%%%%%%%%%%%%%%%%%%%%%%%%%%%%%%%%%%%%%%%%%%%
%% Environment 1
%%%%%%%%%%%%%%%%%%%%%%%%%%%%%%%%%%%%%%%%%%%%%%%%%%
\node[env] at (-2.4,2.7) {$e=1$};
\node[var] (x2a) at (-1.2,1.5) {$X_2$};
\node[var] (x3a) at (0,0.4) {$X_3$};
\node[var] (ya)  at (-1.2,-0.8) {$Y$};
\node[var] (x1a) at (-2.4,0.4) {$X_1$};
\node[var] (x4a) at (-1.2,-2.0) {$X_4$};
\draw[->,thick] (x2a) -- (x3a);
\draw[->,thick] (x1a) -- (ya);
\draw[->,thick] (x3a) -- (ya);
\draw[->,thick] (ya) -- (x4a);
\node[eq] at (-1.2,-4.2) {
$X_1^{(1)} \sim \mathcal N(0,1)$ \\
$X_2^{(1)} \sim \mathcal N(0,1)$ \\
$X_3^{(1)} = \gamma_3 X_2^{(1)} + \mathcal{N}(0,\eta_3^2)$\\
$Y^{(1)} = \beta_1^\star X_1^{(1)} + \beta_3^\star X_3^{(1)} + \mathcal{N}(0,\sigma^2)$\\
$X_4^{(1)} = \gamma_4 Y^{(1)} + \mathcal{N}(0,\eta_4^2)$
};
%%%%%%%%%%%%%%%%%%%%%%%%%%%%%%%%%%%%%%%%%%%%%%%%%%
%% Environment 2
%%%%%%%%%%%%%%%%%%%%%%%%%%%%%%%%%%%%%%%%%%%%%%%%%%
\node[env] at (3.6,2.7) {$e=2$};
\node[var] (x2b) at (4.8,1.5) {$X_2$};
\node[var] (x3b) at (6,0.4) {$X_3$};
\node[var] (yb)  at (4.8,-0.8) {$Y$};
\node[var] (x1b) at (3.6,0.4) {$X_1$};
\node[var] (x4b) at (4.8,-2.0) {$X_4$};
\draw[->,thick] (x2b) -- (x3b);
\draw[->,thick] (x1b) -- (yb);
\draw[->,thick] (x3b) -- (yb);
\draw[->,thick] (yb) -- (x4b);
% hammer indicating intervention
\node at ($(x1b)+(-0.6,0)$) {\Large\faHammer};
\node[eq] at (4.8,-4.2) {
$X_1^{(2)} \sim \mathsf{Unif}\{-a, a\}$\\
$X_2^{(2)} \sim \mathcal{N}(0,1)$\\
$X_3^{(2)} = \gamma_3 X_2^{(2)} + \mathcal{N}(0,\eta_3^2)$\\
$Y^{(2)} = \beta_1^\star X_1^{(2)} + \beta_3^\star X_3^{(2)} + \mathcal{N}(0,\sigma^2)$\\
$X_4^{(2)} = \gamma_4 Y^{(2)} + \mathcal{N}(0,\eta_4^2)$
};
%%%%%%%%%%%%%%%%%%%%%%%%%%%%%%%%%%%%%%%%%%%%%%%%%%
%% Environment 3
%%%%%%%%%%%%%%%%%%%%%%%%%%%%%%%%%%%%%%%%%%%%%%%%%%
\node[env] at (9.6,2.7) {$e=3$};
\node[var] (x2c) at (10.8,1.5) {$X_2$};
\node[var] (x3c) at (12,0.4) {$X_3$};
\node[var] (yc)  at (10.8,-0.8) {$Y$};
\node[var] (x1c) at (9.6,0.4) {$X_1$};
\node[var] (x4c) at (10.8,-2.0) {$X_4$};
\draw[->,thick] (x2c) -- (x3c);
\draw[->,thick] (x1c) -- (yc);
\draw[->,thick] (x3c) -- (yc);
\draw[->,thick] (yc) -- (x4c);
% hammer indicating intervention
\node at ($(x2c)+(-0.6,0)$) {\Large\faHammer};
\node[eq] at (10.8,-4.2) {
$X_1^{(3)} \sim \mathcal{N}(0,1)$\\
$X_2^{(3)} \sim \mathsf{Unif}\{-a,a\} $\\
$X_3^{(3)} = \gamma_3 X_2^{(3)} + \mathcal{N}(0,\eta_3^2)$\\
$Y^{(3)} = \beta_1^\star X_1^{(3)} + \beta_3^\star X_3^{(3)} + \mathcal{N}(0,\sigma^2)$\\
$X_4^{(3)} = \gamma_4 Y^{(3)} + \mathcal N(0,\eta_4^2)$
};
\end{tikzpicture}
}
\caption{\small SCM over four variables across three environments. The hammer indicates the intervened variable. Environment 1 represents an observational setting, and in environments 2 and 3, variables $X_1$ and $X_2$ have received stochastic interventions, respectively.}
\label{fig:illus}
\end{figure}

\begin{table}[!ht]
\centering
{\small
\begin{tabular}{c|c|c|c}
\toprule
 & $S_1=\{1,3\}$ & $S_2=\{1,2\}$ & $S_3=\{4\}$ \\
\midrule
$\beta^{e,S}$
&
$\begin{pmatrix}
\beta_1^\star\\
\beta_3^\star
\end{pmatrix}$
&
$\begin{pmatrix}
\beta_1^\star\\
\beta_3^\star\gamma_3
\end{pmatrix}$
&
$\displaystyle
\frac{\gamma_4V_e}
{\gamma_4^2V_e+\eta_4^2}
$
\\[4mm]
$\nu^{e,S}$
&
$\sigma^2$
&
$(\beta_3^\star)^2\eta_3^2+\sigma^2$
&
$\displaystyle
\frac{V_e\eta_4^2}
{\gamma_4^2V_e+\eta_4^2}
$
\\[4mm]
LSCR-invariant
&
Yes
&
Yes
&
No
\\
\bottomrule
\end{tabular}
}
\caption{\small Population OLS coefficients and risks for different candidate predictor sets. Here $V_e:=\operatorname{Var}(Y^e)$.}
\label{tab:lscr-example}
\end{table}

\vspace{0.05in}
Although the LSCR principle requires invariance only of the first two conditional moments, the SCM \eqref{eqn:scm_main} may exhibit even stronger invariance. In particular, the conditional distribution $Y^e\mid X^e_{S^\star}$ is the same for all environments $e\in\mathcal{E}$. Note that when $Y^e\mid X_S^e\sim \mathcal{N}( X^{e\top}\beta^{e,S},\nu^{e,S})$, LSCR invariance is equivalent to invariance of the full conditional distribution. More generally, this equivalence continues to hold whenever the conditional distribution is fully characterized by $\mathbb{E}\left[ Y^e \mid X^e_S \right] = X^{e\top} \beta^{e,S}$ and $\mathrm{Var}\left[ Y^e \mid X^e_S \right] = \nu^{e,S}$. In addition to being equivalent under these settings, LSCR invariance is considerably more amenable to convex mixed-integer optimization formulations, which enable optimality guarantees. We therefore focus on LSCR invariance throughout the paper. 
%Throughout the paper, we focus on the LSCR invariance.\at{need to say why? for computational reasons}

%\vspace{0.1in}
\subsection{Comparison to other invariance notions}

We compare the LSCR invariance principle with other notions of invariance used in related work.
First, the LSCR principle is essentially the same as that used in ICP~\citep{peters2016causal}, but the two methods employ different testing procedures. For each candidate predictor set $S$, ICP fits a pooled linear regression and tests whether the residual means and variances are invariant across environments using $t$-tests and $F$-tests with Bonferroni correction. However, differences in environment specific regression coefficients may not manifest as differences in the means or variances of the pooled residuals. Consequently, for a subset $S$, ICP can fail to detect violations of invariance in the conditional distribution $Y^e|X_S^e$ even at the population level, as illustrated in the following example.

\vspace{0.05in}
\begin{example}Consider two environments with the invariant causal mechanism $Y^e = X_1^e + \varepsilon^e$, $e \in \{1,2\}$. Suppose additionally, $X_1^e \sim \mathcal{N}(0,v_e)$ and $X_2^e = X_1^e + U^e$, $U^e \sim \mathcal{N}(0,u_e)$ with $X_1^e$, $\varepsilon^e$ and $U^e$ being mutually independent. Suppose $v_1=1$, $v_2=4$, $u_1=4$ and $u_2=1$. Consider set $S = \{2\}$. It is straightforward to see that the distribution $Y^e|X_S^e$ is not invariant across $e$, since $\beta^{1,S} \neq \beta^{2,S}$. However,
a straightforward calculation shows that the pooled residuals have identical distributions across the two environments. Consequently, the criterion used by ICP fails to reject $S$. Since ICP takes the intersection of all non-rejected sets, retaining $S=\{2\}$ alongside $\{1\}$ results in an empty intersection, thereby failing to identify the parent set $\{1\}$.  
\end{example}
This example illustrates the limits of the tests ICP uses for quantifying invariance and helps explain the greater statistical power of our MIP testing procedure for moderate sample sizes, as shown in Figure~\ref{fig:teaser1}.

Second, the condition $\beta = \beta^{e,S}$ for all $e\in\mathcal{E}$ in Definition~\ref{def:LSCR} is essentially equivalent to the invariance principle $\mathbb{E}\left[(Y^e-\beta_S^\top X_S^e)X_S^e\right]=0$ for all $e\in\mathcal{E}$, used in EILLS~\citep{fan2024environment}. However, by additionally requiring invariance of $\nu^{e,S}$ across environments, our proposed LSCR principle can reduce the collection of invariant predictor sets. As illustrated in Example~\ref{ex:identification}, this additional requirement also weakens the identification condition. Intuitively, exploiting more information from the conditional distribution helps distinguish incorrect models through comparisons across environments.

Relatedly, IRM~\citep{arjovsky2019invariant} seeks a representation under which the same prediction rule minimizes risk across environments. When the representation selects a predictor subset and the loss is squared error, this condition reduces to invariance of the corresponding population OLS coefficients. CoCo~\citep{yin2024optimization} similarly exploits coefficient invariance through first order optimality conditions. Like EILLS, both methods rely solely on coefficient invariance, potentially admitting more invariant predictor sets and requiring stronger conditions for causal parent identification. 

%Causal Dantzig~\citep{rothenhausler2019causal} instead exploits invariance of the inner product $\mathbb{E}[X^e(Y^e-\beta^\top X^e)]$ across environments. Unlike LSCR, this principle does not require the moment condition to vanish in each environment, allowing it to accommodate hidden confounding under additive interventions. However, it relies on specific assumptions about the intervention mechanisms. 
%Third, NegDRO~\citep{wang2024causal} considers a risk-invariance principle that requires $\mathbb{E}[(Y^e-\beta^\top X^e)^2]=\mathbb{E}[(Y^f-\beta^\top X^f)^2]$ for all $e,f\in\mathcal{E}$.
%This principle compares unconditional prediction risks and does not require the outcome noise to be independent of the causal predictors, allowing NegDRO to accommodate hidden confounding under suitable intervention assumptions. In contrast, our LSCR principle is based on the conditional relationship between $Y^e$ and a selected predictor set $X_S^e$, requiring both the corresponding population OLS coefficients $\beta^{e,S}$ and minimum OLS risks $\nu^{e,S}$ to be invariant across environments. While our principle does not generally accommodate hidden confounding, it directly connects invariance to the selection of predictor subsets.  The resulting combinatorial optimization problem is addressed through our MIP framework.
Third, Causal Dantzig~\citep{rothenhausler2019causal} exploits invariance of the moment $\mathbb{E}[X^e(Y^e-\beta^\top X^e)]$ across environments, without requiring it to vanish. NegDRO~\citep{wang2024causal} instead uses risk invariance, requiring $\mathbb{E}[(Y^e-\beta^\top X^e)^2]=\mathbb{E}[(Y^f-\beta^\top X^f)^2]$ for all $e,f\in\mathcal{E}$. Both methods can accommodate hidden confounding, but require specific additive intervention mechanisms and a linear structural model for the entire system. Moreover, NegDRO does not explicitly perform variable selection. %In contrast, LSCR requires invariance of both the population OLS coefficients $\beta^{e,S}$ and minimum OLS risks $\nu^{e,S}$ for a selected predictor set $S$. %Unlike Causal Dantzig and NegDRO, LSCR does not generally accommodate hidden confounding. However, it explicitly evaluates invariance over candidate predictor subsets, enabling causal parent identification through our MIP framework. %without exhaustive subset enumeration.

Finally, \cite{martinet2022variance} consider a distributional invariance principle that requires the full distributions of prediction residuals to remain identical across environments, namely $Y^e-f(X_S^e)\overset{d}{=}Y^f-f(X_S^f)$ for all $e,f\in\mathcal E$. Compared with LSCR, this principle exploits information beyond the first two moments and may therefore provide greater statistical power in certain settings. However, testing this invariance requires computationally demanding, generally nonconvex Wasserstein variance minimization, making variable selection challenging in practice. A related approach is BIP~\citep{wu2025bayesian}, which models invariance of the full conditional distribution through a Bayesian framework, requiring $p_e(y\mid x_S)=p_f(y\mid x_S)$ for all $e,f\in\mathcal{E}$. BIP introduces latent binary indicators for variable selection and performs posterior inference over candidate predictor sets. However, its exact implementation requires exhaustive subset enumeration.

\section{Invariant learning via convex mixed-integer programming}
\label{sec:MIP}

\vspace{-0.1in}
\subsection{An invariance score and two invariant learning tasks}
\label{sec:score}
A key component of invariant learning is to quantify the extent to which a candidate model remains stable across environments. Motivated by the LSCR principle, we define the following invariance scores, in population and finite sample, for a given subset $S\subseteq [p]$:
\begin{align}
\begin{split}
    \mathsf{I}(S) &:= \sum_{\substack{e,f \in \mathcal{E} \\ e<f}} \omega^{e,f}\Big[\left\| \beta^{e,S} - \beta^{f,S} \right\|_2^2 + \rho\left( \nu^{e,S} - \nu^{f,S} \right)^2\Big],\\
    \widehat{\mathsf{I}}(S)&:= \sum_{\substack{e,f \in \mathcal{E} \\ e<f}} \omega^{e,f}\Big[\left\| \widehat\beta^{e,S} - \widehat\beta^{f,S} \right\|_2^2 + \rho\left( \widehat\nu^{e,S} - \widehat\nu^{f,S} \right)^2\Big].
\end{split}
\label{eqn:score_pop}
\end{align}
%where $\widehat{\beta}^{e,S}$ and $\widehat{\nu}^{e,S}$ are e
%\at{why do i need two regularization parameters? Why not just have one here?}
The score $\mathsf{I}(S)$ measures violations of the two invariance conditions in Definition~\ref{def:LSCR} through the weighted pairwise discrepancies between $\beta^{e,S}$ and $\beta^{f,S}$, and between $\nu^{e,S}$ and $\nu^{f,S}$, across environments. Here, the weights $\omega^{e,f}>0$ satisfy $\sum_{e,f\in\mathcal{E},e<f} \omega^{e,f}=1$, and we typically choose $\omega^{e,f}= n^en^f/m$ where $m = \sum_{\substack{e',f' \in \mathcal{E},e'<f'}}n^{e'}n^{f'}$. Moreover, $\rho > 0$ is a regularization parameter that weighs different invariance conditions; we discuss its choices below.
By construction,  $\mathsf{I}(\cdot)$ is non-negative and attains the value of zero if and only if $S\in \mathcal{S}_{\rm LSCR}(\mathcal{E})$. We therefore use $\mathsf{I}(S)$ to quantify the degree of invariance of $S$, with smaller values indicating greater invariance across environments. 

The corresponding finite-sample invariance score, $\widehat{\mathsf{I}}(S)$, replaces ${\beta}^{e,S}$ and ${\nu}^{e,S}$ by their empirical counterparts $\widehat {\beta}^{e,S} $ and $\widehat {\nu}^{e,S}$. This empirical invariance score can serve as a building block for a broad range of invariant learning tasks. We describe these tasks below. %Using the reformulation in \eqref{eq:inv-score}, each can be cast as a convex mixed-integer program. 

%\at{still missing that these can be formulated as convex mip?}
%Small values of $\widehat{\mathsf{I}}(S)$ indicate that model $S$ is stable across environments, suggesting that $\widehat{\mathsf{I}}(S)$ can be used to favor invariant predictors. 

\vspace{0.1in}
\noindent\textbf{(Task 1) Invariance-regularized prediction.} Let $\widehat{\mathcal{R}}(\beta)$ denote a weighted mean squared error computed over the pooled data, of a linear model with coefficient vector $\beta \in \mathbb{R}^p$. (We will describe its definition in Section~\ref{sec:task1}). We consider the following optimization problem, which trades off predictive accuracy and invariance across environments while encouraging sparsity through regularization parameters $\lambda_n,\gamma\geq 0$:
\begin{align}
\label{eq:related-task-1}
\min_{\substack{S\subseteq[p],\, \beta\in\mathbb{R}^p}}
\quad
\widehat{\mathcal{R}}(\beta)+ \gamma\widehat{\mathsf{I}}(S) + \lambda_n |S|
\quad\text{subject to: }
\mathrm{supp}(\beta) \subseteq S.
\end{align}
Here, the prediction term $\widehat{\mathcal{R}}(\beta)$ favors models with good predictive performance across environments, the invariance term $\widehat{\mathsf{I}}(S)$ favors sets of predictors with strong invariance across environments, and the sparsity term $|S|$ penalizes the inclusion of additional predictors. The coefficient vector $\beta$ is linked to the invariance and sparsity terms only through the predictor set $S$. In particular, both $\beta$ and the environment specific regression coefficients $\{\widehat{\beta}^{e,S}\}_{e\in\mathcal{E}}$ (that define $\widehat{\mathsf{I}}(S)$) are restricted to the predictors in $S$. For any fixed $S$, the optimal $\beta$ is simply the weighted least-squares estimate obtained by regressing the pooled response $\{Y^e\}_{e\in\mathcal{E}}$ on the pooled predictors $\{X_S^e\}_{e\in\mathcal{E}}$.

Note that the prediction term $\widehat{\mathcal{R}}(\beta)$ can help distinguish among multiple invariant sets. In particular, minimizing $\gamma\widehat{\mathsf{I}}(S)+\lambda_n|S|$ alone can be less powerful in recovering the true set of causal parents $S^\star$, since other subsets may also exhibit invariance across environments. Example~\ref{example:1} illustrates this point: when interventions occur only in environments $e=2$ and $e=3$, both $S_1$ and $S_2$ are LSCR-invariant, but the causal parent set $S_1$ achieves a smaller prediction error than $S_2$. In Section~\ref{sec:task1}, we establish conditions under which the estimator \eqref{eq:related-task-1} recovers the causal parent set $S^\star$ in finite samples.

%We describe choices for $\lambda_n,\gamma$ below. Note that the prediction term $\widehat{\mathcal{R}}(\beta)$ can help distinguish among multiple invariant sets, thereby improving the ability to recover the true set of causal parents $S^\star$. Example~\ref{example:1} illustrates this point: when interventions occur only in environments $e=2$ and $e=3$, both $S_1$ and $S_2$ are LSCR-invariant, but the causal parent set $S_2$ achieves a smaller prediction error than $S_1$. In Section~\ref{sec:task1}, we establish conditions under which the estimator \eqref{eq:related-task-1} recovers the causal parent set $S^\star$ in finite samples. %We also use \eqref{eq:inv-score} to reformulate the estimator as a convex mixed-integer program, enabling its globally optimal solutions to be computed.

\iffalse 
We can use \eqref{eq:inv-score} to obtain the following reformulation of \eqref{eq:related-task-1} as a mixed-integer program:
\begin{align}
\min_{\substack{g\in\{0,1\}^p,\beta\in\mathbb{R}^p}}
\quad
\widehat{\mathcal{R}}(\beta)+ \widehat{\mathcal{I}}(g)\quad\text{subject to: }\beta_i = 0 \text{ if }g_i = 0 \quad\forall i \in [p].
\label{eq:related-task-1-g}
\end{align}
\fi

%analyze the estimator \eqref{eq:related-task-1-g}, reformulate it as a convex mixed-integer program, and establish finite-sample guarantees.

\vspace{0.1in}
\noindent\textbf{(Task 2) Testing for causal parents.} 
Unless there is enough heterogeneity across the environments, solving \eqref{eq:related-task-1} or directly minimizing the invariance score $\widehat{\mathsf{I}}(\cdot)$ is not, in general, sufficient to identify $S^\star$. However, for each $j_0 \in [p]$, the constrained optimization problem
\begin{align} \label{eq:related-task-2}
\min_{S\subseteq[p]}
\quad
\widehat{\mathsf{I}}(S)
\qquad
\text{subject to: }j_0 \notin S
\end{align}
can be used to develop a testing procedure for identifying $S(\mathcal{E}):=
\bigcap_{S\in\mathcal{S}_{\rm LSCR}(\mathcal{E})} S$. Since $S(\mathcal{E})\subseteq S^\star$, the variables contained in $S(\mathcal{E})$ may be interpreted as a subset of causal parents. The optimal objective value of~\eqref{eq:related-task-2} quantifies the strongest degree of invariance attainable among models that exclude a prescribed feature $j_0$. Intuitively, a large optimal objective value indicates that excluding $j_0$ substantially disrupts invariance, suggesting that $j_0\in S(\mathcal{E})$. We therefore use this optimal objective value to construct a test statistic. In Section~\ref{sec:task2}, we develop a bootstrap testing procedure for identifying $\mathcal{S}(\mathcal{E})$. The procedure uses test statistics obtained from solutions to \eqref{eq:related-task-2} and provides asymptotic type-I error control. %We also use \eqref{eq:inv-score} to cast the computation of these test statistics as a convex mixed-integer program.

\vspace{0.05in}
\begin{remark} \label{rem:top-k}
We may wish to enumerate top-$k$ invariant sets. This enumeration can be formulated naturally as a sequence of optimization programs. We first initialize the procedure by setting $\widehat S^1\in \argmin_{S\subseteq[p]} \widehat{\mathsf{I}}(S)$ corresponding to the most invariant model. Suppose that, after $L < k$ iterations, we have obtained distinct models $\widehat S^1,\ldots,\widehat S^L$. To obtain the next distinct model, we solve 
\begin{eqnarray*}
\begin{aligned}
\widehat S^{L+1} \in \argmin_{S \subseteq [p]}&\quad \widehat{\mathsf{I}}(S)\quad 
\text{subject to: } S \neq S^{\ell}\quad \forall \ell \in [L].%\sum_{i:\widehat{g}^\ell_{ij} = 1}(1-g_{i})+ \sum_{i:\widehat{g}^\ell_{i} = 0}g_{i}\geq 1 \quad \forall \ell \in [L].
\end{aligned}
\end{eqnarray*}
In Appendix~\ref{sec:app:top-k}, we describe how this sequence of optimization problems can be formulated as convex quadratic MIPs. 
\end{remark}

\vspace{0.05in}
\begin{remark}
The two tasks involve different tuning parameters. For Task~1, we select $(\gamma,\rho,\lambda_n)$. In practice, we choose these parameters by leave-one-environment-out (LOEO) cross-validation: each environment is successively held out, the estimator is fitted using the remaining environments, and the parameter values minimizing the average held-out mean squared error are selected. For Task~2, only the relative weight $\rho$ needs to be specified. In practice, we could adopt the value of $\rho$ selected for Task~1 when the two tasks are applied to the same data. For both tasks, the default choice is $\rho=1$.
\label{remark:tuning parameter}
\end{remark}

\vspace{0.05in}
\begin{remark} \label{rem:extended-score}
    In addition to the coefficient and risk conditions $\beta=\beta^{e,S}$ and $\nu^{e,S}=\nu^{f,S}$ for all $e,f\in\mathcal E$ in the LSCR invariance principle, we may consider two additional moment conditions for a candidate model $\beta$ to improve statistical power: (i) $\mathbb E[Y^e-X^{e\top}\beta]=\mathbb E[Y^f-X^{f\top}\beta]$ for all $e,f\in\mathcal E$; and (ii) $\mathbb E[X_j^e(Y^e-X^{e\top}\beta)]=\mathbb E[X_j^f(Y^f-X^{f\top}\beta)]$ for all $j\in\operatorname{supp}(\beta)$ and $e,f\in\mathcal E$. Under our assumptions, both conditions are automatically satisfied by any LSCR invariant model and therefore do not change the collection $\mathcal S_{\rm LSCR}(\mathcal E)$. Nevertheless, incorporating their empirical discrepancies into the invariance score may help distinguish non-invariant models and improve statistical power. 
    We define the extended invariance score as
    \begin{align*}
    \begin{split}
        \widehat{\mathsf{I}}_+(S)
    := \widehat{\mathsf{I}}(S) + 
    \gamma_3  \sum_{\substack{e,f \in \mathcal{E} \\ e<f}} \omega^{e,f} \left(\widehat m^{e,S}-\widehat m^{f,S} \right)^2 + \gamma_4\max_{j\in S} \sum_{\substack{e,f\in\mathcal E\\e<f}} \omega^{e,f} (\widehat m_j^{e,S}-\widehat m_j^{f,S})^2,
    \end{split}
    \end{align*}
    where $\widehat m^{e,S}:=\langle
    Y^e-X^{e\top}\widehat{\beta}^S,\mathbf 1
    \rangle_{n^e}$ and $\widehat{m}^{e,S}_j:=\langle
    Y^e-X^{e\top}\widehat\beta^S,X_j^e \rangle_{n^e}$, and $\widehat{\beta}^S$ denotes the pooled OLS estimator restricted to $S$.
    We primarily use this extended score for Task 2 to improve statistical power.
    Details are provided in Appendix~\ref{sec:app:gain-more-power}, where we also show that this invariance score is amenable to convex mixed-integer programming.
\end{remark}

% \at{For MIP-4 test, you are not using just a modified invariance score. As a result, I would define it here as well so it can be referred to late. }

%\vspace{0.1in}
\subsection{Reformulation of optimization problems as convex mixed-integer programs}
The formulations \eqref{eq:related-task-1} and \eqref{eq:related-task-2} do not lead to convenient optimization procedures. This difficulty arises for two reasons. First, for a given support $S$, the quantity $\widehat\beta^{e,S}$ in $\widehat{\mathsf{I}}(S)$ is determined by $S$ through the implicit OLS normal equations. Second, the squared discrepancy $(\widehat\nu^{e,S}-\widehat\nu^{f,S})^2$ in $\widehat{\mathsf{I}}(S)$ generally induces a quartic and non-convex dependence on the regression coefficients, making it unsuitable for exact optimization procedures. 

%The formulation of the finite-sample invariance score, however, does not immediately yield a convenient objective.  First, for a given support $S$, the quantities $\widehat\beta^{e,S}$ and $\widehat\nu^{e,S}$ are linked to $S$ through the implicit OLS normal equations. Second, the squared discrepancy $(\widehat\nu^{e,S}-\widehat\nu^{f,S})^2$ generally induces a quartic and non-convex dependence on the regression coefficients, making it unsuitable for exact optimization procedures. 

Fortunately, the invariance score can be cast within a mixed-integer programming framework, leading to convex MIP formulations for \eqref{eq:related-task-1} and \eqref{eq:related-task-2}. We introduce two sets of new decision variables: (i) A set of binary variables $\{g_j \in \{0,1\}\}_{j=1}^p$, where $g_j = 1$ whenever $j \in S$ and $g_j = 0$ otherwise; (ii) A set of continuous variables $\{\beta^e \in\mathbb{R}^p \}_{e\in \mathcal{E}}$ that parameterize the coefficient for each environment. Then, $\widehat{\mathsf{I}}(S)$ can be written as the solution to the following convex MIP:
\begin{subequations}
\begin{align}
\widehat{\mathsf{I}}(S)= \widehat{\mathcal I}(g)
:=
\min_{\{\beta^e\in\mathbb{R}^p\}_{e\in\mathcal E}}
~&
\sum_{\substack{e,f\in\mathcal E\\ e<f}}\omega^{e,f}\Big[
\|\beta^e-\beta^f\|_2^2
+\rho\bigl(d^{e,f}(\beta^e,\beta^f)\bigr)^2\Big],
\nonumber
\\
\text{subject to:}~&
-M_n g_j \leq \beta_j^e \leq M_n g_j,
&& e\in\mathcal E,\ j\in[p],
\label{eq:inv-score-selection}
\\
&
-\widetilde{M}_n(1-g_j) \leq \Big\langle
X_j^e,\,
Y^e-\sum_{k\in[p]}\beta_k^eX_k^e
\Big\rangle_{n^e}\leq \widetilde{M}_n(1-g_j),
&& e\in\mathcal E,\ j\in[p].
\label{eq:inv-score-normal}
\end{align}
\label{eq:inv-score}
\end{subequations}
Here, the heterogeneity in the OLS risks is written as:
$$
d^{e,f}(\beta^e,\beta^f):= \left\| Y^e \right\|_{n^e}^2 - \left\| Y^f \right\|_{n^f}^2 - \sum_{j\in[p]}\beta^e_j \langle X^e_j, Y^e \rangle_{n^e} + \sum_{j\in[p]}\beta^f_j \langle X^f_j, Y^f \rangle_{n^f}.
$$
See Section~\ref{sec:notation} for the notations $\|\cdot\|_{n^e}$ and $\langle \cdot,\cdot\rangle_{n^e}$. 
The constraints \eqref{eq:inv-score-selection} are so-called big-$M$ constraints that link the binary variables $g$ to the environment specific regression coefficients $\{\beta^e\}_{e\in\mathcal{E}}$. For sufficiently large $M_n$, constraints~\eqref{eq:inv-score-selection} enforce $\beta_j^e=0$ whenever $g_j=0$, and are otherwise nonbinding. Thus, $\mathrm{supp}(\beta^e)\subseteq S$.  Similarly, for sufficiently large $\widetilde M_n$, constraints~\eqref{eq:inv-score-normal} enforce
$\langle
X_j^e,,
Y^e-\sum_{k\in[p]}\beta_k^eX_k^e
\rangle_{n^e}=0$ whenever $g_j=1$, and are nonbinding when $g_j=0$. Together with \eqref{eq:inv-score-selection}, these constraints impose the normal equations for the regression of $Y^e$ on $X_S^e$, so that $\beta_S^e$ is the corresponding OLS coefficient. We discuss choices of $M_n$ and $\widetilde M_n$ below. 

Finally, the objective function reformulates the LSCR invariance score as a convex function of $\{\beta^e\}_{e\in\mathcal{E}}$. Indeed, for an OLS coefficient $\beta^e$, we have
\begin{align*}
\begin{split}
    \widehat{\nu}^{e,S} 
    % &= \left\| Y^e - \sum_{j} X^e_j \beta^e_j \right\|_{n^e}^2 
    = \left\| Y^e \right\|_{n^e}^2 - 2 \sum_{j\in[p]}  \beta^e_j\langle X^e_j, Y^e \rangle_{n^e} + \sum_{j,k \in[p]} \beta^e_j \beta^e_k \langle X^e_j, X^e_k \rangle_{n^e}  
    = \left\| Y^e \right\|_{n^e}^2 -   \sum_{j\in[p]} \beta^e_j\langle X^e_j, Y^e \rangle_{n^e}.
\end{split}
\end{align*}
The last equality follows from the OLS normal equations. 

Interestingly, for a fixed subset $S$ (or its associated binary variables $\{g_j\}_{j=1}^p$), the feasible region of \eqref{eq:inv-score} is a singleton, provided that the OLS solution is unique in each environment. Indeed, as discussed in Section~\ref{sec:setup:LSCR}, the support $S$ uniquely determines the corresponding environment specific OLS coefficients $\{\beta^{e,S}\}_{e\in\mathcal{E}}$. Consequently, the minimization over $\{\beta^e\}_{e\in\mathcal{E}}$ may be omitted when evaluating the objective for a fixed $g$. The key benefit of the reformulation \eqref{eq:inv-score} is that our two invariant learning tasks can be cast as the following convex quadratic MIPs:
%minimizing $\widehat{\mathcal{I}}(g)$ over $g\in\{0,1\}^p$, together with the variants introduced in Section~\ref{sec:MIP:tasks}, directly formulates the search for an invariant set as a convex mixed-integer programming problem.

%Moreover, since $g$ and $S$ encode the same subset, the resulting objective value $\widehat{\mathcal{I}}(g)$ coincides exactly with the finite-sample invariance score $\widehat{\mathsf{I}}(S)$.

%Reformulating the invariance score as in \eqref{eq:inv-score} means that optimization problems \eqref{eq:related-task-1} and \eqref{eq:related-task-2} can be cast as the following convex MIPs:

\begin{align}
\text{Task 1:}&\quad \min_{g\in\{0,1\}^p,\, \beta\in\mathbb{R}^p}
~~
\widehat{\mathcal{R}}(\beta)+ \gamma\widehat{\mathcal{I}}(g)+\lambda_n\sum_{j\in[p]}g_j
\quad\text{subject to:}~~
-M_n g_{j} \leq \beta_j \leq M_n g_{j}\quad\forall j\in[p],\label{eqn:task-1-mip}\\
\text{Task 2:}&\quad \min_{g\in\{0,1\}^p}
~~
\widehat{\mathcal{I}}(g)
\quad\text{subject to:}~~
g_{j_0}=0\label{eqn:task-2-mip}.
\end{align}
Note that the big-$M$ constraints in \eqref{eqn:task-1-mip} ensure that the support of $\beta$ is encoded by $g$: whenever $g_j=0$, $\beta_j = 0$. As discussed next, the formulations \eqref{eqn:task-1-mip} and \eqref{eqn:task-2-mip} allow us to certify global optimality, stop early with explicit guarantees, and develop scalable heuristics that produce solutions of high quality.

\vspace{0.05in}
\begin{remark} Optimization problems \eqref{eqn:task-1-mip} and \eqref{eqn:task-2-mip} depend on the parameters $M_n$ and $\widetilde{M}_n$. One can use heuristic approaches \cite{Park2017BayesianNL,xu2025integer} to select them. Specifically, one can first solve \eqref{eqn:task-1-mip} or \eqref{eqn:task-2-mip} after relaxing the binary variables to $[0,1]$, yielding solutions $\widehat{\beta}^e$ and $\widehat{\beta}$, and then set $M_n = 2\max_{e\in\mathcal{E}}\|\widehat{\beta}^e\|_\infty$ and $\widetilde{M}_n = 2\|\widehat{\beta}\|_\infty$. In practice, however, we avoid specifying $M_n$ and $\widetilde{M}_n$ by using Gurobi's indicator constraints. These allow us to directly enforce $\beta_j^e=0$ for all $e\in\mathcal E$ and $\beta_j=0$ whenever $g_j=0$, leaving Gurobi to handle the corresponding mixed-integer formulation internally.
\end{remark}

%minimizing $\widehat{\mathcal{I}}(g)$ over $g\in\{0,1\}^p$, together with the variants introduced in Section~\ref{sec:MIP:tasks}, directly formulates the search for an invariant set as a convex mixed-integer programming problem. 

%\vspace{0.1in}
%\subsection{Related tasks}
%\label{sec:MIP:tasks}

\iffalse
For example, in procedure~\eqref{eqn:task-1-mip}, one may first obtain an optimal solution $\widehat{S}^{(1)}$, corresponding to the most invariant model, and then exclude it from the feasible model class. Resolving the resulting optimization problem yields $\widehat{S}^{(2)}$, the second most invariant model. Repeating this procedure produces a ranked sequence of invariant sets. Specifically, if $g^\circ\in\{0,1\}^p$ encodes a previously identified support, it can be excluded from the feasible region by imposing the constraint
\begin{align}\label{eq:const-top-k}
    \sum_{j: g^\circ_j=1} (1 - g_j) +\sum_{j:g^\circ_j=0} g_j \geq 1,
\end{align}
which requires the new solution to differ from $g^\circ$ in at least one coordinate.
\fi

\subsection{Exact and scalable heuristic solvers}
\label{sec:MIP:algo}

We discuss two approaches for solving the convex MIPs \eqref{eqn:task-1-mip} and~\eqref{eqn:task-2-mip}. 

\vspace{0.1in}
\noindent\textbf{An exact solver.} One option is to use the mixed-integer optimization solver \texttt{Gurobi}, which can certify global optimality up to prescribed numerical tolerances. The solver uses a branch-and-bound framework. For both formulations, relaxing the binary constraints $g_j\in\{0,1\}$ to $g_j\in[0,1]$ yields a convex quadratic program whose optimal value provides a lower bound on the corresponding mixed-integer problem. These lower bounds are used to prune the search tree: if the bound at a node is no smaller than the objective value of the best feasible solution found so far, then that node and all of its descendants can be discarded. Consequently, the solver can certify optimality without enumerating all $2^p$ possible support configurations.

Lower and upper bounds also provide a certificate of solution quality. The best feasible solution, referred to as the \textit{incumbent}, provides an upper bound on the optimal objective value, whereas the convex relaxations provide a global lower bound. Their difference, referred to as the \textit{optimality gap}, quantifies the suboptimality of the incumbent solution. In particular, the incumbent is globally optimal when this gap is zero, up to numerical tolerances.

Several practical strategies can improve the performance of this exact approach. A high-quality warm start can provide a strong initial upper bound and thereby accelerate pruning; such a warm start can be obtained using the coordinate descent procedure described next. One may also configure the solver to prioritize improvements in the global lower bound, for example, by setting \texttt{MIPFocus=2} in \texttt{Gurobi}. Finally, to reduce computation time, an early-stopping tolerance $\tau^{\rm early}\geq 0$ may be imposed on the optimality gap. If the algorithm terminates with a gap no greater than $\tau^{\rm early}$, the returned solution is guaranteed to be within $\tau^{\rm early}$ of the global optimum.

\begin{figure}[!ht]
    \centering
    \includegraphics[width=0.9\textwidth]{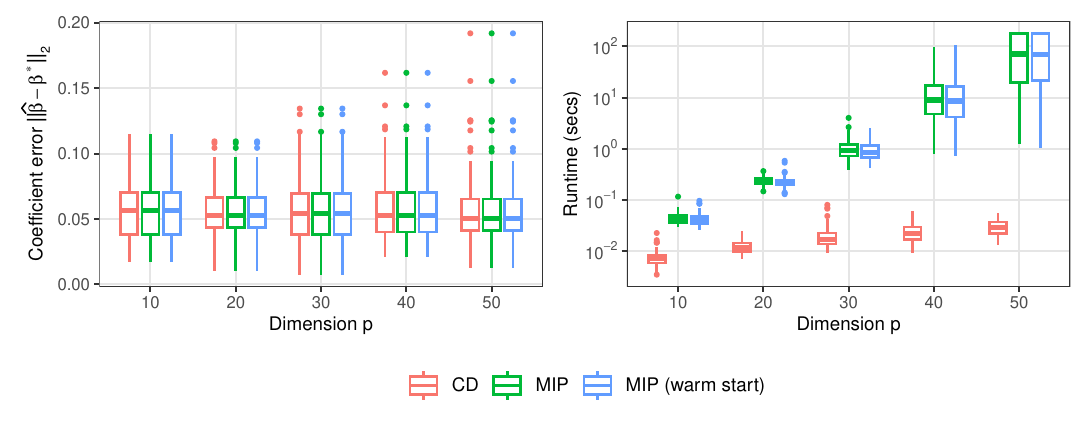}
    \caption{\small Comparison of coefficient estimation error and runtime for Task 1 among heuristic coordinate descent (CD), MIP, and MIP with a CD warm start. Left panel: coefficient estimation error $\|\widehat{\beta}-\beta^\ast\|_2$ across dimensions $p$. Right panel: Runtime across dimensions $p$, shown on a logarithmic scale. Each box summarizes results from 100 independent simulation trials. Implementation details are provided in Appendix~\ref{sec:app:figure-vanilla-cd}. }
    \label{fig:vanilla-cd}
\end{figure}

\vspace{0.1in}
\noindent\textbf{An approximate solver using coordinate descent.}
Coordinate descent provides a substantially more scalable approach for obtaining high-quality feasible solutions. We observe that, for any fixed $g$, the objective values in~\eqref{eqn:task-1-mip} and~\eqref{eqn:task-2-mip} can be evaluated in closed form. This motivates a coordinate descent procedure that initializes from a support $g$ and sequentially updates one coordinate at a time. For each $j\in[p]$, if $g_j = 1$, the algorithm removes $j$ whenever the updated support that sets $g_j = 0$ yields a smaller objective value; otherwise, the support is left unchanged. Similarly, if $g_j=0$, the algorithm adds $j$ whenever the updated support that sets $g_j = 1$ improves the objective. The procedure cycles through the coordinates until no update to a single coordinate decreases the objective. We also employ several practical strategies to accelerate the algorithm, including rank one updates for objective evaluation and local combinatorial optimization. Further details are provided in Appendix~\ref{sec:app:cd-algo}.

Finally, although the coordinate descent procedure is not an exact solver, its solution can be used to initialize the vanilla MIP solver. This provides an initial upper bound of high quality, allowing the MIP solver to further improve the solution and, through its lower bound, either certify global optimality or quantify the remaining optimality gap.

Figure~\ref{fig:vanilla-cd} compares the coordinate descent heuristic (CD) with the exact MIP solver, both with and without a warm start. The MIP solver with a warm start is initialized by the solution obtained from coordinate descent, whereas the vanilla MIP is solved without an initial incumbent. All three methods yield essentially identical estimation performance across the considered dimensions. However, coordinate descent is substantially faster and scales much more favorably with $p$. The computational cost of both MIP variants increases rapidly with $p$, although using the CD solution as a warm start yields a modest runtime improvement over the vanilla MIP.

%\vspace{0.2in}
\section{Task 1: invariance-regularized predictions via convex mixed-integer programming}
\label{sec:task1}
This section studies the estimator \eqref{eqn:task-1-mip} that obtains prediction with invariance regularization. Section~\ref{sec:identifiability} analyzes the estimator in population and presents conditions for which this estimator identifies the causal parents $S^\star$ and the coefficient vector $\beta^\star$. Section~\ref{sec:finite_sample} provides high-dimensional guarantees that allow the number of features to grow with the sample size. Throughout, we set the predictor term $\widehat{\mathcal{R}}(\beta)$ and its population analogue $\mathcal{R}(\beta)$ as:
$$
\widehat{\mathcal{R}}(\beta):= \sum_{e\in\mathcal{E}}{\omega^e}\, \|Y^e-X^e\beta\|_{n^e}^2\quad \text{and}\quad {\mathcal{R}}(\beta):=\sum_{e\in\mathcal{E}}{\omega^e}\, \mathbb{E} \left[ (Y^e-X^e\beta)^2 \right].$$
Thus $\widehat{\mathcal{R}}(\beta)$ is the weighted mean squared error, and $\mathcal{R}(\beta)$ is its corresponding population counterpart. The weights $\omega^e>0$ satisfy $\sum_{e\in\mathcal{E}} \omega^e=1$. In the finite-sample version, we choose $\omega^e= n^e/\sum_{e\in\mathcal{E}}n^e$ for all $e\in\mathcal{E}$.

%the task of identifying an invariant set using the procedure \eqref{eqn:task-1-mip}.

% \vspace{0.1in}
\subsection{Identifiability results in population}
\label{sec:identifiability}
The population analog of estimator \eqref{eqn:task-1-mip} is given by:
\begin{align} \label{eq:pop-task1}
\begin{split}
    (\widetilde{\beta},\, \widetilde{S}) \in \argmin_{\substack{ \beta \in \mathbb{R}^p, S \subseteq[p] } } ~& \sum_{e\in\mathcal{E}} \omega^e\  \mathbb{E} \left[ Y^e - \beta^\top X^e \right]^2 + \gamma \mathsf{I}(S)\quad
    \text{subject to:}~\mathrm{supp}(\beta) \subseteq S.
\end{split}
\end{align}
Here, \eqref{eq:pop-task1} searches over a subset $S$ of covariates and a regression coefficient $\beta$ with support $S$. Its objective is a weighted MSE, computed after pooling data from all environments, plus a regularization term $\mathsf{I}(S)$. Conceptually, the minimization can be decomposed into two stages: first optimizing over the support $S$, and then over the coefficient vector $\beta$. The regularization term encourages the selection of a support with a small invariance score, ideally one that is exactly LSCR-invariant. Once the support solution $\widetilde{S}$ is determined, the subsequent minimization over $\beta$ involves only the MSE term and therefore yields a weighted OLS estimate restricted to $\widetilde{S}$. If $\widetilde{S}\in\mathcal{S}_{\rm LSCR}(\mathcal{E})$, then by \eqref{eq:model-equiv}, the solution $\widetilde{\beta}$ coincides with the environment specific OLS coefficients $\beta^{e,\widetilde{S}}$ for all $e\in\mathcal{E}$. 

When the regularization parameters $\gamma$ and $\rho$ are sufficiently large (which we formalize in Lemma~\ref{lem:identification}), it is apparent that $\widetilde{S}$ must satisfy $\widetilde{S}\in\mathcal{S}_{\rm LSCR}(\mathcal{E})$. In this regime, the solution to \eqref{eq:pop-task1} is equivalent to
$$
\argmin_{\substack{\beta\in\mathbb{R}^p, S\subseteq[p]}} \,
\sum_{e\in\mathcal{E}} \omega^e\  
\mathbb{E} \left[Y^e-\beta^\top X^e\right]^2 \qquad
\text{subject to: }
\mathrm{supp}(\beta) \subseteq S,~\mathsf{I}(S)=0.
$$
Hence, the problem reduces to selecting the most predictive model among all LSCR-invariant models. As alluded to in
Section~\ref{sec:score}, the MSE objective can further distinguish the causal support from other invariant models, particularly when the interventions are not sufficiently heterogeneous. We next present an identification condition under which this distinction is possible, together with the corresponding identification result.

\vspace{0.05in}
\begin{assumption}[Identification]
Let $
G:= \{ j \in[p]: \sum_{e\in\mathcal{E}} \, \omega^e\, \mathbb{E} \left[ X_j^{e} \varepsilon^{e} \right] \neq 0 \}$. For any $S \subseteq [p]$ satisfying $S \cap G \neq \varnothing$, there exists some $e,f \in \mathcal{E}$ such that $
    \beta^{e,S} \neq \beta^{f,S} \text{ or } \nu^{e,S} \neq \nu^{f,S}$. 
\label{ass:identifiability}
\end{assumption}
In essence, this identification condition requires sufficient environmental heterogeneity to distinguish any model containing a variable in $G$. Variables in $G$ are linearly spurious in the sense that they can improve prediction of the noise and hence the prediction risk, despite not belonging to the true causal model.
Under the interventions specified in Assumption~\ref{ass:identifiability}, any model $S$ with $S\cap G\neq\varnothing$ would yield $\mathsf{I}(S)\neq0$, and hence is very likely to be detected.
% In contrast, noncausal models $S$ satisfying $S\cap G=\varnothing$ can be distinguished from the causal model by the MSE risk and therefore require no additional interventions. 

Assumption~\ref{ass:identifiability} is closely related to the identification condition (Condition~4.5) in~\cite{fan2024environment}. Indeed, \cite{fan2024environment} require that for any $S \subseteq [p]$ satisfying $S \cap G \neq \varnothing$, there exists some $e,f \in \mathcal{E}$ such that $\beta^{e,S} \neq \beta^{f,S}$. This similarity arises because both methods adopt an objective of the form MSE plus an invariance penalty, and their penalty essentially sums $\|\beta^{e,S} - \beta^{f,S}\|_2^2$ across $e,f\in\mathcal{E}$. However, the penalty in~\cite{fan2024environment} appears more challenging to formulate as a scalable convex mixed-integer program. 
In addition, our stronger invariance principle, which also exploits heterogeneity in the OLS risks, yields a weaker identification condition and hence greater power. Example~\ref{ex:identification} illustrates this advantage.

Assumption~\ref{ass:identifiability} can also be compared with the sufficient identification conditions in \citet{peters2016causal}. Their Theorem~2 establishes identifiability under several specific intervention schemes, such as single-variable do or noise interventions on every variable other than the target. In contrast, Assumption~\ref{ass:identifiability} does not impose any restrictions on the intervention mechanisms. Moreover, our assumption does not require interventions on variables that are independent of $\varepsilon^e$, such as ancestors of $Y^e$ in a standard SEM. In this sense, Assumption~\ref{ass:identifiability} is less restrictive.

\vspace{0.05in}
\begin{example}
\label{ex:identification}
Consider two environments $e\in\{1,2\}$ and the linear SCM
$$
X_0^e \sim \mathcal N(0,1), \qquad
Y^e = X_0^e + \varepsilon^e, \qquad
X_1^e = a_e Y^e + \eta^e,
$$
where
$\varepsilon^{e}\sim\mathcal N(0,1)$
and $\eta^{e}\sim\mathcal N(0,\upsilon_e)$.
We additionally generate eight independent noise covariates $Z^e_1, \dots, Z^e_8$. All exogenous variables $\varepsilon^e$, $\eta^e$ and $Z^e$ are mutually independent.  
We set $(a_1, a_2) = (1,1/2)$ and $(\upsilon_1, \upsilon_2) = (2, 3/2)$.
The causal support is $S^\star=\{0\}$ while $G=\{1\}$.
For the spurious support $S=\{1\}$, the population OLS coefficient and residual risk are
$
\beta^{e,S}
= 2a_e / (2a_e^2+\upsilon^e)$ and $
\nu^{e,S} = 2 - 4 a_e^2 / ( 2 a_e^2 + \upsilon^e).
$
Under the above parameter choices, we obtain $\beta^{1,S} = \beta^{2,S} = 1/2$, whereas $\nu^{1,S}\neq \nu^{2,S}$. As a result, our identification condition holds but that of~\cite{fan2024environment} is violated. 

For each trial, we generate 100 independent datasets with $n^e=500$ for both environments. We compute, for each method, the proportion of exact support recovery and the proportion of datasets in which the spurious feature $X_1$ is selected. Figure~\ref{fig:identification} shows these proportions averaged over 50 trials. The results demonstrate the limited identification power of \textit{EILLS}~\citep{fan2024environment}, which frequently includes the spurious feature $X_1$. We also report the performance of \textit{ICP}~\citep{peters2016causal}, \textit{Causal Dantzig}~\citep{rothenhausler2019causal}, and \textit{BIP}~\citep{wu2025bayesian}, which also exhibit limited power in this example.

\begin{figure}[!ht]
    \centering
    \includegraphics[width=0.7\textwidth]{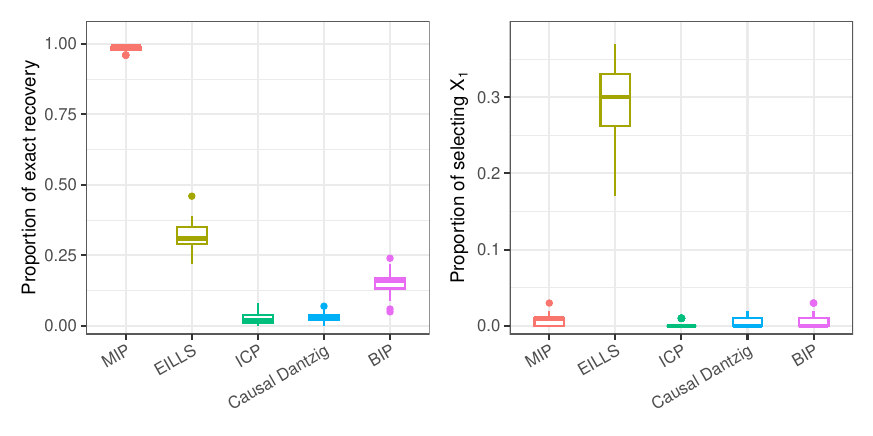}
    \caption{\small Performance in the perturbed version of Example~\ref{ex:identification}. To avoid relying on a parameter configuration that occurs with probability zero, we set
    $(a_1,a_2)=(1,1/2)$ and $(\upsilon_1,\upsilon_2)=(2,3/2)$, each with an additional small Gaussian perturbation. Left panel: proportion of exact support recovery. Right panel: proportion of including the spurious feature $\{1\}$. Each box summarizes results from 50 independent simulation trials. Additional competing methods and full experimental details are provided in Appendix~\ref{sec:app:figure-identification}.}
    \label{fig:identification}
\end{figure}

\end{example}

\vspace{0.05in}
\begin{lemma} \label{lem:identification}
Suppose Assumption~\ref{ass:identifiability} holds. Also assume that there exists $\kappa_L > 0$ such that $\mathbb{E}\left[ X^e X^{e\top} \right] \succeq \kappa_L I_p,\ \forall e\in\mathcal{E} $. Denote
$$
\gamma^\star := \sup_{ S: S \cap G \neq \varnothing } \kappa_L^{-1} \cdot \frac{ \left\| \sum_{e\in\mathcal{E}} \omega^e\ \mathbb{E} [\varepsilon^e X^{e}_S] \right\|_2^2 }{ \mathsf{I}_1(S) + \mathsf{I}_2(S) },
$$
where $\mathsf{I}_1(S) := \sum_{e<f} \omega^{e,f} \|\beta^{e,S} - \beta^{f,S} \|_2^2$ and $
\mathsf{I}_2(S) := \sum_{e<f} \omega^{e,f} (\nu^{e,S} - \nu^{f,S} )^2$. 
When $\gamma \wedge \gamma\rho \geq 2\gamma^\star + \epsilon$ for some $\epsilon >0$, every minimizer of~\eqref{eq:pop-task1} has coefficient vector $\beta^\star$. Moreover, $S^\star$ is the unique smallest selected set among the minimizers.
\end{lemma}

Lemma~\ref{lem:identification} presents the identification result of procedure \eqref{eq:pop-task1}. The proof is provided in Appendix~\ref{sec:app:proof:identification}. Under Assumption~\ref{ass:identifiability}, the score $\mathsf{I}_1(S) + \mathsf{I}_2(S)$ must be positive. Intuitively, for models $S$ satisfying $S\cap G\neq\varnothing$, weaker correlations between $X_S^e$ and $\varepsilon^e$, together with larger invariance scores, reduce the penalty strength required to separate these models from the causal model.

%\vspace{0.1in}
\subsection{Finite-sample guarantees}
\label{sec:finite_sample}
We establish statistical guarantees for estimates obtained from \eqref{eqn:task-1-mip}:
\begin{align*}
\begin{split}
    (\widehat{\beta},\,\widehat{S}) \in \argmin_{\beta\in\mathbb{R}^p,\, S\subseteq[p]} 
    ~& \sum_{e\in \mathcal{E}} \omega^e \, \left\| Y^e - X^e \beta \right\|^2_{n^e} + \gamma \sum_{e< f} \omega^{e,f} \left[ \left\| \widehat\beta^{e,S} - \widehat\beta^{f,S} \right\|_2^2 + \rho\left( \widehat\nu^{e,S} - \widehat\nu^{f,S} \right)^2 \right] + \lambda_n |S|, \\
    \text{subject to:}~
    & \mathrm{supp}(\beta) \subseteq S, \\
    & |S| \leq s_n.
\end{split}
\end{align*}
Here we further impose the sparsity constraint $|S| \leq s_n$, which restricts the number of selected features to at most $s_n$. The required conditions on $s_n$ are specified later. The corresponding detailed MIP formulation is provided in Appendix~\ref{sec:app:mip-details}. %Throughout this section,
%we allow the number of nodes $p$ to vary with the sample size $n$.
To reduce computation time, the procedure may be terminated early once the optimality gap falls below a user specified tolerance $\tau^{\rm early}>0$. We denote the corresponding feasible solution by $\widehat{\beta}^{\rm early}$ and $\widehat{S}^{\rm early}$. 

\vspace{0.05in}
\begin{assumption}[Covariate and noise regularity]
    % There exists some universal constants $0<\kappa_L < \kappa_U < +\infty$ such that $\kappa_L I_p \preceq \Sigma^e \preceq  \kappa_U I_p  $ for all $e\in\mathcal{E}$.
    There exist constants $
    0<\kappa_L \leq \kappa_U<\infty$, 
    $\sigma_x\geq 1$ and $\sigma_\varepsilon>0$
    independent of $n$, $p$, and $e$, such that, for every
    $e\in\mathcal E$, the following conditions hold:
\begin{enumerate}[label = (\arabic*)]
\label{assump:regular}
    \item The population covariance matrix $\Sigma^{e}:=\mathbb{E}[X^e X^{e\top}]$ satisfies $\kappa_L I_p \preceq \Sigma^{e} \preceq \kappa_U I_p$.

    \item The standardized covariate vector is sub-Gaussian: for every
   $v\in\mathbb R^p$,
    $$
    \mathbb E\left[
    \exp\left\{
        v^\top (\Sigma^{e})^{-1/2} X^{e}
    \right\}
    \right]
    \leq
    \exp\left(
        \sigma_x^2 \|v\|_2^2 / 2
    \right).
    $$

    \item The noise variable $\varepsilon^{e}$ is sub-Gaussian: for every $\lambda\in\mathbb{R}$,
    $$
    \mathbb E\left[
        \exp\left\{\lambda\varepsilon^{e}\right\}
    \right]
    \leq
    \exp\left(
        \lambda^2\sigma_\varepsilon^2 / 2
    \right).
    $$
\end{enumerate}
\end{assumption}

\vspace{0.05in}
\begin{assumption}[Model sparsity and moderate magnitude of coefficient] 
\label{assump:sparse-moderate}
The quantity $s_n$ that constrains the number of selected features satisfies $s^\star \leq s_n$. Additionally, the true coefficients satisfy
$$
\|\beta^\star\|_2^2 \lesssim  \frac{(\min_{e\in\mathcal{E}} n^e)^{1 - \delta / 2}}
{ (s_n\log p + \log |\mathcal{E}| )^{1 - \delta / 2} } $$
for some $\delta\in (0,2)$.
\end{assumption}

The conditions in Assumption~\ref{assump:regular} are adapted from Conditions 4.2--4.4 of~\cite{fan2024environment}, which impose uniform covariance and tail bounds across environments. They are very standard in high-dimensional linear regression analysis. From Assumption~\ref{assump:sparse-moderate},
the true model size, upper bounded by $s_n$, may grow with the sample size. The added
constraint $\sum_j g_j \leq s_n$ ensures that the estimated model adheres to the same sparsity level. 
Similarly, the magnitude of coefficients $\|\beta^\star\|_2^2$ is allowed to grow with the sample size, but at a rate not too fast. This condition on $\|\beta^\star\|_2^2$ is required only when $\rho \neq 0$, to control the worst case behavior of the term $(\widehat{\nu}^{e,S} - \widehat{\nu}^{f,S})^2$.
We may assume that all environments have the same scale of samples: $n^e \asymp \min_{f\in\mathcal{E}} n^f, \forall e\in\mathcal{E}$. Then some example choices for these quantities are: $s_n \asymp n_0^{1/5} $, $|\mathcal{E}| \lesssim p$, $\delta = 1/2$, and $\|\beta^\star\|_2^2 \lesssim n_0^{3/5} / (\log p)^{3/4} $.

\vspace{0.05in}
\begin{theorem} \label{thm:convergence}
    Suppose conditions in Assumptions~\ref{ass:identifiability}--\ref{assump:sparse-moderate} and Lemma~\ref{lem:identification} hold. Suppose the sample size in each environment is sufficiently large: $\min_{e\in\mathcal{E}} n^e \gtrsim s_n \log p + \log |\mathcal{E}|$. Then, the following are true with probability at least $1 - 7/p$:
    \begin{enumerate}[label = (\arabic*)]
        \item For the optimal solution $\widehat{\beta}$:
        $$
        \left\| \widehat{\beta} - \beta^\star \right\|_2^2 \lesssim \lambda_n s_n + \frac{s_n \log p + \log |\mathcal{E}|}{ \min_{e\in\mathcal{E}} n^e }.
        $$
        
        \item For the early-stopping solution with threshold $\tau^{\rm early}$:
        $$
        \left\| \widehat{\beta}^{\rm early} - \beta^\star \right\|_2^2 \lesssim \lambda_n s_n + \frac{s_n \log p + \log |\mathcal{E}|}{ \min_{e\in\mathcal{E}} n^e } + \tau^{\rm early}.
        $$
    \end{enumerate}
\end{theorem}

The proof of Theorem~\ref{thm:convergence} is provided in Appendix~\ref{sec:app:proof:convergence}. When $\lambda_n s_n \lesssim (s_n \log p + \log |\mathcal{E}|) / \min_{e\in\mathcal{E}} n^e$, the convergence rate of $\widehat{\beta}$ is dominated by the term $s_n \log p / \min_{e\in\mathcal{E}} n^e$, thereby accommodating high dimensional settings. Moreover, for the early-stopping estimator $\widehat{\beta}^{\rm early}$, if $\tau^{\rm early} \lesssim (s_n \log p + \log |\mathcal{E}|) / \min_{e\in\mathcal{E}} n^e$, then early stopping does not affect its convergence rate.

\vspace{0.05in}
\begin{theorem} \label{thm:consistency}
    Suppose conditions in Assumptions~\ref{ass:identifiability}--\ref{assump:sparse-moderate} and of Lemma~\ref{lem:identification} hold. We also assume that
    \begin{align} \label{eq:beta-min}
        \min_{j\in S^\star} \left| \beta^\star_j \right|^2 \gtrsim \lambda_n + \frac{s_n \log p + \log |\mathcal{E}|}{ \min_{e\in\mathcal{E}} n^e }.
    \end{align}
    Take $\lambda_n \gtrsim (s_n \log p + \log |\mathcal{E}|) / \min_{e\in\mathcal{E}} n^e $. When $\min_{e\in\mathcal{E}} n^e \gtrsim s_n \log p + \log |\mathcal{E}|$,
    the following are true with probability at least $1 - 7/p$:
    \begin{enumerate}[label = (\arabic*)]
        \item The optimal solution model recovers the causal parents: $\widehat{S} = S^\star$.
        \item Whenever the early-stopping threshold satisfies $\tau^{\rm early} \lesssim  (s_n \log p + \log |\mathcal{E}|) / \min_{e\in\mathcal{E}} n^e$, the resulting model recovers the causal parents $\widehat{S}^{\rm early} = S^\star$.
    \end{enumerate}
\end{theorem}
The proof of Theorem~\ref{thm:consistency} is provided in Appendix~\ref{sec:app:proof:consistency}. The argument proceeds by showing that the optimal objective value under the true model $S^\star$ is strictly smaller than that under any incorrect model. This separation is ensured, in part, by the beta-min condition in \eqref{eq:beta-min}, which imposes a lower bound on the signal strength of the true causal predictors so that they can be recovered by the estimated model. It also relies on an appropriate choice of $\lambda_n$, which must be sufficiently large to discourage the inclusion of redundant features. Furthermore, the sub-optimal solution from early stopping can also achieve perfect recovery if the corresponding threshold $\tau^{\rm early}$ is not too large.

%\vspace{0.2in}
\section{Task 2: Testing for causal parents via convex mixed-integer programming}
\label{sec:task2}
Although procedure~\eqref{eqn:task-1-mip} can estimate an invariant set, recovering the true causal parent set $S^\star$ generally requires additional, potentially restrictive identifiability assumptions (e.g., Assumption~\ref{ass:identifiability}). Under substantially weaker requirements, this section studies the task of identifying plausible causal predictors using a procedure of the form~\eqref{eqn:task-2-mip}. In particular, the plausible causal predictor set $S(\mathcal{E})= \bigcap_{S \in\mathcal{S}_{\rm LSCR}(\mathcal{E}) }S$ provides useful information about $S^\star$. 
First, since $S(\mathcal{E})\subseteq S^\star$, it provides a valid lower bound on the causal parent set. At the empirical level, we further show that the estimated plausible causal predictor set, denoted by $\widehat{S}(\mathcal{E})$, satisfies $\widehat{S}(\mathcal{E})\subseteq S^\star$ at a designated significance level. This result provides a form of uncertainty quantification by controlling the inclusion of noncausal variables. Second, every variable in $S(\mathcal{E})$ is a causal parent, and recovering even a subset of causal parents can already be informative and useful for practitioners.

Our uncertainty quantification results are asymptotic, with the sample size growing in each environment. Throughout this section, we assume either that the number of predictors $p$ remains fixed as the sample size increases or that the analysis is restricted to a subset of predictors. For example, procedure~\eqref{eqn:task-1-mip} can be applied to one portion of the data to identify an invariant set. The remaining data can then be used to identify which features in this set are likely to be causal parents.

%Our uncertainty quantification results are asymptotic in terms of a growing sample size in each environment. Throughout the section, we assume that either the number of predictors $p$ is fixed as the sample size increases or that the analysis is done on a subset of the predictors. For example, on a subset of the data, \eqref{eqn:task-1-mip} can be deployed to find an invariant set. On the remaining data, we can identify features in this invariant set that are likely to be causal parents.  

%Recall that the collection of plausible causal predictors is defined as $S(\mathcal{E})= \bigcap_{S \in\mathcal{S}_{\rm LSCR}(\mathcal{E}) } S$. We also denote $n := \min_e n^e$ and $n^\star := \sum_e n^e$.
%In this section, we assume that the samples $\{(X_i^e,Y_i^e)\}_{i=1}^{n^e}$ are independent and identically distributed within each environment $e\in\mathcal{E}$ and are independent across environments. The data-generating distributions follow from the SCM in~\eqref{eqn:scm_main}.

\subsection{A preliminary multiple hypothesis framework} To obtain an estimate $\widehat{S}(\mathcal{E})$, a naive approach, following \citet{peters2016causal}, is to test each candidate set $S$ for LSCR invariance and take the intersection of all sets for which the null hypothesis is not rejected.  This, however, scales exponentially with the number of predictors. Instead, we adopt the approach in \citet{martinet2022variance}, which allows us to estimate $\widehat{S}(\mathcal{E})$ from $p$ hypothesis tests. Specifically, for every $k\in [p]$, we consider the null hypothesis: 
$$
H_{0,k}: k\notin S(\mathcal{E}) \quad \Leftrightarrow \quad \exists S\subseteq[p]\setminus\{k\}  \text{ such that } S \text{ is LSCR invariant}. 
$$
Then, we obtain an estimate $\widehat{{S}}(\mathcal{E})$ via
\begin{align} \label{eq:S-hat-defn}
    \widehat{S}(\mathcal{E}):= \{k\in[p]: H_{0,k} \text{ is rejected} \}= \{k\in[p]: \widehat{T}(k) > \widehat{\tau}_k(\alpha)\}.
\end{align}
Here, $\widehat{T}(k)$ represents a test statistic for the hypothesis $H_{0,k}$ and $\widehat{\tau}_k(\alpha)$ represents the corresponding threshold. For the test to be calibrated, for every $k\in[p]$, we want:
\begin{equation}\mathbb{P}_{H_{0,k}}( \widehat{T}(k) > \widehat\tau_k(\alpha) ) \leq \alpha / p + o(1)\quad \forall k\in[p],
\label{eq:test-cond}
\end{equation}
where the factor of $p$ is due to a Bonferroni correction. Then, the estimated plausible causal predictor set $\widehat{S}(\mathcal{E})$ is unlikely to contain any noncausal variables:
\begin{align*}
\begin{split}
    \mathbb{P} \left( \widehat{S}(\mathcal{E}) \not\subseteq S^\star \right) 
    \leq \mathbb{P}\left( H_{0,k} \text{ is rejected for some noncausal $k$} \right)\leq \sum_{k=1}^p \mathbb{P}_{H_{0,k}} \left( \widehat{T}(k) > \widehat{\tau}_k(\alpha) \right) \leq \alpha+o(1).
\end{split}
\end{align*}
\iffalse
Suppose that, for each hypothesis $H_{0,k}$, we can construct a test statistic $\widehat{T}(k)$ and a corresponding threshold $\tau_k(\alpha)$ such that 
$
\mathbb{P}_{H_{0,k}} ( \widehat{T}(k) \geq \tau_k(\alpha) ) \leq \alpha / p
$.
Then the estimated plausible causal predictor set $\widehat{S}(\mathcal{E})$ is unlikely to contain any noncausal variables:
\fi
\iffalse
To calibrate the test statistics, we propose using a bootstrap procedure to obtain critical values $\widehat{\tau}_k(\alpha)$ for each $k\in[p]$. If these critical values satisfy
\begin{equation}\mathbb{P}_{H_{0,k}}( \widehat{T}(k) > \widehat\tau_k(\alpha) ) \leq \alpha / p + o(1)\quad \forall k\in[p],
\label{eq:test-cond}
\end{equation}
then we have $\limsup_{n\rightarrow\infty} \mathbb{P}(\widehat{S}(\mathcal{E}) \not\subseteq S^\star) \leq \alpha$. 
\fi

In Section~\ref{sec:task2:stat}, we describe how the test statistics $\widehat{T}(k)$ can be obtained by mixed-integer programming. In Section~\ref{sec:task2:bootstrap}, we describe a multiplier bootstrapping procedure to obtain the thresholds $\widehat\tau_k(\alpha)$.

%W
%e provide further details on the construction of this test statistic based on the invariance score in Section~\ref{sec:task2:stat}.

%Our Throughout this section, we assume either that $p$ remains fixed as the sample size increases or that multiple testing is performed on a screened set containing at most $K>0$ features

%Throughout this section, we assume either that $p$ remains fixed as the sample size increases or that multiple testing is performed on a screened set containing at most $K>0$ features. Consequently, the Bonferroni correction factor does not depend on the sample size. 

%\vspace{0.1in}
\subsection{Obtaining test statistics via convex mixed-integer programming}
\label{sec:task2:stat}
By leveraging our invariance score $\widehat{\mathsf{I}}(\cdot)$ in \eqref{eqn:score_pop}
and $\widehat{\mathsf{I}}_+(\cdot)$ in Remark~\ref{rem:extended-score}, we construct two types of test statistics $\widehat{T}(k)$ for the hypothesis $H_{0,k}$. Let $n^\star := \sum_e n^e$. Specifically, we obtain the test statistics via the following optimization problems:
\begin{eqnarray}
\begin{aligned}
  \widehat{T}_\text{MIP-test-2}(k)  &:=  \min_{S\subseteq[p]} 
    \quad & n^\star \widehat{\mathsf{I}}(S)\quad 
\text{subject to:}~~ k\notin S \\
\widehat{T}_\text{MIP-test-4}(k)  &:=  \min_{S\subseteq[p]} 
    \quad & n^\star \widehat{\mathsf{I}}_+(S)\quad 
\text{subject to:}~~ k\notin S
\end{aligned}
\label{eq:test-stat}
\end{eqnarray}
\iffalse
\begin{align} \label{eq:test-stat}
\begin{split}
    \widehat{T}_\text{MIP-test-2}(k)  =  \min_{S\subseteq[p]} 
    \quad & n^\star \widehat{I}(S)\quad 
\text{subject to:}~~ k\notin S.
\end{split}
\end{align}
The corresponding detailed MIP formulation is provided in Appendix~\ref{sec:app:mip-details}.
\fi

With the invariance scores (up to a factor of $n^\star$) as its objective, these optimization problems search for the most invariant model that excludes feature $X_k$. Specifically, when $k\in S(\mathcal{E})$, excluding $X_k$ from the feasible models precludes the existence of an invariant set. Consequently, the optimal objective values $\widehat{T}_\text{MIP-test-2}(k)$ and $\widehat{T}_\text{MIP-test-4}(k)$, which quantify the strongest attainable degrees of invariance under this restriction, are expected to be large.
In contrast, when $k\notin S(\mathcal{E})$, excluding $X_k$ does not preclude the existence of an invariant set $S$, and hence $\widehat{T}_{\text{MIP-test-2}}(k)$ and $\widehat{T}_{\text{MIP-test-4}}(k)$ will be small.  Therefore, the null hypothesis $H_{0,k}$ can be tested by choosing an appropriate threshold $\tau_k$ that separates these two cases. The factor $n^\star$ rescales the invariance score so that the test statistic has a nondegenerate limiting distribution.

As described in Section~\ref{sec:MIP}, both of the optimization problems in \eqref{eq:test-stat} can be reformulated as a convex MIP. See also Section~\ref{sec:app:mip-details} for additional details.
% In Section~\ref{sec:task2:bootstrap}, we develop a bootstrap procedure to approximate this distribution and calibrate the corresponding threshold.

\vspace{0.05in}
\begin{remark}
Computing the test statistics in \eqref{eq:test-stat} requires solving a convex MIP for each $k\in[p]$. To reduce computation, one may use our coordinate descent algorithm to obtain approximate solutions, which often perform well in practice. These solutions can also be used as a warm start of the MIP solver, which can then provide certified optimality gaps and, if desired, solve the problems to global optimality.
\end{remark}

\vspace{0.05in}
\begin{example}
\label{ex:bar}
    We continue with the setup in Example~\ref{ex:teaser}, which contains four covariates $(X_1,X_2,X_3,X_4)$ and a response $Y$. Among these covariates, only $X_1$ and $X_3$ are causal parents of $Y$. We further augment the data with nine independent noncausal variables. Across the three environments, every LSCR-invariant set contains either $\{1,2\}$ or $\{1,3\}$ as a subset. Since both $\{1,2\}$ and $\{1,3\}$ are themselves LSCR-invariant, it follows that $S(\mathcal{E})=\{1\}$.
    Figure~\ref{fig:demo-barplot} presents the results obtained using MIP-test-2. Only feature $X_1$ has a test statistic exceeding the threshold. This is because, after excluding $X_1$, no invariant set exists and therefore even the smallest attainable invariance score is large.
    \begin{figure}[!ht]
    \centering
    \includegraphics[width=0.9\textwidth]{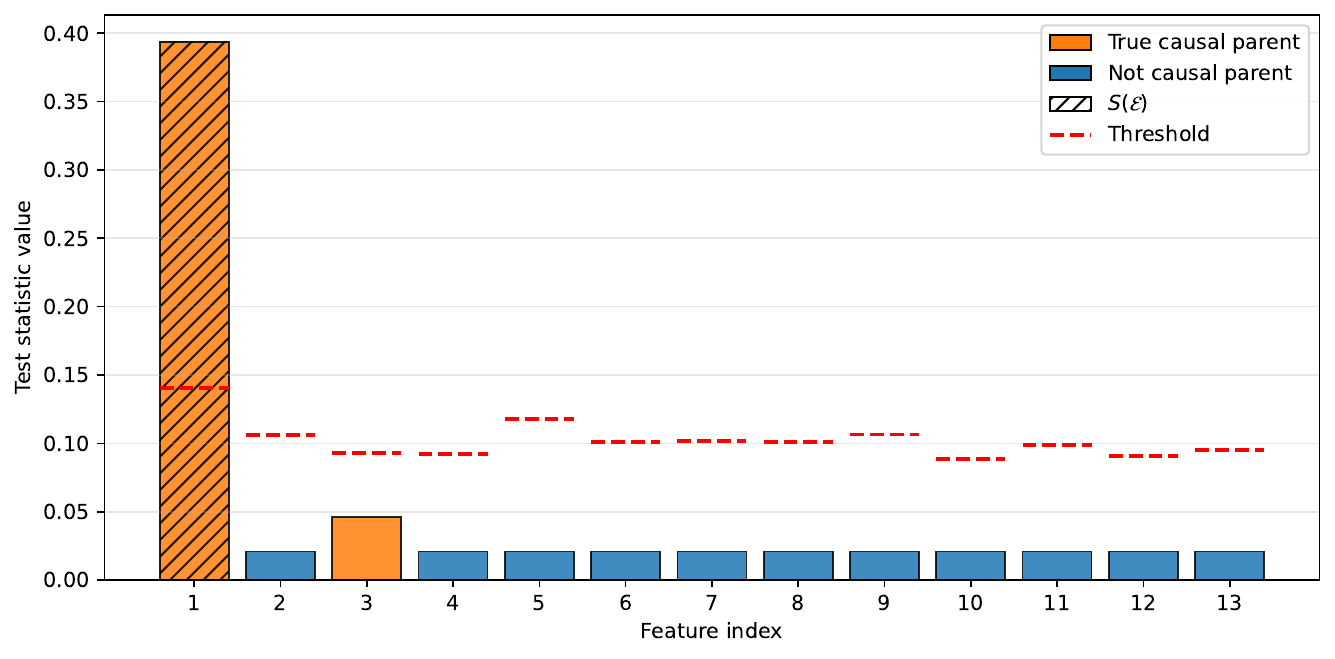}
    \caption{\small Test statistic $\widehat{T}_{\text{MIP-test-2}}(k)$ for every $k$. Orange bars denote the true causal parents, while blue bars denote noncausal variables; the hatched bar indicates the feature in $S(\mathcal{E})$. The red dashed segments show the corresponding bootstrap thresholds that are described in Section~\ref{sec:task2:bootstrap}. Only $X_1$ exceeds its threshold, correctly identifying $S(\mathcal{E})=\{1\}$. Full experimental details are provided in Appendix~\ref{sec:app:figure-teaser}.}
    \label{fig:demo-barplot}
\end{figure}
    
\end{example}

%\vspace{0.1in}
\subsection{A bootstrap procedure for testing}
\label{sec:task2:bootstrap}

Using a multiplier bootstrap procedure, we construct thresholds $\tau_k(\alpha)$ satisfying condition \eqref{eq:test-cond}. Since our approach applies to both test statistics $\widehat{T}_\text{MIP-test-2}(k)$ and $\widehat{T}_\text{MIP-test-4}(k)$, we present it in a general form.

For any $k\not\in S(\mathcal{E})$, let $\mathcal{S}_k$ denote the collection of all LSCR invariant sets excluding feature $X_k$, and let $\widehat{S}$ denote the set encoded by an optimal solution to \eqref{eq:test-stat}. Suppose that $\widehat{S} \in \mathcal{S}_k$. If we set $\tau_k(\alpha)$ to be the $(1-\alpha/p)$-quantile of the distribution of $n^\star \max_{S\in \mathcal{S}_k} \widehat{\mathsf{I}}(S)$, then we have $\mathbb{P}(\widehat{T}(k) > \tau_k(\alpha) ) \leq \alpha/p$. Implementing this strategy requires two key steps. First, because the collection $\mathcal{S}_k$ may be difficult to estimate exactly, we construct an enlarged collection $\mathcal{S}_k^\circ\supseteq\mathcal{S}_k$ such that $\widehat{S}\in\mathcal{S}_k^\circ$ with high probability. Second, we use a multiplier bootstrap procedure to approximate the distribution of $n^\star \max_{S\in \mathcal{S}_k^\circ} \widehat{\mathsf{I}}(S)$.

More specifically, we split the data into two independent subsamples, $\mathcal{D}_1$ and $\mathcal{D}_2$. 
Using $\mathcal D_1$, we construct a candidate collection $\mathcal S_k^\circ$ containing up to $\bar M$ distinct predictor sets. In practice, $\bar M$ can be chosen to provide sufficient diversity among the candidate invariant sets, without requiring the collection to be exhaustive.
One strategy is to generate bootstrap datasets from $\mathcal D_1$ and, for each bootstrap dataset, solve the optimization problem in \eqref{eq:test-stat} to obtain an estimated invariant set. This procedure is repeated until $\bar M$ distinct sets have been collected. We use this strategy in practice for computational efficiency.
Another strategy is to sequentially retrieve sets $S$ in ascending order of the objective value in \eqref{eq:test-stat}, starting with the minimizer, followed by the second best solution, and so forth, until $\bar M$ sets have been retrieved. This sequential procedure is particularly amenable to mixed-integer programming if we use the procedure described in Remark~\ref{rem:top-k}. Details are provided in Appendix~\ref{sec:app:bootstrap}. We also establish a lemma for the second strategy showing that, under a Gaussian model, the resulting collection $\mathcal{S}_k^\circ$ contains $\mathcal{S}_k$ with probability tending to one. Moreover, the condition $\widehat{S}\in\mathcal{S}_k^\circ$ is satisfied as well.

We next obtain a threshold for $\widehat{T}(k)$ using the multiplier bootstrap procedure. Specifically, for each bootstrap replication, the algorithm generates independent Gaussian multipliers and uses them to construct bootstrap analogues of the coefficient and risk components of the invariance score for every $S\in\mathcal{S}_k^\circ$. It then takes the maximum of the resulting bootstrap invariance scores $\widehat{\mathsf I}(S)$ over $S\in\mathcal{S}_k^\circ$, producing one bootstrap realization of the maximal fluctuation. Repeating this procedure yields an empirical approximation to the distribution of $n^\star \max_{S\in\mathcal{S}^\circ_k} \widehat{\mathsf I}(S) $, and the algorithm returns its estimated $(1-\alpha/p)$-quantile as the threshold. The multiplier bootstrap procedure, specifically for the test statistic $\widehat{T}_{\text{MIP-test-2}}(k)$ is described in Algorithm~\ref{algo:bootstrap-threshold}. The extension to $\widehat{T}_{\text{MIP-test-4}}(k)$ is provided in Appendix~\ref{sec:app:gain-more-power}.

\vspace{0.05in}
\begin{remark}
Although our method involves discrete variable selection, the multiplier bootstrap in Algorithm~\ref{algo:bootstrap-threshold} is not used to approximate the distribution of the selected predictor set or the MIP optimizer, for which standard bootstrap methods may be invalid due to nonregularity. Instead, conditional on the candidate collection $\mathcal S_k^\circ$, we use the multiplier bootstrap to approximate the distribution of the largest sampling fluctuation in the invariance statistics across $S\in\mathcal S_k^\circ$. Since $\mathcal S_k^\circ$ is constructed using an independent subsample $\mathcal{D}_1$, it can be treated as fixed when performing inference on $\mathcal D_2$.
\end{remark}

% \at{I think we should add a remark after the algorithm with multiple things:
% 1)  $t_k^b$ can be solved via convex MIP.
% 2) you may say that since it needs to be solved many times you can use coordinate descent and that tends to be accurate.
% 2) even if it is not solved to optimality, we can take add the optimality gap to its value. This ensures that that our guarantees continue to hold?
% }

The following theorem establishes asymptotic type-I error control, ensuring that the estimated set $\widehat{S}(\mathcal{E})$ contains no noncausal variables with probability at least $1-\alpha$. The proof is provided in Appendix~\ref{sec:app:proof:type-1}.

%\widehat{S}(\mathcal{E}) = \{k: \widehat{T}(k)\geq \tau_k(\alpha)\}$$ achieves asymptotic type-I error control. The proof is provided in~\ref{sec:app:proof:type-1}.

% Algorithm ======================
\begin{algorithm}[!ht]
\caption{Bootstrap estimation of the threshold $\widehat{\tau}_k(\alpha)$}
\label{algo:bootstrap-threshold}
\begin{algorithmic}[1]
\Require Multi-environment data 
$\{(X_{\cdot, i}^e,Y_{i}^e)\}_{i=1}^{n^e}$ for $e\in\mathcal{E}$,
collection $\mathcal{S}_k^\circ$, number of bootstrap samples $B$,
and significance level $\alpha$.
\Ensure Bootstrap threshold $\widehat{\tau}_k(\alpha)$.
\For{$b=1,\ldots,B$}
    \State For each $e\in\mathcal{E}$, independently sample
    $
    \xi_1^e,\ldots,\xi_{n^e}^e
    \stackrel{\mathrm{iid}}{\sim}\mathcal{N}(0,1).
    $

    \For{$S\in\mathcal{S}_k^\circ$}

        \State For each $e\in\mathcal{E}$, compute the OLS coefficients
        $\widehat{\beta}^{e,S}$ and 
        $$
        \widehat{r}_i^{e,S} = Y_i^e-(X_{\cdot,i}^{e})^\top\widehat{\beta}^{e,S}, \quad
        \widehat{\Sigma}^e_S = \frac{1}{n^e} \sum_{i=1}^{n^e} X^e_{S,i} X^{e \top}_{S,i}, 
        $$
        $$
        G^{\star,e,S} =
        \frac{1}{\sqrt{n^e}} \sum_{i=1}^{n^e} \xi_i^e X_{S,i}^{e}\widehat{r}_i^{e,S}, \quad
        H^{\star,e,S} =
        \frac{1}{\sqrt{n^e}} \sum_{i=1}^{n^e} \xi_i^e
        \left\{ (\widehat{r}_i^{e,S})^2 - \frac{1}{n^e} \sum_{j=1}^{n^e} (\widehat{r}_j^{e,S})^2
        \right\}.
        $$
        
        \State For each pair $e<f$ in $\mathcal{E}$, compute
        $$
        L_n^{\star,e,f,S} =
        \sqrt{\frac{n^\star}{n^e}} (\widehat{\Sigma}^{e}_S)^{-1}G^{\star,e,S} - \sqrt{\frac{n^\star}{n^f}} (\widehat{\Sigma}^{f}_S)^{-1}G^{\star,f,S},
        $$
        $$
        V_n^{\star,e,f,S} = \sqrt{\frac{n^\star}{n^e}}H^{\star,e,S} - \sqrt{\frac{n^\star}{n^f}}H^{\star,f,S}.
        $$

    \EndFor
    \State Compute
    $$
    t_k^b = \max_{S\in\mathcal{S}_k^\circ} \left\{  \sum_{e<f} \omega^{e,f} \|L_n^{\star,e,f,S}\|_2^2 + \rho \sum_{e<f} \omega^{e,f} |V_n^{\star,e,f,S}|^2 \right\}.
    $$

\EndFor
\State Set
$$
\widehat{\tau}_k(\alpha) = \operatorname{quantile} \left( \{t_k^b\}_{b=1}^B; 1-\alpha/p \right).
$$
\State \Return $\widehat{\tau}_k(\alpha)$.
\end{algorithmic}
\end{algorithm}
% Algorithm ======================

\vspace{0.05in}
\begin{theorem}
\label{thm:type-1}
Suppose that $\widehat{S}\in\mathcal{S}_k$, $\mathcal{S}_k^\circ\supseteq\mathcal{S}_k$, and $n^e/n^\star \rightarrow \rho_e \in(0,1)$ for all $e\in\mathcal{E}$. 
Assume that $\mathbb{E}\, [\|X^e\|_2^4+|Y^e|^4] <\infty$ and $\lambda_{\min}(\Sigma_S^e)>0$.
Assume that the bootstrap sample size $B_n \rightarrow \infty$ as $n \rightarrow \infty$. Then the test statistic $\widehat{T}(k)$ defined in \eqref{eq:test-stat}, together with the bootstrap threshold $\widehat{\tau}_k(\alpha)$ obtained from Algorithm~\ref{algo:bootstrap-threshold}, provides asymptotic type-I error control for testing $H_{0,k}$:
$$
\mathbb{P}_{H_{0,k}}\!\left(
\widehat{T}(k)>\widehat{\tau}_k(\alpha)
\right)
\leq
\frac{\alpha}{p}+o(1)
$$
Consequently, for the estimated plausible causal predictor set defined in \eqref{eq:S-hat-defn}, 
$$
\lim\inf_{n \rightarrow\infty} \mathbb{P}\left( \widehat{S}(\mathcal{E}) \subseteq S^\star \right) \geq 1 - \alpha.
$$
\end{theorem}

% \at{a comment a person may have (for example Ali did in your exam) is that if you have to worry about irregularity as bootstrap is known to not be valid for discrete structures. The key distinction is that we are not using the bootstrap to approximate the distribution of a discrete model selection procedure or the optimizer of the mixed-integer program. Such distributions can be nonregular. We use the multiplier bootstrap to approximate the distribution of the maximum of continuous, asymptotically linear empirical processes indexed by candidate predictor subsets.  }

% Beyond type-I error control, the power of the proposed procedure can be further improved by incorporating two additional invariance conditions for a candidate model $\beta$: (i) $\mathbb{E}[ Y^e - X^e\beta ] =\mathbb{E}[ Y^f - X^f\beta ]$ for all $e,f\in\mathcal{E}$; and (ii) $\mathbb{E} [X^e_j (Y^e - X^e \beta)] = \mathbb{E} [X^f_j (Y^f - X^f \beta)] $ for all $j\in \mathrm{supp}(\beta)$ and $e,f\in\mathcal{E}$. Incorporating the corresponding empirical discrepancies into the invariance score provides additional information for distinguishing non-invariant models and can therefore improve the power to detect plausible causal predictors. The details are provided in Appendix~\ref{sec:app:gain-more-power}. \at{this should have been presented earlier}

\vspace{0.05in}
\begin{example}
\label{ex:task2_results}
    We use the same setup as in Example~\ref{ex:bar}. Figure~\ref{fig:teaser1} compares the performance of four methods for recovering $S(\mathcal{E})=\{1\}$ as the minimum sample size $\min_{e\in\mathcal{E}} n^e$ increases. We consider two approaches using MIP: MIP-test-2, which uses the test statistic in~\eqref{eq:test-stat} based on two components of the invariance score, and MIP-test-4, given in~\eqref{eq:mip-task2-power}, which incorporates two additional components into the test statistic. The MIP methods, particularly MIP-test-4, achieve nearly perfect recovery at relatively small sample sizes while remaining computationally efficient in this toy example. The performance gap between MIP-test-2 and MIP-test-4 demonstrates the benefit of incorporating the two additional components into the test statistic. ICP also achieves nearly perfect recovery when the sample size is sufficiently large because it uses the same invariance principle as MIP-test-2, although its testing procedure appears less powerful. Both MIP approaches and ICP achieve the desired type-I error control, whereas the WVM fails to maintain the target coverage.
\end{example}

\vspace{0.05in}
\begin{figure}[!ht]
\centering
\includegraphics[width=\textwidth]{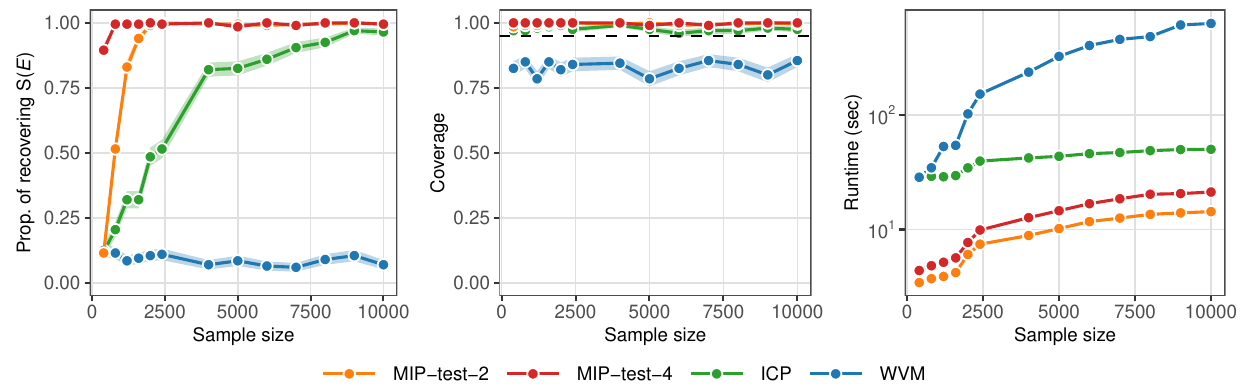}
\caption{\small Performance of Example~\ref{ex:task2_results} as sample size $\min_e n^e$ varies. 
Left panel: proportion of 200 trials in which the true plausible causal predictor set $S(\mathcal E)$ is recovered.
Middle panel: coverage probability, defined as the proportion of 200 trials in which $\widehat S(\mathcal E)\subseteq S^\star$, with the dashed line indicating the target level $0.95$.
Right panel: average runtime over 200 trials, shown on a logarithmic scale. We compare our proposed methods, \textit{MIP-test-2} and \textit{MIP-test-4}, with \textit{ICP}~\citep{peters2016causal} and \textit{WVM}~\citep{martinet2022variance}.
Solid lines show averages and shaded bands indicate $\pm 1$ standard error. 
Additional experimental details are provided in Appendix~\ref{sec:app:figure-teaser}.}
\label{fig:teaser1}
\end{figure}

%\vspace{0.2in}
\section{Experiments}
\label{sec:experiments}

We conduct four complementary simulation studies to investigate different aspects of our proposed framework. Across all experiments, data are generated from linear SCMs following~\eqref{eqn:scm_main}. The true causal parent set $S^\star$ and its coefficient vector $\beta^\star$ remain fixed across environments, so that the causal conditional distribution $Y^e \mid X^e_{S^\star}$ is invariant. Meanwhile, the covariate distributions are perturbed through environment specific interventions. 

We focus on the estimator for Task 1 in Sections~\ref{sec:experiment:stat-behavior}--\ref{sec:experiment:p-baseline} and the testing procedure for Task 2 in Sections~\ref{sec:experiment:teaser2}--\ref{sec:experiment:test-based}. In Section~\ref{sec:experiment:stat-behavior}, we examine the statistical behavior of our invariant estimator~\eqref{eq:related-task-1}, including its dependence on the number of environments and sample size, as well as its robustness to the choice of tuning parameters. In Section~\ref{sec:experiment:p-baseline}, we compare its statistical and computational performance with existing invariant prediction methods. In Section~\ref{sec:experiment:teaser2}, we compare the procedures developed for Tasks~1 and~2 and investigate how their different goals affect empirical performance. Finally, Section~\ref{sec:experiment:test-based} compares the statistical power, coverage, and computational efficiency of different testing procedures as the number of environments increases. Details of data generation mechanisms and implementation are provided in Appendix~\ref{sec:app:experiments}.

\subsection{Behavior of estimator \texorpdfstring{\eqref{eq:related-task-1}}{(5)} }
\label{sec:experiment:stat-behavior}
We generate data from a linear structural equation model with $p=20$ variables and introduce environmental heterogeneity through soft interventions. Across environments, only the linear coefficients used to generate the covariates are perturbed, while the causal mechanism for the response, including its parent set $S^\star$ and structural coefficients $\beta^\star$, remains unchanged. In these experiments, we use the exact MIP solver and solve all trials to global optimality.

Figure~\ref{fig:E_n} shows that both sample size and environmental heterogeneity play important roles in recovery. For a fixed sample size, increasing the number of environments substantially reduces the coefficient estimation error and improves support recovery. With more environments, interventions produce stronger variation and hence more information for distinguishing invariant causal predictors from spurious ones, consistent with our identification theory. In particular, perfect recovery is achieved at considerably smaller sample sizes when more environments are available. For a fixed number of environments, performance also improves steadily as the sample size in each environment increases. The coefficient estimation error decreases and the probability of perfect recovery approaches one, in line with our theory.

\begin{figure}[!ht]
\centering
\includegraphics[width=\textwidth]{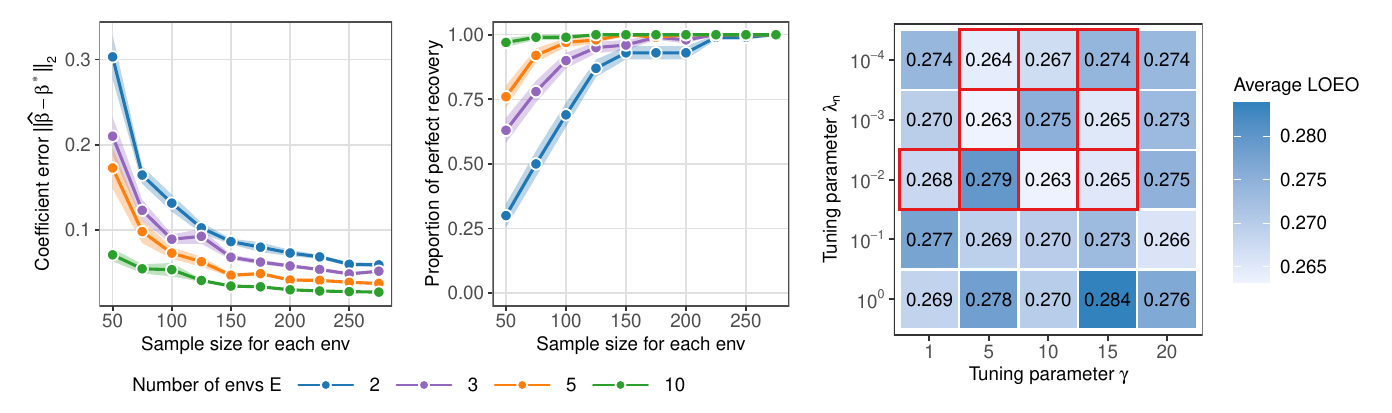}
\caption{\small Statistical behavior and tuning robustness of estimator \eqref{eq:related-task-1}. 
Left panel: average coefficient estimation error $\|\widehat{\beta}-\beta^\star\|_2$ over 100 trials, as the sample size per environment increases, for $|\mathcal{E}|\in\{2,3,5,10\}$ environments. 
Middle panel: proportion of 100 trials achieving perfect recovery of $S^\star$ under the same settings. Solid lines show averages and shaded bands indicate $\pm 1$ standard error. 
Right panel: average leave-one-environment-out (LOEO) criterion over 50 trials, for each combination of tuning parameters $\gamma$ and $\lambda_n$; red outlines indicate tuning combinations with recovery proportion above $0.95$.
Additional experimental details are provided in Appendix~\ref{sec:app:figure-E_n}.}
\label{fig:E_n}
\end{figure}

The right panel examines robustness to the regularization parameters $\gamma$ and $\lambda_n$, with $\rho$ fixed at $1$. Each cell reports the average leave-one-environment-out (LOEO) validation error (see Remark~\ref{remark:tuning parameter}), while the red outlines indicate parameter combinations for which the true causal support $S^\star$ is recovered. Correct recovery occurs over a broad range of $(\gamma,\lambda_n)$ values rather than at a narrowly tuned choice, demonstrating that the proposed estimator is relatively stable with respect to the regularization parameters. Moreover, several of the smallest LOEO validation errors fall within the region of correct recovery, further supporting the effectiveness and robustness of the tuning procedure based on LOEO.

%\vspace{0.1in}
\subsection{Comparison of estimator \texorpdfstring{\eqref{eq:related-task-1}}{(5)} with existing methods}
\label{sec:experiment:p-baseline}
To avoid favoring any particular intervention model, we consider a strong intervention regime. Across environments, the mechanisms generating covariates may differ in their graph structures, structural coefficients, and noise distributions, while the causal conditional distribution $Y^e\mid X_{S^\star}^e$ remains unchanged. We let the number of causal predictors increase with the ambient dimension $p$: we set $|S^\star| = \lfloor 0.4p\rfloor$ when $\lfloor 0.4p\rfloor < 16$, and $|S^\star| = 16 + \lfloor 0.1p\rfloor$ otherwise.
We generate 10 environments with $n^e=500$ observations per environment.
Full details of the data generation process are provided in Appendix~\ref{sec:app:figure-baseline_p}.

\begin{figure}[!ht]
\centering
\includegraphics[width=\textwidth]{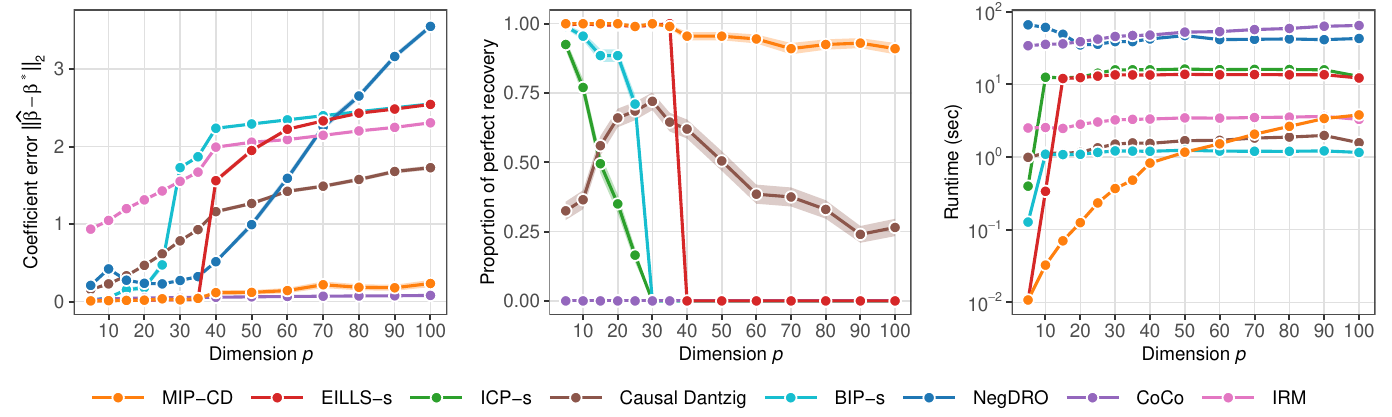}
\caption{\small Performance comparison between our method \textit{MIP-CD} and the baseline methods as the dimension $p$ increases. 
Left panel: average coefficient estimation error $\|\widehat{\beta}-\beta^\star\|_2$ over 200 trials. 
Middle panel: proportion of 200 trials achieving perfect recovery of $S^\star$. 
Right panel: average runtime in seconds over 200 trials, shown on a logarithmic scale. Solid lines show averages and shaded bands indicate $\pm 1$ standard error.
Additional experimental details are provided in Appendix~\ref{sec:app:figure-baseline_p}.}
\label{fig:baseline_p}
\end{figure}

We compare our coordinate descent implementation of (5), \textit{MIP-CD}, with screened \textit{EILLS-s}~\citep{fan2024environment}, screened \textit{ICP-s}~\citep{peters2016causal}, \textit{Causal Dantzig}~\citep{rothenhausler2019causal}, screened \textit{BIP-s}~\citep{wu2025bayesian}, \textit{NegDRO}~\citep{wang2024causal}, \textit{CoCo}~\citep{yin2024optimization}, and \textit{IRM}~\citep{arjovsky2019invariant}. We use \textit{MIP-CD} rather than the exact MIP solver because these experiments involve relatively large dimensions and many simulation trials, for which coordinate descent is computationally more efficient.

Since the exact implementations of \textit{ICP}, \textit{BIP}, and \textit{EILLS} require enumeration over candidate predictor subsets, they do not scale to the larger dimensions considered here. Following the screening strategy used by~\citet{wu2025bayesian}, we first screen the covariates based on the magnitudes of ridge regression coefficients. We keep the top 10 variables for \textit{ICP-s} and \textit{BIP-s} and the top 15 variables for \textit{EILLS-s}. These screening sizes are chosen based on computational feasibility: \textit{ICP-s} with 10 variables and \textit{EILLS-s} with 15 variables have comparable runtimes, both exceeding 10 seconds when screening becomes necessary, while the available implementation of \textit{BIP} is more restrictive. Although approximate implementations are also available, we found that the Gumbel based approximation of \textit{EILLS} rarely selects any variables in our experiments, making the comparison uninformative. Similarly, the exact implementation of \textit{BIP} already exhibits low recovery rates at smaller dimensions. We therefore focus on the screened exact implementations of these methods to provide more informative comparisons. Consequently, once $|S^\star|$ exceeds the screening size, perfect support recovery becomes impossible for these methods, even under ideal screening. In our experiments, this occurs from $p=30$ onward for \textit{ICP-s} and \textit{BIP-s}, and from $p=40$ onward for \textit{EILLS-s}.

% We compare our coordinate descent implementation of \eqref{eq:related-task-1}, \textit{MIP-CD}\at{is there a reason we leave out MIP?}, with screened \textit{EILLS-s}~\citep{fan2024environment}, screened \textit{ICP-s}~\citep{peters2016causal}, \textit{Causal Dantzig}~\citep{rothenhausler2019causal}, screened \textit{BIP-s}~\citep{wu2025bayesian}, \textit{NegDRO}~\citep{wang2024causal}, \textit{CoCo}~\citep{yin2024optimization},
% and \textit{IRM}~\citep{arjovsky2019invariant}. 
% Because the exact implementations of \textit{ICP}, \textit{BIP}, and \textit{EILLS} require enumeration over candidate predictor subsets, they do not scale to the larger dimensions considered here. We therefore first screen the covariates using the magnitudes of ridge regression coefficients \at{is that what is recommended by them?}. We retain the top 10 variables \at{why not larger? enumeration is possible up to 15? Then it makes sense to do $15$ for all of them? Also, doesn't EILLS and BIP have approximation procedures for larger settings? Since you run our approximation method, why not run theirs?} for \textit{ICP-s} and \textit{BIP-s} and the top 15 variables for \textit{EILLS-s}, and then apply the corresponding methods to the screened sets. Consequently, once $|S^\star|$ exceeds the screening size, perfect support recovery is impossible for these methods, even under ideal screening. For the dimensions considered here, this occurs from $p=30$ onward for \textit{ICP-s} and \textit{BIP-s} and from $p=40$ onward for \textit{EILLS-s}.

The remaining baselines involve different compromises. The methods \textit{CoCo}, \textit{NegDRO}, and \textit{IRM} primarily estimate a coefficient vector rather than performing explicit variable selection. Moreover, \textit{Causal Dantzig} and \textit{NegDRO} are developed under additive intervention mechanisms, where interventions are restricted to the noise terms. The comparison therefore evaluates how these methods perform under a more general intervention setting that is not tailored to their respective assumptions.

The right panel compares computational cost. Despite solving a discrete combinatorial optimization problem, our \textit{MIP-CD} remains highly competitive computationally and is the fastest method for $p<50$ in this experiment. Its runtime increases gradually with the dimension and remains comparable to most competing methods at larger $p$. These results demonstrate the scalability of the proposed MIP approach and its computational advantage over explicit subset enumeration.

%\vspace{0.1in}
\subsection{Recovering \texorpdfstring{$S(\mathcal E)$}{S(E)} versus an invariant predictor set}
\label{sec:experiment:teaser2}
We compare methods designed for the two distinct tasks considered in this paper. The first group, including our proposed \textit{MIP-test} in Section~\ref{sec:task2}, \textit{ICP}, and \textit{WVM}, aims to recover the plausible causal predictor set $S(\mathcal E)$ with controlled false discoveries. The second group, including \textit{MIP}, \textit{EILLS} and \textit{BIP-s}, instead aims to identify an invariant set. These two objectives need not coincide when the causal parent set is not identifiable.

\begin{figure}[!ht]
\centering
\includegraphics[width=\textwidth]{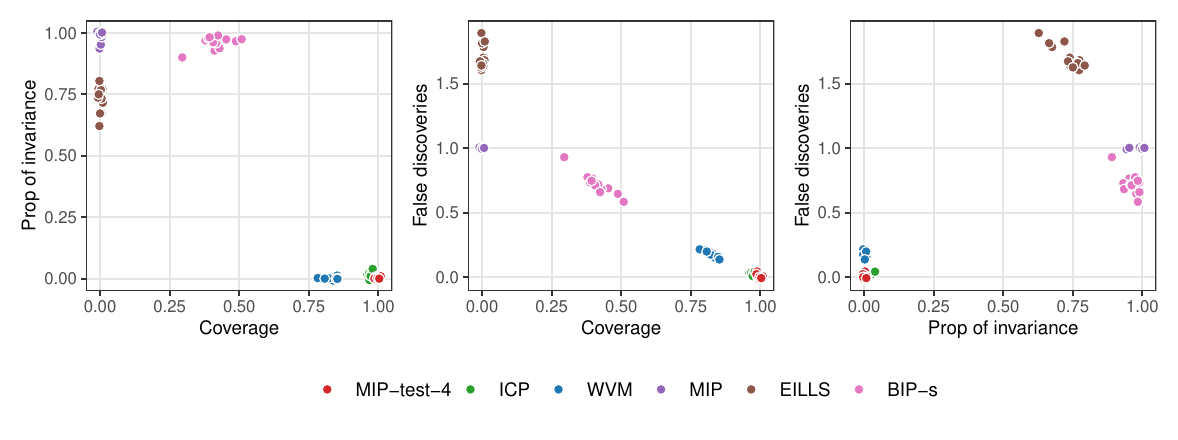}
\caption{\small Trade-offs between coverage, invariance, and false discoveries across methods as the sample size per environment varies. Each point corresponds to one sample size $n^e$ ranging from $400$ to $10000$. 
Left panel: proportion of 200 trials in which the estimated set is LSCR invariant, versus coverage probability, defined as the proportion of 200 trials in which $\widehat{S} \subseteq S^\star$. 
Middle panel: average number of false discoveries $|\widehat{S} \setminus S^\star|$ over 200 trials, versus coverage probability. 
Right panel: average number of false discoveries, versus the proportion of trials in which the estimated set is invariant. 
Additional experimental details are provided in Appendix~\ref{sec:app:figure-teaser}. }
\label{fig:teaser2}
\end{figure}

Figure~\ref{fig:teaser2} shows the trade-off between the two groups of methods using the setup of Example~\ref{ex:bar}. Methods targeting $S(\mathcal E)$, including \textit{MIP-test-4}, \textit{ICP}, and \textit{WVM}, generally achieve high coverage and few false discoveries, but the selected sets are rarely themselves invariant. In contrast, methods targeting an invariant set, including \textit{MIP}, \textit{EILLS}, and \textit{BIP-s}, recover invariant models much more frequently, but tend to have lower coverage and more false discoveries.

This behavior follows from the structure of the example. Here, $S(\mathcal E)=\{1\}$, and multiple invariant models exist: every invariant set contains either $\{1,2\}$ or $\{1,3\}$, whereas their intersection is only $\{1\}$. Importantly, $S(\mathcal E)$ itself is not invariant. Therefore, methods targeting $S(\mathcal E)$ are naturally conservative: they identify predictors shared across all invariant models, leading to few false discoveries, but the resulting set need not itself define an invariant predictive model. By contrast, methods targeting an invariant set seek a richer model satisfying the invariance criterion; because such a set is not unique, however, the selected model need not coincide with the true causal parent set $S^\star$.
% \at{Is this just an analysis on of the earlier examples in the paper? It seems that we are missing an extensisve evaluation of the testing portion of our procedure where we try more complicated setups etc.}

\subsection{Comparison of testing procedures}
\label{sec:experiment:test-based}

We further compare the statistical and computational performance of testing procedures for Task 2, under the intervention regime used in Section~\ref{sec:experiment:p-baseline}. We generate data from a linear structural equation model with $p=10$ covariates and $n^e=500$ observations per environment. In each environment, we generate a directed acyclic graph (DAG) by randomly selecting edges between pairs of nodes according to a common topological ordering, while keeping the edges pointing to the response fixed. The coefficients and noise variances of the covariates also vary across environments, whereas the causal conditional distribution $Y^e \mid X^e_{S^\star}$ remains invariant. Compared with the setting in Section~\ref{sec:experiment:teaser2}, the DAG structure and causal parent set are independently generated in each trial. Increasing the number of environments also provides greater heterogeneity across environments and improves causal identification. We compare our proposed methods, \textit{MIP-test-2} and \textit{MIP-test-4}, with \textit{ICP} and \textit{WVM}.

\begin{figure}[!ht]
\centering
\includegraphics[width=0.95\textwidth]{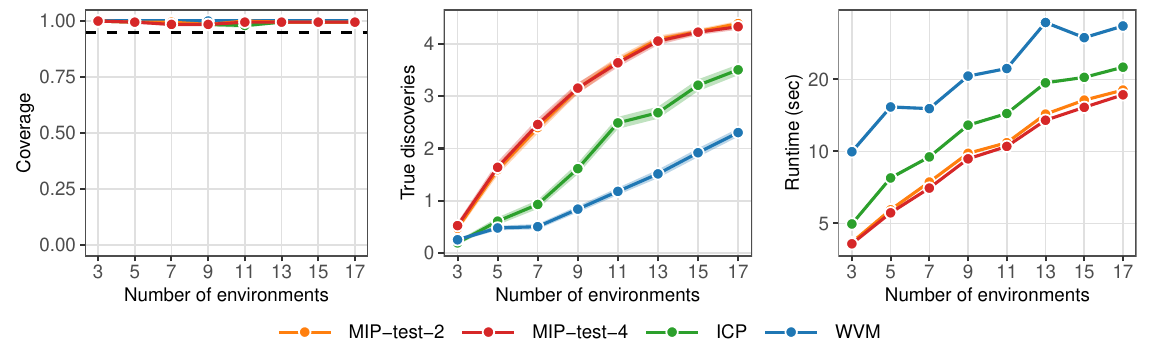}
\caption{\small Performance comparison of testing procedures for Task 2 as the number of environments increases. 
Left panel: coverage probability, defined as the proportion of 200 trials in which $\widehat{\mathcal S}(\mathcal E)\subseteq S^\star$, with the dashed line indicating the target level 0.95. 
Middle panel: average number of true discoveries $|\widehat{S}(\mathcal{E}) \cap S^\star|$ over 200 trials; the number of true causal parents $S^\star$ is randomly selected from $\{4,5\}$. 
Right panel: average runtime in seconds over 200 trials, shown on a logarithmic scale. Solid lines show averages, and shaded bands indicate $\pm1$ standard error.
Additional experimental details are provided in Appendix~\ref{sec:app:test-based}. }
\label{fig:test-based}
\end{figure}

Figure~\ref{fig:test-based} presents the results as the number of environments increases. All four methods maintain coverage above the target level of 0.95, indicating effective control of false causal discoveries in this experiment. Meanwhile, the number of true discoveries increases with the number of environments, as additional environments provide more information for distinguishing causal predictors from spurious variables. Our proposed \textit{MIP-test-2} and \textit{MIP-test-4} consistently identify more true causal predictors than \textit{ICP} and \textit{WVM}, demonstrating greater statistical power. Moreover, both MIP procedures are substantially faster than \textit{WVM} and remain competitive with \textit{ICP}. 
% Overall, these results demonstrate the statistical power and computational efficiency of our testing framework.

\section{Real data: Causal chambers}
\label{sec:realdata}

Causal chambers \citep{gamella2025causal} are devices controlled by computers that allow for manipulating and measuring variables in physical systems, providing a rich testbed for causal discovery algorithms. Following the experimental setup of \citet{gu2025causality}, we consider the light tunnel chamber, which consists of a controllable light source, two rotating polarizers, and sensors measuring light intensity. We aim to recover the causal parents of the infrared intensity measurement $Y=\widetilde{I}_3$ among 11 observed covariates, with the known causal parent set $S^\star = \{R, G, B, \theta_1, \theta_2\}$. In particular, the causal effects of the polarizer angles ($\theta_1$ and $\theta_2$) on light intensity $\widetilde{I}_3$ are nonlinear.

We use the reference environment and 19 intervention environments corresponding to different manipulations of the physical system. To investigate how the number of environments affects causal variable selection, we vary $|\mathcal E|\in\{3,4,\ldots,10\}$, with 1000 observations per environment. For each trial, we randomly permute the intervention environments and construct nested environment sets by sequentially adding interventions to the reference environment. We repeat this procedure for 100 randomized environment configurations and compare our proposed methods, \textit{MIP} and \textit{MIP-test-4}, with \textit{EILLS}, \textit{ICP}, \textit{Causal Dantzig}, and \textit{BIP}. Additionally, we compare with \textit{FAIR}~\citep{gu2025causality}, a method based on neural networks for invariant causal learning. This method is relevant here because of the nonlinear causal effects of $(\theta_1,\theta_2)$ on $\widetilde{I}_3$. Additional experimental details are provided in Appendix~\ref{sec:app:causal-chamber}.

\begin{figure}[!ht]
\centering
\includegraphics[width=\textwidth]{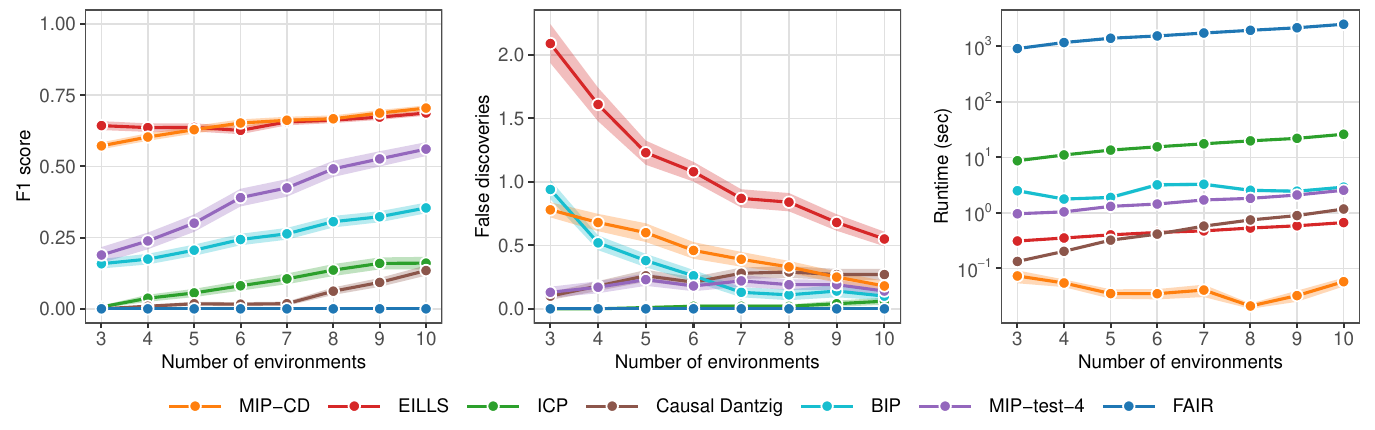}
\caption{\small Performance comparison of invariant causal variable selection methods on the Causal Chamber dataset as the number of environments increases. 
Left panel: average F1 score over 100 trials, for recovering the true causal parent set $S^\star$. 
Middle panel: average number of false discoveries $|\widehat S\setminus S^\star|$ over 100 trials.
Right panel: average runtime in seconds over 100 trials, shown on a logarithmic scale. Solid lines show averages, and shaded bands indicate $\pm 1$ standard error.
Additional experimental details are provided in Appendix~\ref{sec:app:causal-chamber}.}
\label{fig:causal_chamber}
\end{figure}

Figure~\ref{fig:causal_chamber} presents the F1 score, number of false discoveries, and runtime as the number of environments increases. \textit{MIP-CD} achieves the highest F1 score as more environments become available. In particular, compared with \textit{EILLS}, \textit{MIP-CD} achieves a comparable or higher F1 score while producing substantially fewer false discoveries. Meanwhile, \textit{MIP-test} exhibits more conservative variable selection, achieving lower F1 scores than \textit{MIP-CD} and \textit{EILLS}. Even so, it consistently outperforms all other competing methods in terms of F1 score, and produces fewer false discoveries than \textit{MIP-CD} and \textit{EILLS}. Notably, \textit{FAIR} fails to recover any causal predictors in this experiment, possibly due to optimization difficulties associated with its Gumbel procedure. Even with 10,000 optimization iterations and runtimes exceeding 1,000 seconds, the computational effort remains insufficient for \textit{FAIR}. In contrast, both \textit{MIP-CD} and \textit{MIP-test} are substantially more computationally efficient. Overall, these results demonstrate the effectiveness and computational efficiency of our proposed methods in identifying causal predictors from heterogeneous environments in a real physical system.

\section{Extensions to nonlinearity and hidden confounding}
\label{sec:extensions}
\vspace{-0.1in}
\subsection{Extension to nonlinear causal effects}
Our framework can also be extended to settings with nonlinear causal effects. Specifically, consider an additive structural equation model
$$
Y^e=\sum_{j\in S^\star}f_j^\star(X_j^e)+\varepsilon^e,
\qquad
\varepsilon^e\perp X_{S^\star}^e,\quad
\varepsilon^e\sim\mathcal F,\quad
\mathbb E[\varepsilon^e]=0,
$$
where $f_j^\star$ denotes the nonlinear causal effect of $X_j$ on $Y$. To accommodate nonlinear effects, we approximate each $f_j^\star$ using $R_n$ basis functions:
$$
f_j^\star(X_j^e)\approx\sum_{r=1}^{R_n}\theta_{jr}b_r(X_j^e),
$$
where the coefficients $\theta_j=(\theta_{j1},\ldots,\theta_{jR_n})^\top$ are treated as a group associated with each original predictor $X_j$. We can then apply the LSCR invariance principle to the expanded covariates, together with a group $\ell_0$-type penalty. The resulting optimization problem can still be formulated as a convex MIP.
One limitation is that the LSCR invariance principle may not hold exactly under finite basis expansions. In particular, the population OLS coefficients and residual risks may vary across environments due to approximation errors, even when the underlying nonlinear causal mechanism remains invariant. We therefore consider this extension as an empirical investigation and leave theoretical guarantees to future work. Further details on the formulation and numerical results are provided in Appendix~\ref{sec:app:nonlinear}.

\subsection{Extension to hidden confounding}
\label{sec:confounding}
We next consider settings with hidden confounding, which require a different invariance principle. 
Specifically, we consider a linear structural equation model in which exogenous latent variables may affect both the covariates and the response, while interventions change the noise variances of the observed covariates without altering the structural coefficients or the response mechanism.
In this setting, the LSCR principle may no longer hold, since the population OLS coefficients can vary across environments even when the underlying causal coefficients remain invariant.
Nevertheless, under the noise intervention model, we have
$$
\mathbb{E}\left[X_j^e\left(Y^e-X^{e\top}\beta^\star\right)\right]
=
\mathbb{E}\left[X_j^f\left(Y^f-X^{f\top}\beta^\star\right)\right],
\quad \forall j\in[p], \ e,f\in\mathcal{E}.
$$
This motivates an alternative invariance principle based on residual moments, which can accommodate hidden confounding and is also considered in~\citet{rothenhausler2019causal}.
Based on this principle, we define an invariance score that measures discrepancies in these residual moments across environments:
$$
\mathsf{I}_h(\beta)
=
\sum_{e<f}\omega^{e,f}\sum_{j=1}^p
\left( \mathbb{E}\left[X_j^e\left(Y^e - X^{e\top} \beta\right)\right] - \mathbb{E}\left[X_j^f\left(Y^f - X^{f\top} \beta\right)\right] \right)^2.
$$
By construction, $\mathsf{I}_h(\beta)=0$ if and only if the cross moments between the residuals and covariates remain invariant across environments.
We then combine its empirical counterpart with the empirical MSE loss and a sparsity penalty, leading to a convex MIP formulation.
Further details on the model, optimization formulation, and numerical experiments are provided in Appendix~\ref{sec:app:confounding}.

\section{Discussion}
\label{sec:discussion}

We propose a unified MIP framework for invariant causal variable selection and inference across heterogeneous environments. Based on the LSCR invariance principle, our framework addresses both causal parent recovery under identifiability and causal discovery with type-I error control when full identification is not possible. Our methods enjoy both statistical and optimization guarantees, and experiments on synthetic and real data demonstrate their effectiveness.

Our work opens several directions for future research. First, although our framework can be extended to nonlinear additive models through basis expansions, approximation errors may prevent the LSCR invariance principle from holding exactly. Developing more general methods for nonlinear causal effects while preserving computational tractability and optimization guarantees would be valuable. Second, our identification results depend on the heterogeneity of the available environments. Investigating how to design informative interventions or select environments to improve causal identification is an important direction. Finally, although our MIP formulations avoid explicit subset enumeration, solving larger problems to optimality can still be computationally challenging. Developing specialized branch-and-bound algorithms that further exploit the invariance score structure may substantially improve computational efficiency.

\section*{Acknowledgments}
We acknowledge funding from NSF grant DMS-2413074 (Xiaozhu Zhang and Armeen Taeb).
Rahul Mazumder acknowledges funding from Office of Naval Research (N000142512504). 

%\at{references still not consistent. some journals capitalized others are not.}
%\at{make sure references are consistent}

\vspace{0.2in}
\bibliographystyle{agsm}
\bibliography{refs}

@article{fan2024environment,
  title={Environment invariant linear least squares},
  author={Fan, Jianqing and Fang, Cong and Gu, Yihong and Zhang, Tong},
  journal={The Annals of Statistics},
  volume={52},
  number={5},
  pages={2268--2292},
  year={2024},
  publisher={Institute of Mathematical Statistics}
}

@article{peters2016causal,
  title={Causal inference by using invariant prediction: identification and confidence intervals},
  author={Peters, Jonas and B{\"u}hlmann, Peter and Meinshausen, Nicolai},
  journal={Journal of the Royal Statistical Society Series B: Statistical Methodology},
  volume={78},
  number={5},
  pages={947--1012},
  year={2016},
  publisher={Oxford University Press}
}

@article{wang2024causal,
  title={Causal Invariance Learning via Efficient Nonconvex Optimization},
  author={Wang, Zhenyu and Hu, Yifan and B{\"u}hlmann, Peter and Guo, Zijian},
  journal={arXiv preprint arXiv:2412.11850},
  year={2024}
}

@inproceedings{martinet2022variance,
  title={Variance minimization in the wasserstein space for invariant causal prediction},
  author={Martinet, Guillaume G and Strzalkowski, Alexander and Engelhardt, Barbara},
  booktitle={International Conference on Artificial Intelligence and Statistics},
  pages={8803--8851},
  year={2022},
  organization={PMLR}
}

@article{rothenhausler2019causal,
  title={Causal dantzig},
  author={Rothenh{\"a}usler, Dominik and B{\"u}hlmann, Peter and Meinshausen, Nicolai},
  journal={The Annals of Statistics},
  volume={47},
  number={3},
  pages={1688--1722},
  year={2019},
  publisher={JSTOR}
}

@inproceedings{Kamath2021DoesIR,
  title={Does Invariant Risk Minimization Capture Invariance?},
  author={Pritish Kamath and Akilesh Tangella and Danica J. Sutherland and Nathan Srebro},
  booktitle={International Conference on Artificial Intelligence and Statistics},
  year={2021},
}

@article{xu2025integer,
  title={Integer programming for learning directed acyclic graphs from non-identifiable Gaussian models},
  author={Xu, Tong and Taeb, Armeen and K{\"u}{\c{c}}{\"u}kyavuz, Simge and Shojaie, Ali},
  journal={Biometrika},
  pages={asaf032},
  year={2025},
  publisher={Oxford University Press}
}

@article{manzour2021integer,
  title={Integer programming for learning directed acyclic graphs from continuous data},
  author={Manzour, Hasan and K{\"u}{\c{c}}{\"u}kyavuz, Simge and Wu, Hao-Hsiang and Shojaie, Ali},
  journal={INFORMS Journal on Optimization},
  volume={3},
  number={1},
  pages={46--73},
  year={2021},
  publisher={INFORMS}
}

@article{Zhang2025ConvexMP,
  title={Convex Mixed-Integer Programming for Causal Additive Models with Optimization and Statistical Guarantees},
  author={Xiaozhu Zhang and Nir Keret and Ali Shojaie and Armeen Taeb},
  year={2025},
  journal={arxiv preprint arXiv:2511.21126}
}

@article{kucukyavuz2023consistent,
  title={Consistent second-order conic integer programming for learning Bayesian networks},
  author={K{\"u}{\c{c}}{\"u}kyavuz, Simge and Shojaie, Ali and Manzour, Hasan and Wei, Linchuan and Wu, Hao-Hsiang},
  journal={Journal of Machine Learning Research},
  volume={24},
  number={322},
  pages={1--38},
  year={2023}
}

@article{Park2017BayesianNL,
  title={Bayesian Network Learning via Topological Order},
  author={Young Woong Park and Diego Klabjan},
  journal={Journal of Machine Learning Research},
  year={2017},
  volume={18},
pages = {1-30}
}

@article{shen2026causality,
  title={Causality-oriented robustness: Exploiting general noise interventions},
  author={Shen, Xinwei and B{\"u}hlmann, Peter and Taeb, Armeen},
  journal={Journal of the American Statistical Association},
  volume={121},
  number={553},
  pages={704--715},
  year={2026},
  publisher={Taylor \& Francis}
}

@article{behdin2026multi,
  title={Multi-task learning for sparsity pattern heterogeneity: statistical and computational perspectives},
  author={Behdin, Kayhan and Loewinger, Gabriel and Kishida, Kenneth T and Parmigiani, Giovanni and Mazumder, Rahul},
  journal={Journal of the Royal Statistical Society Series B: Statistical Methodology},
  volume={88},
  number={3},
  pages={1029--1059},
  year={2026},
  publisher={Oxford University Press UK}
}

@article{hazimeh2020fast,
  title={Fast best subset selection: Coordinate descent and local combinatorial optimization algorithms},
  author={Hazimeh, Hussein and Mazumder, Rahul},
  journal={Operations Research},
  volume={68},
  number={5},
  pages={1517--1537},
  year={2020},
  publisher={INFORMS}
}

@article{wu2025bayesian,
  title={Bayesian invariance modeling of multi-environment data},
  author={Wu, Luhuan and Yin, Mingzhang and Wang, Yixin and Cunningham, John P and Blei, David M},
  journal={arXiv preprint arXiv:2506.22675},
  year={2025}
}

@article{yin2024optimization,
  title={Optimization-based causal estimation from heterogeneous environments},
  author={Yin, Mingzhang and Wang, Yixin and Blei, David M},
  journal={Journal of Machine Learning Research},
  volume={25},
  number={168},
  pages={1--44},
  year={2024}
}

@article{arjovsky2019invariant,
  title={Invariant risk minimization},
  author={Arjovsky, Martin and Bottou, L{\'e}on and Gulrajani, Ishaan and Lopez-Paz, David},
  journal={arXiv preprint arXiv:1907.02893},
  year={2019}
}

@article{rojas2018invariant,
  title={Invariant models for causal transfer learning},
  author={Rojas-Carulla, Mateo and Sch{\"o}lkopf, Bernhard and Turner, Richard and Peters, Jonas},
  journal={Journal of Machine Learning Research},
  volume={19},
  number={36},
  pages={1--34},
  year={2018}
}

@article{pfister2021stabilizing,
  title={Stabilizing variable selection and regression},
  author={Pfister, Niklas and Williams, Evan G and Peters, Jonas and Aebersold, Ruedi and B{\"u}hlmann, Peter},
  journal={The Annals of Applied Statistics},
  volume={15},
  number={3},
  pages={1220--1246},
  year={2021},
  publisher={JSTOR}
}

@article{rothenhausler2021anchor,
  title={Anchor regression: Heterogeneous data meet causality},
  author={Rothenh{\"a}usler, Dominik and Meinshausen, Nicolai and B{\"u}hlmann, Peter and Peters, Jonas},
  journal={Journal of the Royal Statistical Society: Series B (Statistical Methodology)},
  volume={83},
  number={2},
  pages={215--246},
  year={2021},
  publisher={Wiley Online Library}
}

@article{henzi2025invariant,
  title={Invariant probabilistic prediction},
  author={Henzi, Alexander and Shen, Xinwei and Law, Michael and B{\"u}hlmann, Peter},
  journal={Biometrika},
  volume={112},
  number={1},
  pages={asae063},
  year={2025},
  publisher={Oxford University Press}
}

@article{heinze2018invariant,
  title={Invariant causal prediction for nonlinear models},
  author={Heinze-Deml, Christina and Peters, Jonas and Meinshausen, Nicolai},
  journal={Journal of Causal Inference},
  volume={6},
  number={2},
  pages={20170016},
  year={2018},
  publisher={De Gruyter}
}

@article{gu2025causality,
  title={Causality pursuit from heterogeneous environments via neural adversarial invariance learning},
  author={Gu, Yihong and Fang, Cong and B{\"u}hlmann, Peter and Fan, Jianqing},
  journal={The Annals of Statistics},
  volume={53},
  number={5},
  pages={2230--2257},
  year={2025},
  publisher={Institute of Mathematical Statistics}
}

@article{meinshausen2015maximin,
 author = {Nicolai Meinshausen and Peter Bühlmann},
 journal = {The Annals of Statistics},
 number = {4},
 pages = {1801--1830},
 publisher = {Institute of Mathematical Statistics},
 title = {MAXIMIN EFFECTS IN INHOMOGENEOUS LARGE-SCALE DATA},
 urldate = {2026-10-08},
 volume = {43},
 year = {2015}
}

@article{kook2025model,
  title={Model-based causal feature selection for general response types},
  author={Kook, Lucas and Saengkyongam, Sorawit and Lundborg, Anton Rask and Hothorn, Torsten and Peters, Jonas},
  journal={Journal of the American Statistical Association},
  volume={120},
  number={550},
  pages={1090--1101},
  year={2025},
  publisher={Taylor \& Francis}
}

@article{bertsimas2016best,
  title={Best subset selection via a modern optimization lens},
  author={Bertsimas, Dimitris and King, Angela and Mazumder, Rahul},
  journal={The annals of statistics},
  pages={813--852},
  year={2016},
  publisher={JSTOR}
}

@article{bertsimas2020sparse,
  title={Sparse High-Dimensional Regression: Exact Scalable Algorithms and Phase Transitions},
  author={Bertsimas, Dimitris and Van Parys, Bart},
  journal={The Annals of Statistics},
  volume={48},
  number={1},
  pages={300--323},
  year={2020},
  doi={10.1214/18-AOS1804}
}

@article{gamella2025causal,
  title={Causal chambers as a real-world physical testbed for AI methodology},
  author={Gamella, Juan L and Peters, Jonas and B{\"u}hlmann, Peter},
  journal={Nature Machine Intelligence},
  volume={7},
  number={1},
  pages={107--118},
  year={2025},
  publisher={Nature Publishing Group UK London}
}

\clearpage
\appendix

% \tableofcontents
\renewcommand{\thesection}{\Alph{section}}

\setcounter{section}{0}
\vspace{-0.2in}
\section{Details on MIP algorithms}

\vspace{-0.1in}
\subsection{MIP formulations for both tasks}
\label{sec:app:mip-details}

The procedure \eqref{eqn:task-1-mip} can be formulated as a convex MIP similarly as in Section \ref{sec:MIP}:
\begin{subequations}\label{eq:mip-task1}
\begin{align}
     \min_{
     \substack{ 
     \beta \in\mathbb{R}^p \\
     \{\beta^e \in \mathbb{R}^p\}_{e\in\mathcal{E}} \\ 
     g \in \{0,1\}^p 
     }} 
     \quad&
     \sum_{e\in\mathcal{E}} \omega^e
     \Big\| Y^e - \sum_{k=1}^p X_k^e \beta_k \Big\|_{n^e}^2
     + \gamma \sum_{\substack{e,f\in\mathcal{E}\\e<f}}\omega^{e,f}\Big[ 
     \|\beta^e-\beta^f\|_2^2
     + \rho \left( d^{e,f}(\beta^e,\beta^f) \right)^2 \Big]
     + \lambda_n \sum_{j=1}^p g_j,
     \nonumber\\
\text{subject to:}\quad&
     -M_n g_j \leq \beta_j^e \leq M_n g_j,
     \quad\quad \forall e\in\mathcal{E},\ j\in[p],
     \label{eq:mip-task1-selection}\\
&
     -M_n g_j \leq \beta_j \leq M_n g_j,
     \quad\quad \forall j\in[p],
     \label{eq:mip-task1-selection2}\\
&
     -\widetilde M_n(1-g_j)
     \leq
     \Big\langle X_j^e,
     Y^e-\sum_{k=1}^p\beta_k^eX_k^e
     \Big\rangle_{n^e}
     \leq
     \widetilde M_n(1-g_j),
     \quad\quad \forall e\in\mathcal{E},\ j\in[p].
     \label{eq:mip-task1-normal}
\end{align}
\end{subequations}
Here, the heterogeneity in the OLS risks is written as:
$$
d^{e,f}(\beta^e,\beta^f):= \left\| Y^e \right\|_{n^e}^2 - \left\| Y^f \right\|_{n^f}^2 - \sum_{j\in[p]}\beta^e_j \langle X^e_j, Y^e \rangle_{n^e} + \sum_{j\in[p]}\beta^f_j \langle X^f_j, Y^f \rangle_{n^f}.
$$
Compared with \eqref{eq:inv-score}, the additional selection constraints in \eqref{eq:mip-task1-selection2} ensure that the coefficient vector $\beta$ appearing in the MSE term is restricted to the same set as the environment specific coefficients $\{\beta^e\}_{e\in\mathcal{E}}$.

% Here, $\omega^{(e)} \geq 0$ and $\omega^{(e,f)}\geq 0$ are weights that weigh the impact of each enviornment or a pair of environments in the objective with $\sum_{e\in\mathcal{E}}\omega^{(e)} = \sum_{e,f\in\mathcal{E}, e\neq f}\omega^{(e,f)}=1$. We suggest either choosing the weights to be uniform across enviornments, i.e., $\omega^{(e)} = \frac{1}{|\mathcal{E}|}$ and $\omega^{(e,f)} = \frac{1}{|\mathcal{E}|(|\mathcal{E}|-1)/2}$, or in the finite-sample version of the estimator (discussed shortly), set them based on available data size in each environment. 

% As discussed in Section~\ref{sec:MIP:algo}, this convex MIP formulation maintains valid lower and upper bounds on the objective value throughout the optimization process, and the gap between them provides a certificate of solution quality. When the optimality gap reaches zero, the incumbent solution is globally optimal; we denote the resulting coefficient vector and support by $\widehat{\beta}$ and $\widehat{S}$, respectively. To reduce computation time, the procedure may instead be terminated early once the optimality gap falls below a user-specified tolerance $\tau^{\rm early}>0$. We denote the corresponding feasible solution by $\widehat{\beta}^{\rm early}$ and $\widehat{S}^{\rm early}$. In the following section, we show that, provided $\tau^{\rm early}$ is appropriately controlled, early stopping has only a limited effect on the statistical properties of the resulting estimators.

\vspace{0.1in}
Similarly, the procedure \eqref{eqn:task-2-mip}, and equivalently \eqref{eq:test-stat}, can be formulated as the following convex MIP:
\begin{subequations}
\begin{align}
     \widehat{T}(k) =\min_{
     \substack{ 
     \{\beta^e \in \mathbb{R}^p \}_{e\in\mathcal{E}}  \\ 
     g \in \{0,1\}^p } 
     } \quad&  n^\star\sum_{\substack{e,f \in\mathcal{E}\\e<f}} \omega^{e,f}\Big[  \left\| \beta^e - \beta^f \right\|_2^2 + \rho \left( d^{e,f}(\beta^e,\beta^f) \right)^2\Big], \\
\text{subject to:} \quad& \beta^e_j = 0 \text{ if } g_j = 0, \qquad\qquad\qquad\qquad\qquad\qquad \forall e\in\mathcal{E}, \ j\in [p], \label{eq:mip-task2-selection} \\ 
& \Big\langle X_j^e, Y^e - \sum_{k=1}^p \beta_k^e X_k^e \Big\rangle_{n^e}  = 0 \text{ if } g_j = 1, \qquad\ \  \forall e \in \mathcal{E}, \ j\in[p], \label{eq:mip-task2-normal} \\
& g_{k} = 0. \label{eq:mip-task2-exclude}
\end{align}
\label{eq:mip-task2} 
\end{subequations}
The constraints~\eqref{eq:mip-task2-selection} and \eqref{eq:mip-task2-normal} have the same interpretations as \eqref{eq:inv-score-selection} and \eqref{eq:inv-score-normal}, respectively. The constraint~\eqref{eq:mip-task2-exclude} prevents feature $X_k$ from being selected.

\vspace{0.1in}
\subsection{A heuristic coordinate descent}
\label{sec:app:cd-algo}

We focus on the optimization problem~\eqref{eq:mip-task1}; the coordinate descent algorithm for~\eqref{eq:mip-task2} can be developed similarly and is simpler.

First, we note that, for any fixed support $S$, the objective function in~\eqref{eq:mip-task1} can be evaluated in closed form, and we denote its value by $f(S)$. Specifically, let $\widehat{\beta}^{e,S}$ denote the environment specific OLS solution and define the pooled OLS solution as
$$
\widehat{\beta}^S :=\argmin_{\substack{\tilde\beta\in\mathbb R^p, \
\mathrm{supp}(\tilde\beta)\subseteq S}} \
\sum_{e\in\mathcal{E}} \omega^e\ \widehat{\mathbb{E}} \left[ \left( Y^{e}-\tilde\beta^\top X^{e}\right)^2\right],
$$
Then the objective value depends only on the environment specific and pooled OLS solutions and can be written as
$$
f(S) = \sum_{e\in\mathcal{E}} \omega^e\left\| Y^e - \sum_{k=1}^p X^e_k \widehat{\beta}^S_k \right\|^2_{n^e}
+ \gamma \sum_{e<f} \omega^{e,f} \left\| \widehat{\beta}^{e,S} - \widehat{\beta}^{f,S} \right\|_2^2 + \gamma\rho \sum_{e<f} \omega^{e,f} \left( \widehat{\nu}^{e,S} - \widehat{\nu}^{f,S} \right)^2
+ \lambda_n |S|.
$$
Consequently, $f(S)$ can be evaluated efficiently as long as the environment specific and pooled OLS solutions can be updated efficiently. In a coordinate descent procedure, each candidate move modifies the current support by only one feature, either adding a feature $k\notin S$ or removing a feature $j\in S$. The corresponding inverse Gram matrices and OLS coefficients can therefore be updated using rank one formulas, as described in Algorithms~\ref{algo:rank1-add} and \ref{algo:rank1-remove}. As a result, matrix inversion is required only once for the initial support, while subsequent candidate supports can be evaluated through inexpensive rank one updates.

Our coordinate descent algorithm updates the support one feature at a time. For each $j\in[p]$, if $j\in S$, the algorithm removes $j$ whenever the updated support $S\setminus\{j\}$ yields a smaller objective value; otherwise, the support remains unchanged. Similarly, if $j\notin S$, the algorithm adds $j$ whenever $S\cup\{j\}$ decreases the objective value. The procedure cycles through the coordinates until no addition or removal of any feature further improves the objective.

When the number of variables $p$ is large, repeatedly cycling through all features can be computationally expensive. To accelerate the procedure, we adopt an active set strategy~\citep{hazimeh2020fast,behdin2026multi}. We begin with an initial active set $\mathcal{A}_{\rm active}\subseteq[p]$ and run coordinate descent restricted to this set, imposing
$$
\beta_j^e=0 \quad \text{and}\quad \beta_j=0, \qquad
\forall j\notin\mathcal{A}_{\rm active},\quad e\in\mathcal{E}.
$$
Once coordinate descent converges on the restricted problem, we perform one full pass over all $p$ features. If no coordinate update decreases the objective, then the current solution is also a coordinatewise local optimum of the original problem. Otherwise, we enlarge the active set to include the newly selected features,
$$
\mathcal{A}_{\rm active} \gets \mathcal{A}_{\rm active} \cup 
\left\{ j: \beta_j^e\neq0\text{ for some }e\in\mathcal{E} \ \text{or}\ \beta_j\neq0 \right\},
$$
and repeat the procedure.

Finally, to further improve the quality of the solution obtained by coordinate descent, we optionally apply a local search procedure~\citep{hazimeh2020fast,behdin2026multi}. Starting from the converged support $S$, we consider all swaps between two features, in the form $S\setminus\{j\}\cup\{k\}$, where $j\in S$ and $k\notin S$. If any such swap yields a smaller objective value, we update the support accordingly and restart the coordinate descent procedure. If no improving swap is found, the algorithm terminates at the current solution.

The full algorithm is provided in Algorithms~\ref{algo:coordinate-update}, \ref{algo:local-search} and \ref{algo:coordinate-descent}.

\vspace{0.1in}
\subsection{Enumerating \texorpdfstring{top-$k$}{top-k} invariant sets}
\label{sec:app:top-k}

As discussed in Remark~\ref{rem:top-k}, the top-$k$ most invariant predictor sets can be obtained by solving a sequence of MIP problems. Suppose that we have obtained $L<k$ distinct predictor sets $\widehat S^1,\ldots,\widehat S^L$, ordered by their empirical invariance scores. To obtain the next set, we solve
$$
\widehat S^{L+1} \in \arg\min_{S\subseteq[p]}\widehat{\mathsf I}(S) \quad\text{s.t.}\quad S\neq\widehat S^\ell,\quad \forall\ell\in[L].
$$

Recall that the binary variables $g_j\in\{0,1\}$ indicate whether predictor $j$ is selected. The constraint $S\neq\widehat S^\ell$ can be expressed as
\begin{equation}
\sum_{j\in\widehat S^\ell}(1-g_j)
+\sum_{j\notin\widehat S^\ell}g_j
\geq 1,
\qquad \forall\ell\in[L].
\label{eq:const-top-k1}
\end{equation}
This constraint requires the new predictor set to differ from $\widehat S^\ell$ in at least one variable. Since these constraints are linear in the binary variables, they can be directly added to the MIP formulation of the invariance score without affecting its convexity. By successively adding the exclusion constraints~\eqref{eq:const-top-k1} and resolving the MIP, we obtain the top-$k$ distinct predictor sets in increasing order of their empirical invariance scores.

\vspace{0.2in}
% START Algorithm ===============================
\begin{algorithm}[!ht]
\caption{rank one update when adding a feature}
\label{algo:rank1-add}
\begin{algorithmic}[1]
\Require Current support $S$, feature $j\notin S$, 
OLS coefficients $\{\widehat{\beta}^{e,S}\}_{e\in\mathcal{E}}$ and $\widehat{\beta}^{S}$,
and inverse Gram matrices
$\{(G_S^e)^{-1}\}_{e\in\mathcal{E}}$ and
$(G_S^{\mathrm{pool}})^{-1}$, where
$$
G_S^e:=(\boldsymbol{X}_S^e)^\top \boldsymbol{X}_S^e,
\qquad
G_S^{\mathrm{pool}}:=
\sum_{e\in\mathcal E}\frac{\omega^e}{n_e}(\boldsymbol{X}_S^e)^\top \boldsymbol{X}_S^e.
$$
\Ensure Updated inverse Gram matrices and OLS coefficients on
$R=S\cup\{j\}$.
\For{$e\in\mathcal{E}$}
    \State Compute
    $$
    c^e = \boldsymbol{X}_S^{e\top} \boldsymbol{X}_j^e, \qquad
    d^e = \boldsymbol{X}_j^{e\top} \boldsymbol{X}_j^e, \qquad
    z^e = \boldsymbol{X}_j^{e\top} \boldsymbol{Y}^e.
    $$
    \State Compute
    $$
    u^e=(G_S^e)^{-1}c^e, \qquad
    \delta^e=d^e-(c^e)^\top u^e.
    $$
    \State Update the inverse Gram matrix:
    $$
    (G_R^e)^{-1} =
    \begin{bmatrix}
    (G_S^e)^{-1}+\dfrac{u^e(u^e)^\top}{\delta^e}
    &
    -\dfrac{u^e}{\delta^e}
    \\[6pt]
    -\dfrac{(u^e)^\top}{\delta^e}
    &
    \dfrac{1}{\delta^e}
    \end{bmatrix}.
    $$
    \State Compute the coefficient of the newly added feature, and update:
    $$
    \widehat{\beta}_j^{e,R} = \frac{ z^e-(c^e)^\top\widehat{\beta}^{e,S} } {\delta^e}, \qquad
    \widehat{\beta}^{e,R} =
    \begin{bmatrix}
    \widehat{\beta}^{e,S} - u^e\widehat{\beta}_j^{e,R}
    \\ \widehat{\beta}_j^{e,R}
    \end{bmatrix}.
    $$
\EndFor
\State Compute the pooled quantities
$$
c^{\mathrm{pool}}=\sum_{e\in\mathcal{E}}\omega^e\frac{ \boldsymbol{X}_S^{e\top} \boldsymbol{X}_j^e }{n^e}, \qquad
d^{\mathrm{pool}} = \sum_{e\in\mathcal{E}} \omega^e \frac{ \boldsymbol{X}_j^{e\top} \boldsymbol{X}_j^e }{n^e}, \qquad
z^{\mathrm{pool}} = \sum_{e\in\mathcal{E}}\omega^e \frac{ \boldsymbol{X}_j^{e\top} \boldsymbol{Y}^e }{n^e}.
$$
\State Compute
$$
u^{\mathrm{pool}} = (G_S^{\mathrm{pool}})^{-1}c^{\mathrm{pool}}, \qquad
\delta^{\mathrm{pool}} = d^{\mathrm{pool}} - (c^{\mathrm{pool}})^\top u^{\mathrm{pool}}.
$$
\State Update the inverse Gram matrix:
$$
(G_R^{\mathrm{pool}})^{-1} =
\begin{bmatrix}
(G_S^{\mathrm{pool}})^{-1} +
\dfrac{ u^{\mathrm{pool}}(u^{\mathrm{pool}})^\top }{ \delta^{\mathrm{pool}} }
&
-\dfrac{u^{\mathrm{pool}}}{\delta^{\mathrm{pool}}}
\\[6pt]
-\dfrac{(u^{\mathrm{pool}})^\top}{\delta^{\mathrm{pool}}}
&
\dfrac{1}{\delta^{\mathrm{pool}}}
\end{bmatrix}.
$$
\State Compute the coefficient of the newly added feature, and update:
$$
\widehat{\beta}_j^R = \frac{ z^{\mathrm{pool}} - (c^{\mathrm{pool}})^\top\widehat{\beta}^S }{ \delta^{\mathrm{pool}} }, \qquad
\widehat{\beta}^R =
\begin{bmatrix}
\widehat{\beta}^S - u^{\mathrm{pool}}\widehat{\beta}_j^R
\\
\widehat{\beta}_j^R
\end{bmatrix}.
$$
\State \Return
$\{(G_R^e)^{-1},\widehat{\beta}^{e,R}\}_{e\in\mathcal{E}}$
and
$\{(G_R^{\mathrm{pool}})^{-1},\widehat{\beta}^R\}$.

\end{algorithmic}
\end{algorithm}
% END Algorithm ===============================

\vspace{0.2in}
% START Algorithm ===============================
\begin{algorithm}[t]
\caption{rank one update when removing a feature}
\label{algo:rank1-remove}
\begin{algorithmic}[1]
\Require Current support $S$, feature $j\in S$, OLS coefficients $\{\widehat{\beta}^{e,S}\}_{e\in\mathcal{E}}$ and $\widehat{\beta}^S$, and inverse Gram matrices
$\{(G_S^e)^{-1}\}_{e\in\mathcal{E}}$ and
$(G_S^{\mathrm{pool}})^{-1}$.
\Ensure Updated inverse Gram matrices and OLS coefficients on
$R=S\setminus\{j\}$.
\State Reorder the coordinates, if necessary, so that feature $j$ is the
last coordinate of $S$.
\For{$e\in\mathcal{E}$}
    \State Partition
    $$
    (G_S^e)^{-1} =
    \begin{bmatrix}
    A^e & b^e \\
    (b^e)^\top & d^e
    \end{bmatrix},
    \qquad
    \widehat{\beta}^{e,S} = \begin{bmatrix}
    \widehat{\beta}_R^{e,S}\\
    \widehat{\beta}_j^{e,S}
    \end{bmatrix}.
    $$
    \State Update the inverse Gram matrix and the OLS coefficient:
    $$
    (G_R^e)^{-1} = A^e-\frac{b^e(b^e)^\top}{d^e}, \qquad
    \widehat{\beta}^{e,R} = \widehat{\beta}_R^{e,S} - \frac{b^e}{d^e}\widehat{\beta}_j^{e,S}.
    $$
\EndFor
\State Partition the pooled quantities as
$$
(G_S^{\mathrm{pool}})^{-1} =
\begin{bmatrix}
A^{\mathrm{pool}} & b^{\mathrm{pool}}\\
(b^{\mathrm{pool}})^\top & d^{\mathrm{pool}}
\end{bmatrix}, \qquad
\widehat{\beta}^{S} =
\begin{bmatrix}
\widehat{\beta}_R^S\\
\widehat{\beta}_j^S
\end{bmatrix}.
$$
\State Update
$$
(G_R^{\mathrm{pool}})^{-1} = A^{\mathrm{pool}} - \frac{ b^{\mathrm{pool}}(b^{\mathrm{pool}})^\top }{ d^{\mathrm{pool}} }, \qquad
\widehat{\beta}^{R} = \widehat{\beta}_R^S - \frac{ b^{\mathrm{pool}} }{ d^{\mathrm{pool}} } \widehat{\beta}_j^S.
$$
\State \Return
$\{(G_R^e)^{-1},\widehat{\beta}^{e,R}\}_{e\in\mathcal{E}}$
and
$\{(G_R^{\mathrm{pool}})^{-1},\widehat{\beta}^R\}$.
\end{algorithmic}
\end{algorithm}
% END Algorithm ===============================

\vspace{0.2in}
% START Algorithm ===============================
\begin{algorithm}[t]
\caption{\textsc{CoordinateUpdate}: Coordinate update on a prescribed set }
\label{algo:coordinate-update}
\begin{algorithmic}[1]
\Require Current support $S$, index set $\mathcal{A}\subseteq[p]$, 
objective function $f$, and indicator $\mathtt{converge}$.
\Ensure Updated support $S$ and indicator $\mathtt{changed}$.
\State $\mathtt{changed}\gets\mathtt{False}$.
\Repeat
    \State $\mathtt{updated}\gets\mathtt{False}$.
    \For{$j\in\mathcal{A}$}
        \If{$j\in S$}
            \State $S'\gets S\setminus\{j\}$.
        \Else
            \State $S'\gets S\cup\{j\}$.
        \EndIf
        \State Evaluate $f(S')$ using the corresponding rank one update.
        \If{$f(S')<f(S)$}
            \State $S\gets S'$ and $f(S)\gets f(S')$.
            \State Update the inverse Gram matrices and OLS coefficients.
            \State $\mathtt{updated}\gets\mathtt{True}$.
            \State $\mathtt{changed}\gets\mathtt{True}$.
        \EndIf
    \EndFor
\Until{$\mathtt{updated}=\mathtt{False}$ or $\mathtt{converge}=\mathtt{False}$}
\State \Return $S$ and $\mathtt{changed}$.
\end{algorithmic}
\end{algorithm}
% END Algorithm ===============================

\vspace{0.2in}
% START Algorithm ===============================
\begin{algorithm}[t]
\caption{\textsc{LocalSearch}: Two-feature local search}
\label{algo:local-search}
\begin{algorithmic}[1]
\Require Current support $S$ and objective function $f$.
\Ensure Updated support $S$ and indicator $\mathtt{improved}$.
\State $\mathtt{improved}\gets\mathtt{False}$.
\For{$j\in S$}
    \For{$k\notin S$}
        \State $S'\gets (S\setminus\{j\})\cup\{k\}$.
        \State Evaluate $f(S')$ using a rank one removal followed by a rank one addition.
        \If{$f(S')<f(S)$}
            \State $S\gets S'$ and $f(S)\gets f(S')$.
            \State Update the inverse Gram matrices and OLS coefficients.
            \State $\mathtt{improved}\gets\mathtt{True}$.
            \State \Return $S$ and $\mathtt{improved}$.
        \EndIf
    \EndFor
\EndFor
\State \Return $S$ and $\mathtt{improved}$.
\end{algorithmic}
\end{algorithm}
% END Algorithm ===============================

\vspace{0.2in}
% START Algorithm ===============================
\begin{algorithm}[t]
\caption{\textsc{CoordinateDescent}: Coordinate descent}
\label{algo:coordinate-descent}
\begin{algorithmic}[1]
\Require Initial support $S$, initial active set 
$\mathcal{A}_{\rm active}\supseteq S$, and indicator \texttt{local\_search}.
\Ensure Estimated support $\widehat{S}$.
\State Compute the inverse Gram matrices and OLS coefficients on the initial support $S$.
\Repeat
    \State $(S,\cdot)\gets
    \textsc{CoordinateUpdate}
    (S,\ \mathcal{A}_{\rm active},\ f,\ \mathtt{True})$.
    \State $(S,\ \mathtt{changed})\gets
    \textsc{CoordinateUpdate}
    (S,\ [p],\ f,\ \mathtt{False})$.
    \If{$\mathtt{changed}=\mathtt{True}$}
        \State $\mathcal{A}_{\rm active} \gets \mathcal{A}_{\rm active}\cup S$.
    \EndIf
\Until{$\mathtt{changed} = \mathtt{False}$}
\If{\texttt{local\_search} = \texttt{True}}
    \Repeat
        \State $(S,\ \mathtt{improved})
        \gets\textsc{LocalSearch}(S,\ f)$.
        \If{$\mathtt{improved}=\mathtt{True}$}
            \State $\mathcal{A}_{\rm active} \gets\mathcal{A}_{\rm active}\cup S$.
            \Repeat
                \State $(S,\ \cdot)\gets
                \textsc{CoordinateUpdate}(S,\ \mathcal{A}_{\rm active},\ f,\ \mathtt{True})$.

                \State $(S,\ \mathtt{changed})\gets \textsc{CoordinateUpdate}(S,\ [p],\ f,\ \mathtt{False})$.

                \State $\mathcal{A}_{\rm active} \gets\mathcal{A}_{\rm active}\cup S$.
            \Until{$\mathtt{changed}=\mathtt{False}$}
        \EndIf
    \Until{$\mathtt{improved}=\mathtt{False}$}
\EndIf
\State $\widehat{S}\gets S$.
\State \Return $\widehat{S}$.
\end{algorithmic}
\end{algorithm}
% END Algorithm ===============================

\clearpage
\vspace{-0.2in}
\section{More on the Bootstrap procedure}
\label{sec:app:bootstrap}

\vspace{-0.1in}
\subsection{Procedures that estimate \texorpdfstring{$\mathcal{S}_k^\circ$}{Sk circ} }

In Section~\ref{sec:task2:bootstrap}, we briefly introduced two procedures for constructing a super-collection of $\mathcal{S}_k$. The first procedure, which can be described in Algorithm~\ref{algo:perturb-D1-task2}, collects $\bar M$ candidate sets by repeatedly resampling the observations in $\mathcal{D}_1$ with replacement. The intuition is that these bootstrap perturbations induce sufficient variation in the selected models to recover the invariant sets in $\mathcal{S}_k$. This procedure, however, is largely heuristic and does not come with strong theoretical guarantees.

% Algorithm ===============================
\begin{algorithm}[!ht]
\caption{Construction of $\mathcal{S}_k^\circ$ via perturbation}
\label{algo:perturb-D1-task2}
\begin{algorithmic}[1]
\Require Dataset $\mathcal{D}_1$,
feature index $k$, and maximum number of sets $\bar{M}$.
\Ensure A collection of sets $\mathcal{S}_k^\circ$.
\State Initialize $\mathcal{S}_k^\circ \gets \emptyset$.
\For{$m=1,\ldots,\bar{M}$}
    \State Construct a bootstrap sample $\mathcal{D}_1^{(m)}$ by drawing observations from $\mathcal{D}_1$ with replacement within each environment.
    \State Solve the optimization problem \eqref{eq:mip-task2} using
    $\mathcal{D}_1^{(m)}$.
    \State Let $\widehat{S}^{(m)}$ denote the set encoded by the optimal solution.
    \State Update $ \mathcal{S}_k^\circ \gets \mathcal{S}_k^\circ\cup\{\widehat{S}^{(m)}\}$.
\EndFor
\State \Return $\mathcal{S}_k^\circ$.
\end{algorithmic}
\end{algorithm}
% Algorithm ===============================

The second procedure sequentially retrieves sets $S$ in ascending order of the objective value in~\eqref{eq:mip-task2}, motivated by the intuition that smaller objective values correspond to more invariant models. Consequently, the top-$\bar M$ sets are expected to contain the invariant sets in $\mathcal{S}_k$. The details of this procedure are presented in Algorithm~\ref{algo:top-M-task2}. 
Notably, such a sequential retrieval based on the ordering of objective values would not be readily available without the mixed-integer programming formulation, particularly the exclusion constraints imposed in line~3.
Furthermore, under the Gaussian model $Y^e \mid X^e_S \sim \mathcal{N}(X^{e\top} \beta^{e,S}, \nu^{e,S})$, both requirements for the testing procedure, $\mathcal{S}_k \subseteq \mathcal{S}_k^\circ$ and $\widehat{S} \in \mathcal{S}_k $, can be satisfied with probability tending to one.
The proof of Lemma~\ref{lem:bootstrap-D1} is provided in Appendix~\ref{sec:app:proof:bootstrap-D1}.

\vspace{0.1in}
\begin{lemma}\label{lem:bootstrap-D1}
    Suppose that $\mathcal{D}_1$ and $\mathcal{D}_2$ each contain half of the original dataset, and $|\mathcal{S}_k|\leq \bar{M}$. Under the Gaussian model $Y^e \mid X^e_S \sim \mathcal{N}(X^{e\top} \beta^{e,S}, \nu^{e,S})$ for all $e\in \mathcal{E}$, for the set $\widehat{S}$ encoded by the optimal solution to~\eqref{eq:mip-task2} and the collection $\mathcal{S}^\circ_k$ obtained from Algorithm \ref{algo:top-M-task2}, we have
    $$
    \mathbb{P}\left( \left\{ \mathcal{S}_k \subseteq \mathcal{S}_k^\circ \right\} \cap \left\{ \widehat{S} \in \mathcal{S}_k \right\} \right) \rightarrow 1,
    $$
    as the sample sizes $\min_{e\in\mathcal{E}} n^e\rightarrow \infty$.
\end{lemma}

% Algorithm ===============================
\begin{algorithm}[!ht]
\caption{Sequential retrieval of the top-$\bar{M}$ sets}
\label{algo:top-M-task2}
\begin{algorithmic}[1]
\Require Dataset $\mathcal{D}_1$,
feature index $k$, and maximum number of sets $\bar{M}$.
\Ensure A collection of sets $\mathcal{S}^\circ_k$.
\State Initialize $\mathcal{S}^\circ_k \gets \emptyset$.
\For{$m=1,\ldots,\bar{M}$}
    \State Solve the mixed-integer program \eqref{eq:mip-task2} on $\mathcal{D}_1$, together with exclusion constraints 
    $$
    \sum_{j \in S} (1 - g_j) + \sum_{j \notin S} g_j \geq 1, \quad \forall S\in \mathcal{S}^\circ_k.
    $$
    \State Let $\widehat{S}^{(m)}$ denote the set encoded by the optimal solution.
    \State Update $\mathcal{S}^\circ_k
    \gets
    \mathcal{S}^\circ_k \cup \{\widehat{S}^{(m)}\}$.
\EndFor
\State \Return $\mathcal{S}^\circ_k$.
\end{algorithmic}
\end{algorithm}
% Algorithm ===============================

\vspace{0.1in}
\subsection{Gaining more power}
\label{sec:app:gain-more-power}
Although the two components of the invariance score $\mathsf{I}(S)$, namely the discrepancies in the OLS coefficients and residual variances, are sufficient to characterize LSCR invariance, incorporating additional measures of discrepancy across environments can further improve the power of the resulting test. We additionally consider the following two conditions for a model $\beta$: (i) $\mathbb{E}[ Y^e - X^e\beta ] =\mathbb{E}[ Y^f - X^f\beta ]$ for all $e,f\in\mathcal{E}$; and (ii) $\mathbb{E} [X^e_j (Y^e - X^e \beta)] = \mathbb{E} [X^f_j (Y^f - X^f \beta)] $ for all $j\in \mathrm{supp}(\beta)$ and $e,f\in\mathcal{E}$. Both conditions are automatically satisfied by any LSCR-invariant model and therefore do not alter the collection $\mathcal{S}_{\rm LSCR}(\mathcal{E})$ characterized by the LSCR principle. Nevertheless, incorporating their empirical discrepancies into the invariance score can provide additional information for distinguishing non-invariant models, thereby improving the power to detect plausible causal predictors.

The corresponding MIP formulation for testing $H_{0,k_0}$ can be modified as:
\begin{subequations} \label{eq:mip-task2-power}
\begin{align}
    \min_{\substack{ \beta \in\mathbb{R}^p ,\
     \{\beta^e \in \mathbb{R}^p \}_{e\in\mathcal{E}}  \\ 
     g \in \{0,1\}^p, \ \{ t_l \geq 0 \}_{l=1}^4  } } \quad& n^\star \left( \gamma_1 t_1 + \gamma_2 t_2 + \gamma_3 t_3 + \gamma_4 t_4 \right) \nonumber \\
    \text{subject to:} \quad& \sum_{e<f} \omega^{e,f} \|\beta^e - \beta^f\|_2^2 \leq t_1, \label{eq:mip-task2more-t1}  \\
    & \sum_{e<f} \omega^{e,f} \left( \left\| Y^e \right\|_{n^e}^2 - \left\| Y^f \right\|_{n^f}^2 - \sum_{j=1}^p\beta^e_j \langle X^e_j, Y^e \rangle_{n^e} + \sum_{j=1}^p\beta^f_j \langle X^f_j, Y^f \rangle_{n^f} \right)^2 \leq t_2, \label{eq:mip-task2more-t2} \\ 
    & \sum_{e<f} \omega^{e,f} \left( \Big\langle Y^e - \sum_{j=1}^p X^e_j \beta_j, 1 \Big\rangle_{n^e} -  \Big\langle Y^f - \sum_{j=1}^p X^f_j \beta_j, 1 \Big\rangle_{n^f}  \right)^2 \leq t_3, \label{eq:mip-task2more-t3} \\
    & \sum_{e<f} \omega^{e,f} \left( \Big\langle Y^e - \sum_{j=1}^p X^e_j \beta_j, X^e_k \Big\rangle_{n^e} -  \Big\langle Y^f - \sum_{j=1}^p X^f_j \beta_j, X^f_k \Big\rangle_{n^f}  \right)^2 \leq t_4,  \quad \text{if } g_k =1, \label{eq:mip-task2more-t4} \\
    & \beta^e_j = 0 \text{ if } g_j = 0, \qquad\qquad\qquad\qquad\qquad\qquad \forall e\in\mathcal{E}, \ j\in [p], \label{eq:mip-task2more-selection} \\ 
    & \beta_j = 0 \text{ if } g_j = 0, \qquad\qquad\qquad\qquad\qquad\qquad \forall  j\in [p], \label{eq:mip-task2more-selection1} \\ 
    & \Big\langle X_j^e, Y^e - \sum_{k=1}^p \beta_k^e X_k^e \Big\rangle_{n^e}  = 0 \text{ if } g_j = 1, \qquad\ \  \forall e \in \mathcal{E}, \ j\in[p], \label{eq:mip-task2more-normal} \\
    & \sum_{e\in\mathcal{E}} \frac{n^e}{n^\star} \Big\langle X^e_j, Y^e - \sum_{k=1}^p \beta_k X^e_k \Big\rangle_{n^e}  = 0 \text{ if } g_j = 1, \qquad\quad \ j\in[p], \label{eq:mip-task2more-normal1} \\
    & g_{k_0} = 0.
\end{align}
\end{subequations}
For the objective function, the first two components coincide with those in formulation~\eqref{eq:mip-task2}, through constraints~\eqref{eq:mip-task2more-t1} and \eqref{eq:mip-task2more-t2}. The last two components, introduced through \eqref{eq:mip-task2more-t3} and \eqref{eq:mip-task2more-t4}, capture the additional discrepancies across environments in the residual mean, corresponding to condition~(i), and in the cross moment between residual and active variables, corresponding to condition~(ii).
Constraints~\eqref{eq:mip-task2more-selection} and \eqref{eq:mip-task2more-normal} ensure that each $\beta^e$ is the environment specific OLS solution, whereas constraints~\eqref{eq:mip-task2more-selection1} and \eqref{eq:mip-task2more-normal1} ensure that $\beta$ is the pooled OLS solution based on the full dataset $\{(X^e,Y^e)\}_{e\in\mathcal{E}}$ of all environments. Importantly, the pooled OLS coefficient is used in the additional discrepancy terms. If the environment specific OLS coefficients were used instead, the corresponding empirical moment conditions would hold automatically by the OLS normal equations: for every active predictor $j$ with $g_j=1$, $\langle X_j^e,\,Y^e-\sum_{k=1}^p \beta_k^e X_k^e \rangle_{n^e}=0$, and
$\langle Y^e-\sum_{j=1}^p X_j^e\beta_j^e,\,1 \rangle_{n^e}=0$.
In that case, these quantities would provide no additional information for detecting violations of invariance. Using the pooled OLS solution instead allows the last two components to capture meaningful discrepancies across environments and thereby potentially improve the power of the test.

Using this modified objective value as the test statistic, we accordingly adapt the multiplier bootstrap procedure in Algorithm~\ref{algo:bootstrap-threshold}. In line 4, we additionally compute the pooled OLS coefficient $\widehat{\beta}^S$ and 
$$
\widehat{r}_i^{e,S,\mathrm{pool}}  = Y_i^e-(X_{i}^{e})^\top\widehat{\beta}^S, \quad
\widehat{\Sigma}_S = \sum_{e\in\mathcal{E}} \frac{n^e}{n^\star} \widehat{\Sigma}^e_S, \quad
B^{\star,S} = (\widehat{\Sigma}_{S})^{-1} \left[\frac{1}{\sqrt{n^\star}} \sum_{e\in\mathcal{E}} \sum_{i=1}^{n^e} \xi_i^e X_{S,i}^{e} \widehat{r}_i^{e,S, \rm pool} \right],
$$
$$
M^{\star,e,S} = \frac{1}{\sqrt{n^e}} \sum_{i=1}^{n^e} \xi^e_i \left( \widehat{r}_i^{e,S, \rm pool} - \frac{1}{n^e} \sum_{i=1}^{n^e} \widehat{r}_i^{e,S, \rm pool} \right), 
$$
and
$$
Q_j^{\star,e,S} = \frac{1}{\sqrt{n^e}} \sum_{i=1}^{n^e} \xi^e_i \left(X^e_{ij} \widehat{r}_i^{e,S, \rm pool} - \frac{1}{n^e} \sum_{i=1}^{n^e} X^e_{ij}\widehat{r}_i^{e,S, \rm pool} \right).
$$
In line 5, for each pair $e<f$ in $\mathcal{E}$, we additionally compute
$$
O_n^{\star, e,f, S} = \sqrt{\frac{n^\star}{n^e}} M^{\star,e, S} - \sqrt{\frac{n^\star}{n^f}} M^{\star,f, S} - \left( \frac{1}{n^e} \sum_{i=1}^{n^e} X_{S,i}^{e} -\frac{1}{n^f} \sum_{i=1}^{n^f} X^{f}_{S,i} \right)^\top B^{\star, S}
$$
and
$$
K_{n,j}^{\star, e,f, S} = \sqrt{\frac{n^\star}{n^e}} Q_j^{\star,e, S} - \sqrt{\frac{n^\star}{n^f}} Q_j^{\star,f, S} - \left( \frac{1}{n^e} \sum_{i=1}^{n^e} X_{ij}^e X_{S,i}^{e} -\frac{1}{n^f} \sum_{i=1}^{n^f} X_{ij}^f X_{S,i}^{f} \right)^\top B^{\star, S}.
$$
Finally, the line 7 is modified as:
$$
t_k^b \gets \max_{S \in\mathcal{S}^\circ} \left\{ \sum_{e<f} \omega^{e,f}  \left( 
\gamma_1  \left\| L_n^{\star,e,f,S} \right\|_2^2 + 
\gamma_2  |V_n^{\star, e,f, S}|^2 \ +   
\gamma_3  \left| O_n^{\star, e,f S} \right|^2 \right) + 
\gamma_4\max_{j\in S} \sum_{e<f} \omega^{e,f}  \left| K_{n,j}^{\star, e,f,S} \right|^2  \right\}. 
$$
The validity of this modified bootstrap procedure can be established using the same argument as in Appendix~\ref{sec:app:proof:type-1}.

\section{Extension to nonlinear causal effects}
\label{sec:app:nonlinear}

We now investigate an extension of our framework to nonlinear structural equation models. Although the LSCR invariance principle is developed for linear regression, our variable selection procedure can be extended to nonlinear additive models using basis expansions. Specifically, consider the additive model
$$
Y^e=\sum_{j\in S^\star}f_j^\star(X_j^e)+\varepsilon^e,
\qquad
\varepsilon^e\perp X_{S^\star}^e,\quad
\varepsilon^e\sim\mathcal F,\quad
\mathbb E[\varepsilon^e]=0,
$$
where the nonlinear functions $f_j^\star$ and the noise distribution remain invariant across environments.
To accommodate nonlinear effects, we expand each feature $X_j^e$ into $R_n$ basis functions $b_r(X_j^e)$, $r\in[R_n]$. The nonlinear function $f_j^\star(X_j^e)$ can then be approximated by a linear combination of these basis functions. Denote the expanded covariate vector by
$ Z_j^e= (b_1(X_j^e),\ldots,b_{R_n}(X_j^e) )^\top \in\mathbb R^{R_n}$, and concatenate those vectors into $Z^e\in\mathbb R^{pR_n}$. For a candidate predictor set $S$, we consider the additive regression model
$$
f_S(X_S^e)
=\sum_{j\in S}\sum_{r=1}^{R_n}
\theta_{jr}b_r(X_j^e),
$$
where $\theta_j=(\theta_{j1},\ldots,\theta_{jR_n})^\top$ denotes the coefficient vector associated with the original predictor $X_j$.

We now apply the LSCR invariance principle to the expanded covariates. For each environment $e$ and candidate set $S$, define the population OLS coefficients and corresponding risks as
$$
\theta^{e,S}:=\arg\min_{\theta:\,\theta_j=0,\ j\notin S}
\mathbb E\left[(Y^e-\theta^\top Z^e)^2\right], \qquad
\nu^{e,S}:=\mathbb E\left[
(Y^e-(\theta^{e,S})^\top Z^e)^2
\right].
$$
Let $\widehat\theta^{e,S}$ and $\widehat\nu^{e,S}$ denote their empirical counterparts.
The corresponding empirical invariance score is
$$
\widehat{\mathsf I}_{\rm nonlinear}(S) =
\sum_{e<f}\omega^{e,f}
\left[ \|\widehat\theta^{e,S}
-\widehat\theta^{f,S}\|_2^2
+\rho(\widehat\nu^{e,S}
-\widehat\nu^{f,S})^2 \right].
$$
As in our original framework, we combine this score with the pooled prediction loss and a group $\ell_0$-penalty.
Let $g_j\in\{0,1\}$ indicate whether predictor $j$ is selected, and introduce coefficients $\theta^e\in\mathbb R^{pR_n}$ for each environment. 
The resulting MIP formulation is
\begin{align*}
\begin{split}
\min_{\substack{
\theta\in\mathbb R^{pR_n}, \\
\{\theta^e\in\mathbb{R}^{pR_n}\}_{e\in\mathcal{E}}\\
g\in\{0,1\}^p}}
\quad&
\sum_{e\in\mathcal E}\omega^e
\left\|Y^e-\sum_{j=1}^p\sum_{r=1}^{R_n}
\theta_{jr}b_r(X_j^e)\right\|_{n^e}^2 
+ \lambda_n\sum_{j=1}^p g_j 
+ \gamma\sum_{e<f}\omega^{e,f}
\left[ \|\theta^e-\theta^f\|_2^2
+ \rho\,d^{e,f}(\theta^e,\theta^f)^2
\right],\\
\text{subject to:}\quad&
\theta_{jr}=0 \text{ if } g_j=0,\qquad\forall j\in[p],\ r\in[R_n],\\
& \theta_{jr}^e = 0 \text{ if } g_j = 0, \qquad\forall e\in\mathcal{E},\ j\in[p],\ r\in[R_n], \\
&\left\langle
b_r(X_j^e),\,
Y^e-\sum_{k=1}^p\sum_{s=1}^{R_n}
\theta_{ks}^e b_s(X_k^e)
\right\rangle_{n^e}=0 \text{ if }g_j=1,\qquad\forall e\in\mathcal{E},\  j\in[p],\ r\in[R_n],
\end{split}
\end{align*}
where
$$
d^{e,f}(\theta^e,\theta^f)
:=
\widehat\nu^e(\theta^e)
-\widehat\nu^f(\theta^f), \qquad
\widehat\nu^e(\theta^e)
=
\|Y^e\|_{n^e}^2
-\sum_{j=1}^p\sum_{r=1}^{R_n}
\theta_{jr}^e
\left\langle b_r(X_j^e),Y^e\right\rangle_{n^e}.
$$
For computational efficiency, we may use coordinate descent to obtain approximate solutions, with each coordinate update corresponding to the addition or removal of an entire group of basis functions.

\vspace{0.1in}
\begin{example} \label{ex:nonlinear}
We generate data from a nonlinear structural equation model with $p=6$ covariates and $|S^\star|=2$ causal parents. Specifically, the response is generated according to
$$
Y^e=\sum_{j\in S^\star}\beta_j^\star\tanh(2\sin(X_j^e))+\varepsilon^e, \quad \varepsilon^e\sim \mathcal{N}(0, 0.25),
$$
where the causal coefficients remain invariant across environments. We consider five environments, allowing the structural coefficients and noise variances governing the covariates to vary across environments. We approximate the nonlinear effects using a Fourier basis with four basis functions for each predictor, and apply our group-$\ell_0$ formulation using coordinate descent. We compare the resulting procedure, denoted by nonlinear MIP, with nonlinear ICP~\citep{heinze2018invariant}, WVM~\citep{martinet2022variance}, and FAIR~\citep{gu2025causality}.

Figure~\ref{fig:nonlinear} compares the performance of these methods as the sample size per environment increases. Our proposed nonlinear MIP achieves substantially higher rates of exact support recovery and identifies more true causal predictors than the competing methods. In particular, its recovery probability approaches one as the sample size increases, while the other methods rarely achieve exact recovery in this experiment. Moreover, nonlinear MIP is considerably faster than the competing methods. 
\end{example}

\begin{figure}[!ht]
\centering
\includegraphics[width=\textwidth]{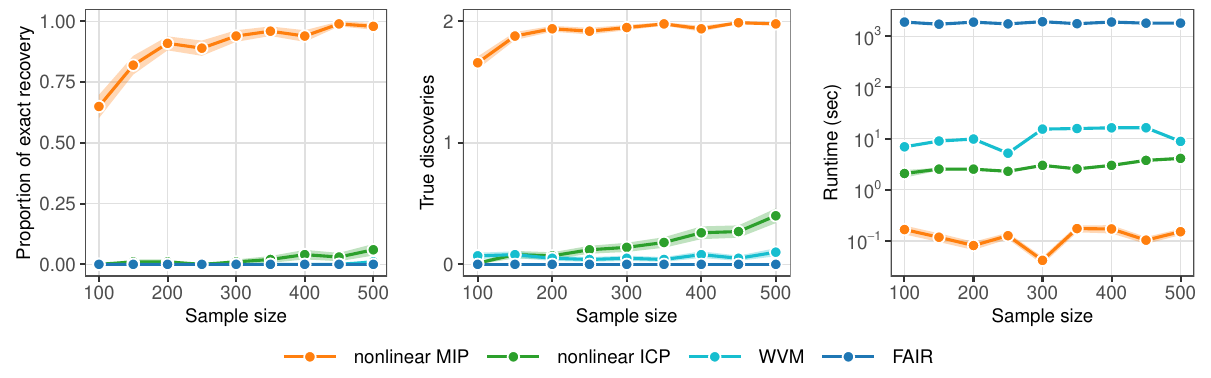}
\caption{\small Performance comparison under nonlinear causal effects as the sample size $n$ increases. 
Left panel: proportion of 100 trials achieving exact recovery of $S^\star$. 
Middle panel: average number of true positives $|\widehat S\cap S^\star|$ over 100 trials. Right panel: average runtime in seconds over 100 trials, shown on a logarithmic scale. Solid lines show averages, and shaded bands indicate $\pm 1$ standard error. Additional experimental details are provided in Appendix~\ref{sec:app:figure-nonlinear}.}
\label{fig:nonlinear}
\end{figure}

One limitation of this extension is that the LSCR invariance principle may not hold exactly under finite basis expansions. To see this, denote the true nonlinear regression function by
$ f^\star(X_{S^\star}^e) =\sum_{j\in S^\star}f_j^\star(X_j^e)$.
Suppose that it admits the representation
$$
f^\star(X_{S^\star}^e) = \sum_{j\in S^\star}\sum_{r=1}^{R_n} \theta^\star_{jr}b_r(X_j^e) + r(X_{S^\star}^e),
$$
where $\theta^\star$ denotes the $L_2$ projection coefficients of $f^\star$ onto the span of selected basis functions under the observational distribution, and $r$ denotes the corresponding approximation error.
Let $B_{S^\star}^e:=\{ b_r(X_j^e) \}_{j\in S^\star,r\in[R_n]}$ denote the vector of basis functions, and define $\Sigma_{S^\star}^e:=\mathbb E[B_{S^\star}^e(B_{S^\star}^e)^\top]$.
By construction, $\mathbb E_{P_{\rm obs}}[B_{S^\star}r(X_{S^\star})]=0$, although this orthogonality need not hold in other environments.
Assuming $\Sigma_{S^\star}^e$ is invertible, the population OLS coefficients satisfy
$$
\theta_{S^\star}^{e,S^\star} =
\theta^\star+  (\Sigma_{S^\star}^e)^{-1}
\mathbb{E}\left[ B_{S^\star}^e r(X_{S^\star}^e) \right].
$$
Similarly, the corresponding residual risk is
$$
\nu^{e,S^\star}
= \sigma_\varepsilon^2
+\mathbb E\left[r(X_{S^\star}^e)^2\right] -
\mathbb E\left[B_{S^\star}^e r(X_{S^\star}^e)\right]^\top
(\Sigma_{S^\star}^e)^{-1}
\mathbb E\left[B_{S^\star}^e r(X_{S^\star}^e)\right],
$$
where $\sigma_\varepsilon^2=\mathbb E[(\varepsilon^e)^2]$ is invariant across environments.
When $r=0$, meaning that the true nonlinear functions lie in the span of the selected basis functions, we have $\theta_{S^\star}^{e,S^\star}=\theta^\star$ and $\nu^{e,S^\star}=\sigma_\varepsilon^2$ for all $e\in\mathcal E$.
Thus, both conditions of the LSCR invariance principle hold at the true causal parent set $S^\star$.
However, when $r\neq0$, the additional terms in the expressions above may depend on the covariate distributions, which can vary across environments.
Consequently, both the population OLS coefficients and residual risks may differ across environments, even though the underlying nonlinear causal mechanism remains invariant.

Increasing the number of basis functions $R_n$ may reduce the approximation error and bring the invariance score closer to zero for $S^\star$.
However, increasing $R_n$ introduces more coefficients to estimate, potentially leading to larger statistical errors and higher computational costs.
Therefore, the choice of basis dimension involves a trade-off between approximation accuracy and estimation accuracy.
We consider this extension as an empirical investigation and leave theoretical guarantees for nonlinear models to future work.

\section{Extensions to when some variables are unobserved}
\label{sec:app:confounding}
We now discuss an extension of the mixed-integer programming framework for invariant prediction in the presence of hidden variables. We consider a structural equation model with noise interventions and latent confounding. Let $H\in\mathbb{R}^q$ denote the hidden variables, $X_A\in\mathbb{R}^{|A|}$ the observed ancestors of the causal parents, $X_{S^\star}\in\mathbb{R}^{|S^\star|}$ the causal parents of $Y$, $Y\in\mathbb{R}$ the response, and $X_D\in\mathbb{R}^{|D|}$ the descendants of $Y$. The latent variables $H$ may directly affect all of these observed variables, while the intervention mechanism, indexed by $e\in\mathcal{E}$, is assumed to be unaffected by $H$. Specifically, for each $e\in\mathcal{E}$,
$$
H^e \sim \mathcal{N}(0, I_q), \qquad
X^e_a = \phi_a^{\top} H^e + \varepsilon^e_a,\ \forall a\in A, \qquad
X^e_{j} = \theta^\top_j X^e_A + \phi_j^\top H^e + \varepsilon^e_j, \ \forall j\in S^\star.
$$
The response and its descendants are generated according to
$$
Y^e = \beta^{\star\top}_{S^\star} X_{S^\star}^e + \phi_Y^{\top} H^e + \varepsilon^e_Y, \quad
X_d^e = \theta_d Y^e + \phi_d^\top H^e + \varepsilon^e_d, \ \forall d\in D.
$$
The noise variables satisfy
$$
\varepsilon^e_a \sim \mathcal{N}(0, \nu_a^e),\ \forall a\in A, \quad
\varepsilon^e_j \sim \mathcal{N}(0, \nu_j^e),\ \forall j\in S^\star, \quad
\varepsilon^e_Y \sim \mathcal{N}(0, \nu_Y), \quad
\varepsilon^e_d \sim \mathcal{N}(0, \nu_d^e), \ \forall d\in D
$$
where the noise variables are mutually independent and independent of $H^e$. Across environments, interventions are allowed to change only the noise levels $\nu_j^e$ of the observed covariates $j\in A\cup S^\star\cup D$, while all structural coefficients, the distribution of $H^e$, and the noise distribution of the response remain invariant. Thus, the environments induce heterogeneity in the covariate distributions without directly intervening on the response mechanism. Such models with noise interventions provide a standard framework for studying causal identification in the presence of latent confounding and have been widely considered in the literature~\citep{rothenhausler2019causal, shen2026causality}.

Since the latent variables $H^e$ are unobserved, the LSCR invariance principle can no longer be expected to identify the structural coefficient $\beta^\star$. In particular, when $S=S^\star$, the population OLS coefficient is given by
$$
\beta^{e,S^\star}_{S^\star} = \beta^\star_{S^\star} + \left(\Sigma^{e}_{S^\star} \right)^{-1} \mathbb{E}\left[X^e_{S^\star}H^{e\top}\right]\phi_Y.
$$
The second term represents the bias induced by latent confounding and, in general, depends on the interventions across environments. Consequently, the OLS coefficient on the true causal parent set need not remain invariant across environments, even though the structural coefficient $\beta^\star$ itself is invariant. This motivates the need for an alternative invariance principle that can accommodate hidden confounding.

\vspace{0.1in}
\begin{definition}[Hidden-invariant prediction model]
\label{def:hidden-invariant}
We say that $\beta\in\mathbb{R}^p$ is a hidden invariant prediction model if
$$
\mathbb{E}\left[ X^e_j (Y^e - X^e \beta) \right] =
\mathbb{E}\left[ X^f_j (Y^f - X^f \beta) \right], \quad \forall j\in[p],\ \forall e,f\in\mathcal{E}.
$$
\end{definition}

A prediction model is invariant under hidden confounding, based on Definition~\ref{def:hidden-invariant}, if the cross moments between its residual and each covariate remain constant across environments. In particular, for the true structural coefficient $\beta^\star$, the cross moments $\mathbb E[X_j^e(Y^e-X^{e\top}\beta^\star)]$ remain invariant across environments for all $j\in[p]$. This follows because the structural coefficients, the distribution of the latent variables, and the response noise distribution remain unchanged across environments, while interventions only affect the noise variances of the covariates. The same invariance principle is also considered in~\cite{rothenhausler2019causal}.
Motivated by this observation, we define the invariance score under hidden confounding as
$$
\mathsf{I}_h(\beta)= \sum_{e<f} \omega^{e,f} \sum_{j=1}^p
\left( \mathbb{E}\left[ X_j^e\left(Y^e-X^{e\top}\beta\right) \right] -
\mathbb{E}\left[ X_j^f\left(Y^f-X^{f\top}\beta\right) \right] \right)^2.
$$
The score quantifies the degree to which the moments between the residuals and covariates vary across environments. In particular, $\mathsf{I}_h(\beta)=0$ if and only if these moments are invariant across environments, so that $\beta$ satisfies the proposed invariance principle under hidden confounding.

With this newly defined invariance score, we propose a procedure with invariance regularization that estimates $\beta^\star$ in the same spirit as~\eqref{eq:pop-task1}:
\begin{align*} 
\begin{split}
    \min_{\substack{ \beta \in \mathbb{R}^p  } }  \quad \sum_{e\in\mathcal{E}} \omega^e\ \mathbb{E} \left[ Y^e - \beta^\top X^e \right]^2 + \gamma \mathsf{I}_h (\beta).
\end{split}
\end{align*}
Furthermore, the finite-sample convex MIP formulation is given by:
\begin{subequations}
\begin{align*}
     \min_{
     \substack{ \beta \in\mathbb{R}^p   \\ 
     g \in \{0,1\}^p } 
     } \quad& \sum_{e\in\mathcal{E}} \omega^e\left\| Y^e - \sum_{k=1}^p X_k^e \beta_k \right\|_{n^e}^2 
     + \lambda \sum_{j=1}^p g_j  \\
     & \qquad + \gamma \sum_{e<f }  \omega^{e,f} \sum_{j=1}^p \left( \Big\langle X_j^e, Y^e - \sum_{k=1}^p X_k^e \beta_k \Big\rangle - \Big\langle X_j^f, Y^f - \sum_{k=1}^p X_k^f \beta_k \Big\rangle \right)^2  \\
\text{subject to:} \quad& \beta_j = 0 \text{ if } g_j = 0, \qquad\forall j\in [p].
\end{align*}
\label{eq:mip-hidden} 
\end{subequations}
where a penalty $\sum_{j=1}^p g_j$ is introduced to encourage sparsity.

\vspace{0.1in}
\begin{example}
\label{ex:hidden}
We generate data according to the noise-intervention model introduced in this section, with $q=2$, $|A|\in\{2,3\} $, $|S^\star|=4$, $|D|\in\{2,3\}$, and eight additional independent covariates. Figure~\ref{fig:hidden} shows that our proposed method, MIP-hidden, achieves the lowest coefficient estimation error and the highest rate of perfect recovery. We focus on comparisons with Causal Dantzig~\citep{rothenhausler2019causal} and negDRO~\citep{wang2024causal}, both of which are also designed to accommodate hidden confounding.
One possible explanation for the improved performance of MIP-hidden is that its objective includes both the MSE loss and an $\ell_0$-type sparsity penalty, in addition to the invariance score. In contrast, Causal Dantzig does not directly account for prediction loss, while negDRO does not explicitly encourage sparsity. The mixed-integer formulation therefore allows MIP-hidden to perform discrete variable selection while also accounting for prediction accuracy, particularly when the underlying causal model is sparse.

\begin{figure}[!ht]
\centering
\includegraphics[width=0.85\textwidth]{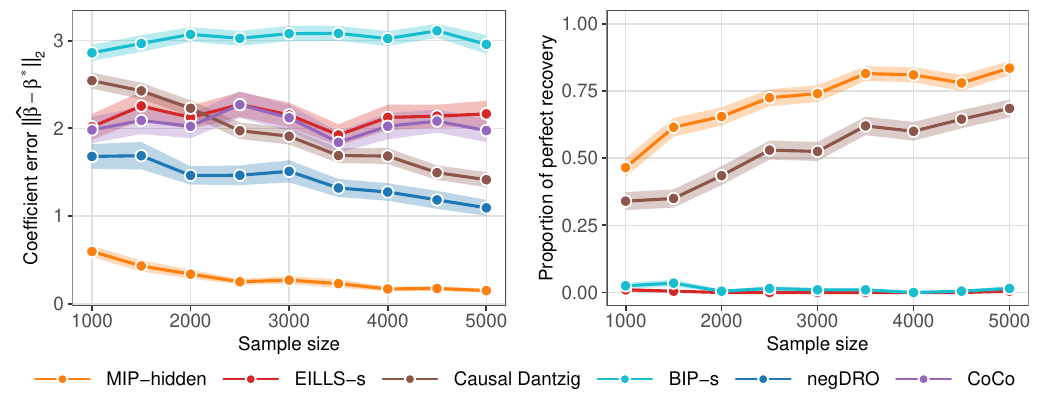}
\caption{\small Performance comparison under hidden confounding and noise interventions (Example~\ref{ex:hidden}) as the sample size $n$ increases. Left panel: average coefficient estimation error $\|\widehat{\beta}-\beta^\star\|_2$ over 200 trials. 
Right panel: proportion of 200 trials achieving perfect recovery of $S^\star$, reported for methods that perform variable selection. Solid lines show averages, and shaded bands indicate $\pm 1$ standard error. Additional competing methods and full experimental details are provided in Appendix~\ref{sec:app:figure-hidden}.}
\label{fig:hidden}
\end{figure}

\end{example}

\section{Proofs of results in Section~\ref{sec:task1}}
\label{sec:app:proof:task1}

In this section, we provide proofs of results regarding the procedure~\eqref{eq:pop-task1}. To facilitate the proof, we define the objective function as $\mathsf{F}(\beta)$, and
$$
\mathsf{R}(\beta) := \sum_{e\in\mathcal{E}} \omega^e\ \mathbb{E}\left[Y^e - \beta^\top X^e \right]^2, \ \
\mathsf{I}_1(S) := \sum_{e<f} \omega^{e,f} \left\|\beta^{e,S} - \beta^{f,S} \right\|_2^2, \ \
\mathsf{I}_2(S) := \sum_{e<f} \omega^{e,f} \left\|\nu^{e,S} - \nu^{f,S} \right\|_2^2.
$$
% We rewrite $\gamma_1 = \gamma$ and $\gamma_2 = \gamma\rho$.
% Note that for any given coefficient vector $\beta$, the subset $S$ is uniquely determined as $S = \mathrm{supp}(\beta)$.
As a result, $\mathsf{I}(S) = \mathsf{I}_1(S) + \rho \mathsf{I}_2(S)$ and $\mathsf{F}(\beta) = \mathsf{R}(\beta) + \gamma \mathsf{I}_1(S) + \gamma \rho \mathsf{I}_2(S)$. We also use $\widehat{\mathsf{F}}(\beta)$, 
$\widehat{\mathsf{R}}(\beta)$, $\widehat{\mathsf{I}}_1(S)$ and $\widehat{\mathsf{I}}_2(S)$ to denote the corresponding finite-sample quantities. 
We use $X_S^{e}\in\mathbb{R}^{|S|}$ to denote the subvector of $X^e$ indexed by $S$, and $X_i^e\in\mathbb{R}^p$ to denote the $i$-th observation of $X^e$. Similarly, we use $Y^e_i\in\mathbb{R}$ to denote the $i$-th observation of $Y^e$. When needed, we write $\boldsymbol{X}^e\in\mathbb{R}^{n^e\times p}$ for the corresponding sample matrix, and $\boldsymbol{Y}^e\in\mathbb{R}^{n^e}$ for the corresponding sample vector. We further define the covariance matrix restricted to $S$ as
$ \Sigma^{e}_S:=\mathbb{E}\!\left[
X^{e}_SX^{e\top}_S\right] $,
and denote its empirical counterpart by $\widehat{\Sigma}^{e}_S$. Finally, for the population OLS coefficient $\beta^{e,S}$ defined in \eqref{eqn:beta_r}, we denote its restriction to the coordinates in $S$ by ${\beta}_S^{e,S}\in\mathbb{R}^{|S|}$; its empirical counterpart is denoted by $\widehat{\beta}^{e,S}_S \in \mathbb{R}^{|S|}$.

\vspace{0.1in}
\subsection{Proof of Lemma~\ref{lem:identification}}
\label{sec:app:proof:identification}
Before the proof, we restate \textbf{Lemma~\ref{lem:identification}} in a more detailed form.

\begin{lemma}
\label{lem:identification2}
Suppose Assumption~\ref{ass:identifiability} and \ref{assump:regular} hold. Denote
$$
\gamma^\star := \sup_{ S: S \cap G \neq \varnothing } \kappa_L^{-1} \cdot \frac{ \left\| \sum_{e\in\mathcal{E}}  \omega^e\ \mathbb{E} [\varepsilon^e X^e_S] \right\|_2^2 }{ \mathsf{I}_1(S) + \mathsf{I}_2(S) }.
$$
When $\gamma \wedge \gamma \rho \geq 2\gamma^\star + \epsilon$ for some $\epsilon >0$, 
$$
\mathsf{F}(\beta) - \mathsf{F}(\beta^\star) \geq \frac{1}{2} \kappa_L \left\| \beta - \beta^\star \right\|_2^2 + \epsilon\cdot \left( \mathsf{I}_1(S) + \mathsf{I}_2(S) \right) .
$$
\end{lemma}

\begin{proof}
We largely follow the proof idea of Theorem A.1 in~\cite{fan2024environment}. First, denote $\Delta^\beta = \beta - \beta^\star$, $S = \mathrm{supp}(\beta)$ and $\bar S = S \cup S^\star$. Then for any $\delta \in (0,1)$,
\begin{align*}
\begin{split}
    \mathsf{R}(\beta) - \mathsf{R}(\beta^\star) 
    &= \sum_{e\in\mathcal{E}} \omega^e\left( \mathbb{E}\left[Y^e - \beta^\top X^e \right]^2 - \mathbb{E}\left[Y^e - \beta^{\star\top} X^e \right]^2 \right) \\
    &= \sum_{e\in\mathcal{E}} \omega^e\left( \mathbb{E}\left[ (\beta^\star - \beta)^\top X^e + \varepsilon^e \right]^2 - \mathbb{E}\left[ \varepsilon^e \right]^2 \right) 
    = \sum_{e\in\mathcal{E}} \omega^e\left(  (\Delta^\beta)^\top \Sigma^e \Delta^\beta - 2 (\Delta^\beta)^\top \mathbb{E} \left[\varepsilon^e X^e\right] \right) \\
    &= \sum_{e\in\mathcal{E}} \omega^e (\Delta^{\beta}_{\bar S})^\top \Sigma^{e}_{\bar S} \Delta^{\beta}_{\bar S} -2 (\Delta^{\beta}_{\bar S})^\top  \cdot  \sum_{e\in\mathcal{E}} \omega^e  \mathbb{E} \left[\varepsilon^e X^{e}_{\bar S}\right]  \\
    % &= \sum_{e\in\mathcal{E}} \omega^e (\Delta^{\beta}_{\bar S})^\top \Sigma^{e}_{\bar S} \Delta^{\beta}_{\bar S} - \sum_{e\in\mathcal{E}} 2 \omega^e(\Delta^{\beta}_{\bar S})^\top (\Sigma^{e}_{\bar S})^{1/2} (\Sigma^{e}_{\bar S})^{-1/2} \cdot \mathbb{E} \left[\varepsilon^e X^{e}_{\bar S}\right]  \\
    % &\geq (1 - \delta) \sum_{e\in\mathcal{E}} \omega^e(\Delta^{\beta}_{\bar S})^\top \Sigma^{e}_{\bar S}  \Delta^{\beta}_{\bar S}  - \delta^{-1}  \left( \mathbb{E} \left[\varepsilon^e X^{e}_{\bar S} \right] \right)^\top (\Sigma^{e}_{\bar S})^{-1} \cdot \sum_{e\in\mathcal{E}} \omega^e  \mathbb{E} \left[\varepsilon^e X^{e}_{\bar S}\right] \\
    &\geq (1-\delta) \kappa_L  \left\| \beta - \beta^\star \right\|_2^2 - \delta^{-1} \kappa_L^{-1} \left\| \sum_{e\in\mathcal{E}} \omega^e\ \mathbb{E} [ \varepsilon^e X_{S}^{e} ] \right\|_2^2.
\end{split}
\end{align*}
Consequently, by taking $\delta = 1/2$,
\begin{align*}
\begin{split}
    &\mathsf{F}(\beta) - \mathsf{F}(\beta^\star) \\
    ={}& \mathsf{R}(\beta) - \mathsf{R}(\beta^\star) + \gamma \mathsf{I}_1(S) + \gamma\rho \mathsf{I}_2(S) \\
    \geq{}&  \frac{\kappa_L}{2}  \left\| \beta - \beta^\star \right\|_2^2 - 2 \kappa_L^{-1 } \left\| \sum_{e\in\mathcal{E}} \omega^e\ \mathbb{E} [ \varepsilon^e X^{e}_S ] \right\|_2^2 + 2\gamma^\star \left( \mathsf{I}_1(S) + \mathsf{I}_2(S) \right) +  (\gamma - 2 \gamma^\star) \mathsf{I}_1(S) + (\gamma\rho - 2 \gamma^\star) \mathsf{I}_2(S) \\
    \geq{}&  \frac{\kappa_L}{2}  \left\| \beta - \beta^\star \right\|_2^2 + \epsilon \cdot \left( \mathsf{I}_1(S) + \mathsf{I}_2(S) \right),
\end{split}
\end{align*}
where the last inequality follows from the definition of $\gamma^\star$, as well as the conditions $\gamma \wedge \gamma\rho \ge 2 \gamma^\star + \epsilon$.

\end{proof}

\vspace{0.1in}
\subsection{Preliminary results}
\label{sec:app:proof:prelim}

In this section, we present several events with high probabilities that relate the population quantities to their empirical counterparts.

\begin{lemma} 
\label{lem:D-R}
Define the event 
$$
\mathcal{T}_1=\left\{ \max_{ \beta\in \mathcal{B}(s_n) } \left| \mathsf{R}(\beta) - \mathsf{R}(\beta^\star) - \widehat{\mathsf{R}}(\beta) + \widehat{\mathsf{R}}(\beta^\star) \right|\leq (C_0 \delta_n + \frac{\kappa_L}{4}) \left\| \beta - \beta^\star \right\|_2^2 + \frac{ (C_0' \delta_n)^2 }{\kappa_L} 
\right\},
$$
where $C_0, C_0' >0$ are constants, $\mathcal{B}(s_n) = \{\beta\in\mathbb{R}^p :|\mathrm{supp}(\beta)| \leq s_n \}$, and 
$$
\delta_n = \sqrt{ \frac{2s_n + s_n\log p + \log p}{\sum_e n^e} } + \frac{2 s_n + s_n\log p + \log p}{\sum_e n^e} .
$$
Then $\mathbb{P}(\mathcal{T}_1) \geq 1 - 1/p$.
\end{lemma}
\begin{proof}
    Appealing to Lemma C.7 of \cite{fan2024environment}, where we let
    $$
    V(s_n, t) = s_n \log(e p / s_n) + s^\star + t \leq s_n (1+\log p) + s_n + t,
    $$
    and let $t = \log p$, we obtain that
    $$
    \max_{\beta\in\mathcal{B}(s_n) } \left| \mathsf{R}(\beta) - \mathsf{R}(\beta^\star) - \widehat{\mathsf{R}}(\beta) + \widehat{\mathsf{R}}(\beta^\star) \right|\leq \left( C_0\delta_n \left\| \beta - \beta^\star \right\|_2^2 + C_0' \delta_n \left\|\beta - \beta^\star \right\|_2 \right).
    $$
    with probability at least $1-1/p$.
    The constants $C_0, C_0'>0$ are only associated with $\kappa_U$, $\sigma_x$ and $\sigma_\varepsilon$.
    Now applying Young's inequality yields
    \begin{align*}
    \begin{split}
       & C_0 \delta_n \left\| \beta - \beta^\star \right\|_2^2 + C_0' \delta_n \left\|\beta - \beta^\star \right\|_2
       \leq C_0 \delta_n \left\| \beta - \beta^\star \right\|_2^2 + \frac{(C_0' \delta_n)^2}{4} \frac{4}{\kappa_L} + \frac{\kappa_L}{4} \left\|\beta - \beta^\star \right\|_2^2 
       \\
       ={}& (C_0 \delta_n + \frac{\kappa_L}{4}) \left\| \beta - \beta^\star \right\|_2^2 + \frac{ (C_0' \delta_n)^2 }{\kappa_L}.
    \end{split}
    \end{align*}
    The result hence follows.
    
\end{proof}

\vspace{0.2in}
\begin{lemma}
    Define the event
    $$
    \mathcal{T}_2 =\left\{ \max_{ \substack{ |S|\leq s_n \\ e\in\mathcal{E} } } \left\| \widehat{\Sigma}^{e}_S - \Sigma^{e}_S \right\|_{\rm op} \lesssim \sqrt{ \frac{s_n \log p + \log |\mathcal{E}|}{\min_{e\in\mathcal{E}} n^e} }   \right\},
    $$
    where $\Sigma^e_S = \mathbb{E}[X^e_S X^{e\top}_S]$.
    Then $\mathbb{P}(\mathcal{T}_2) \geq 1 - 2/p$.
\end{lemma}
\begin{proof}
    Recall that we assume $\kappa_L I_p \preceq \Sigma^{e} \preceq \kappa_U I_p $ in Assumption~\ref{assump:regular} part (1). Then for any $S \subseteq [p]$, we also have $\kappa_L I_{|S|} \preceq \Sigma^{e}_S \preceq \kappa_U I_{|S|} $. Since
    $$
    \widehat{\Sigma}^{e}_S = \frac{1}{n^e} \sum_{i=1}^{n^e} X_{S,i}^{e} X_{S,i}^{e\top} 
    \quad \text{ with } \quad 
    \mathbb{E}\left[ \widehat{\Sigma}^{e}_S \right] = \Sigma^{e}_S,
    $$
    using the Bernstein inequality for sub-exponential distributions, 
    we have
    $$
    \left\| \widehat{\Sigma}^{e}_S - \Sigma^{e}_S \right\|_{\rm op} \leq C_1 \sigma^2_x \kappa_U \left( \sqrt{\frac{|S| + t}{n^e}} + \frac{|S| + t}{n^e}\right)
    $$
    with probability at least $1 - 2 e^{-c_1t}$, where $C_1,c_1>0$ are absolute constants. In other words, with probability at least $1 - 2e^{-t} / \sum_{j=0}^{|S|}\binom{p}{j} |\mathcal{E}|$,
    $$
    \left\| \widehat{\Sigma}^{e}_S - \Sigma^{e}_S \right\|_{\rm op} \lesssim \sqrt{\frac{|S| + \log \sum_{j=0}^{|S|} \binom{p}{j} + \log |\mathcal{E}| + t }{n^e}}
    \lesssim \sqrt{ \frac{|S| \log(ep / |S| ) + \log |\mathcal{E}| + t }{n^e} }.
    $$
    Now applying the union bound over all supports for $S$ with $|S| \leq s_n$ and all $|\mathcal{E}|$ environments, and letting $t = \log p$, we have
    $$
    \mathbb{P}(\mathcal{T}_2) \geq 1 - 2/p.
    $$

\end{proof}

\vspace{0.2in}
\begin{lemma}\label{lem:concen-Xr}
    Define the event
    $$
    \mathcal{T}_3 = \left\{ \max_{ \substack{ |S|\leq s_n \\ e\in\mathcal{E} } } \left\|\frac{1}{n^e}  \sum_{i=1}^{n^e} X^{e}_{S,i}\cdot r_i^{e,S}  \right\|_2 \lesssim (1 + \|\beta - \beta^\star\|_2)   \sqrt{\frac{s_n \log p + \log |\mathcal{E}| }{ \min_{e\in\mathcal{E}} n^e }}, \quad \text{with } \mathrm{supp}(\beta) \subseteq S  \right\}, 
    $$
    where $r^{e,S} = Y^e - X^{e\top} {\beta}^{e,S}$.
    Then $\mathbb{P}(\mathcal{T}_3) \geq 1 - 2/p$.
\end{lemma}
\begin{proof}
    Denote $Z_i^{e,S} = X^{e}_{S,i}  r_i^{e,S}$. Due to the property of the population OLS coefficient ${\beta}^{e,S}$, we have
    $\mathbb{E}[Z_i^{e,S}] = 0$. Now using the Bernstein inequality for sub-exponential distributions, we have
    $$
    \left\| \frac{1}{n^e} \sum_{i=1}^{n^e} Z_i^{e,S} \right\|_2 \leq C_2 K^{e,S} \left( \sqrt{\frac{|S| +t}{n^e}} + \frac{|S| + t}{n^e}  \right)
    $$
    with probability at least $1 - 2 e^{-c_2 t}$, where $C_2, c_2 > 0$ are absolute constants. However, the number $K^{e,S}$ depends on $\beta^\star$. To see this, note that
    $$
    K^{e,S} = \sup_{ \|u\|_2 = 1 } \left\| u^\top Z^{e,S} \right\|_{\psi_1} \leq \sup_{\|u\|_2 = 1} \left\| u^\top X^{e}_S \right\|_{\psi_2} \left\| r^{e,S} \right\|_{\psi_2} 
    \lesssim \sigma_x \sqrt{\kappa_U} \left\| r^{e,S} \right\|_{\psi_2} ,
    $$
    where the last inequality follows from Assumption~\ref{assump:regular} part (2). Moreover, 
    for any $\beta$ with $\mathrm{supp}(\beta) \subseteq S$,
    $$
    r^{e,S} = Y^{e} - X^{e\top}_S {\beta}_S^{e,S}
    = X^{e\top} \left( \beta^\star - \beta \right) + \varepsilon^e + X^{e\top}_S \left( \beta_S - {\beta}_S^{e,S} \right);
    $$
    Now let $q^{e,\beta} = X^{e\top} \left( \beta^\star - \beta \right) + \varepsilon^e$, with the population normal equation, we have
    $$
    0 = \mathbb{E}\left[ X^{e}_S r^{e,S} \right] 
    = \mathbb{E} \left[ X^{e}_S q^{e,\beta} \right] + \Sigma^{e}_S \left( \beta_S - {\beta}_S^{e,S} \right).
    $$
    As a result, 
    \begin{align*}
    \begin{split}
        \left\| \beta_S - {\beta}_S^{e,S} \right\|_2
    & = \left\| (\Sigma^{e}_S)^{-1} \mathbb{E}\left[X^{e}_S q^{e,\beta} \right]  \right\|_2
    \leq \frac{1}{\kappa_L} \sup_{\|u\|_2=1} \left| u^\top \mathbb{E}[X^{e}_S q^{e,\beta}] \right| \\
    &\leq \frac{1}{\kappa_L} \sup_{\|u\|_2=1}  \sqrt{ \mathbb{E}\left[( u^\top X^{e}_S)^2 \right] } \sqrt{ \mathbb{E}(q^{e,\beta})^2}
    \lesssim \frac{\sqrt{\kappa_U}}{\kappa_L} \left\| q^{e,\beta} \right\|_{\psi_2} \\
    &\leq \frac{\sqrt{\kappa_U}}{\kappa_L} \left( \left\| X^{e\top}(\beta^\star - \beta ) \right\|_{\psi_2} + \left\| \varepsilon^e \right\|_{\psi_2} \right) 
    \lesssim \frac{\sqrt{\kappa_U}}{\kappa_L} \left( \sigma_x \sqrt{\kappa_U} \|\beta - \beta^\star \|_2 + \sigma_\epsilon  \right),
    \end{split}
    \end{align*}
    where the inequalities follow from Assumption~\ref{assump:regular}.
    Consequently, we have
    \begin{align*}
    \begin{split}
    K^{e,S} 
    &\lesssim  \left\| r^{e,S} \right\|_{\psi_2} 
    \leq \left\| q^{e,\beta} \right\|_{\psi_2} + \left\| X^{e\top}_S \left( \beta_S - \beta^{e,S}_S \right) \right\|_{\psi_2} \\
    & \leq \left( \sigma_x \sqrt{\kappa_U} \|\beta - \beta^\star \|_2 + \sigma_\epsilon  \right) + \sigma_x \sqrt{\kappa_U} \left\| \beta_S - \beta^{e,S}_S \right\|_2 
    \lesssim 1 + \left\| \beta - \beta^\star \right\|_2.
    \end{split}
    \end{align*}
    In other words, with probability at least $1 - 2 e^{-t} / \sum_{j=0}^{|S|}\binom{p}{j} |\mathcal{E}| $,
    $$
    \left\| \frac{1}{n^e} \sum_{i=1}^{n^e} Z_i^{e,S} \right\|_2 \lesssim (1 + \|\beta - \beta^\star\|_2)  \sqrt{ \frac{|S| \log(ep / |S| ) + \log |\mathcal{E}| + t }{n^e} }.
    $$
    Now applying the union bound over all supports for $S$ with $|S| \leq s_n$ and all $|\mathcal{E}|$ environments, and letting $t = \log p$, we have
    $$
    \mathbb{P}(\mathcal{T}_3) \geq 1 - 2/p.
    $$

\end{proof}

\vspace{0.2in}
\begin{lemma}
    Define the event
    $$
    \mathcal{T}_4 = \left\{ \max_{\substack{ |S| \leq s_n \\ e\in\mathcal{E}  } } \left| \frac{1}{n^e} \sum_{i=1}^{n^e} (r_i^{e,S})^2 - \mathbb{E}\left[r^{e,S}\right]^2 \right| 
    \lesssim (1 + \|\beta - \beta^\star\|_2)^2 \sqrt{\frac{s_n \log p + \log |\mathcal{E}| }{ \min_{e\in\mathcal{E}} n^e }}, \quad \text{with } \mathrm{supp}(\beta) \subseteq S \right\},
    $$
    where $r^{e,S} = Y^e - X^{e\top} {\beta}^{e,S}$.
    Then $\mathbb{P}(\mathcal{T}_4) \geq 1 - 2/p.$
\end{lemma}
\begin{proof}
    Recall that in the proof of Lemma \ref{lem:concen-Xr}, we showed that $r^{e,S}$ is sub-Gaussian with $\left\| r^{e,S}\right\|_{\psi_2} \lesssim 1 + \|\beta - \beta^\star\|_2$. Then $(r^{e,S})^2 - \mathbb{E}[r^{e,S}]^2$ is sub-exponential with $\left\| (r^{e,S})^2 - \mathbb{E}[r^{e,S}]^2 \right\|_{\psi_1} \lesssim (1 + \|\beta - \beta^*\|_2)^2$. By Bernstein’s inequality, we have
    $$
    \left| \frac{1}{n^e} \sum_{i=1}^{n^e} (r_i^{e,S})^2 - \mathbb{E}[r^{e,S}]^2 \right| \leq C_3 (K^{e,S})^2 \left( \sqrt{\frac{t}{n^e}} + \frac{t}{n^e} \right)
    $$
    with probability at least $1 - e^{-c_3 t}$, where $C_3, c_3 >0$ are absolute constants, and $K^{e,S} \lesssim  1 + \|\beta - \beta^\star\|_2$. In other words, with probability at least $1 - 2 e^{-t} / \sum_{j=0}^{|S|}\binom{p}{j} |\mathcal{E}| $,
    $$
    \left| \frac{1}{n^e} \sum_{i=1}^{n^e} (r_i^{e,S})^2 - \mathbb{E}[r^{e,S}]^2 \right| \lesssim (1 + \|\beta - \beta^\star\|_2)^2 \sqrt{ \frac{|S| \log(ep / |S| ) + \log |\mathcal{E}| + t }{n^e} }.
    $$
    Now applying the union bound over all supports for $S$ with $|S| \leq s_n$ and all $|\mathcal{E}|$ environments, and letting $t = \log p$, we have
    $$
    \mathbb{P}(\mathcal{T}_4) \geq 1 - 2/p.
    $$
\end{proof}

\vspace{0.2in}
\begin{lemma} \label{lem:D-I}
Under the event $\mathcal{T}_2 \cap \mathcal{T}_3 \cap \mathcal{T}_4$, for any $S$ with $|S|\leq s_n$ and $\mathrm{supp}(\beta) \subseteq S$,
\begin{align*}
\begin{split}
   & \left| \mathsf{I}(S) - \mathsf{I}(S^\star) - \widehat{\mathsf{I}}(S) + \widehat{\mathsf{I}}(S^\star) \right|  \\
    \leq{}& \epsilon / (2\gamma) \cdot \left( \mathsf{I}_1(S) + \mathsf{I}_2(S) \right) + ((16+64\rho^2)\gamma/\epsilon + 16 + 64\rho) r_n^2 \\
    &+(16\gamma/\epsilon + 8) r_n^2 \cdot \|\beta - \beta^\star \|_2^2 
    + (64\rho^2\gamma/\epsilon + 32\rho) r_n^2 \cdot \|\beta - \beta^\star\|_2^4,
\end{split}
\end{align*}
where 
$r_n \asymp \sqrt{(s_n \log p + \log |\mathcal{E}|) / \min_{e\in\mathcal{E}} n^e}$, and $\epsilon>0$ is defined in Lemma~\ref{lem:identification}.
\end{lemma}
\begin{proof}
% Recall that $\mathsf{I}_1(\beta) = \sum_{e,f} \omega^{e,f} \left\|\beta^{e,S} - \beta^{f,S} \right\|_2^2$ and $\mathsf{I}_2(\beta) = \sum_{e,f} \omega^{e,f} \left|\sigma^{2,e,S} - \sigma^{2,f,S} \right|^2$.
    Recall $\mathsf{I}(S) = \mathsf{I}_1(S) + \rho \mathsf{I}_2(S)$. Therefore, it suffices to bound 
    $$
    \left| \mathsf{I}_1(S) - \mathsf{I}_1(S^\star) - \widehat{\mathsf{I}}_1(S) + \widehat{\mathsf{I}}_1(S^\star) \right| \quad\text{ and }\quad
    \left| \mathsf{I}_2(S) - \mathsf{I}_2(S^\star) - \widehat{\mathsf{I}}_2(S) + \widehat{\mathsf{I}}_2(S^\star) \right|,
    $$
    % Before we proceed, denote $r_n = \sqrt{(s_n \log p + \log |\mathcal{E}|) / \min_{e\in\mathcal{E}} n^e}.$
    which are established in Steps I and II of this proof. The final bound is obtained in Step III.
    
    \vspace{0.1in}
    \textit{\textbf{Step I.}} Bound $\left| \mathsf{I}_1(S) - \mathsf{I}_1(S^\star) - \widehat{\mathsf{I}}_1(S) + \widehat{\mathsf{I}}_1(S^\star) \right|$. 
    For any $e\in\mathcal{E}$, the sample error of OLS coefficients satisfies
    $$
    \widehat{\beta}^{e,S}_S - {\beta}_S^{e,S} = (\widehat{\Sigma}^{e}_S)^{-1} \left( \frac{1}{n^e} \sum_{i=1}^{n^e} X^{e}_{S,i} r^{e,S}_i \right),
    $$
    and under event $\mathcal{T}_2 \cap \mathcal{T}_3$,
    \begin{align} \label{eq:diff-in-beta}
        \left\| \widehat{\beta}^{e,S}_S - {\beta}_S^{e,S} \right\|_2 \leq (1 + \| \beta - \beta^\star \|_2) \cdot r_n
    \end{align}
    for any $\beta$ with $\mathrm{supp}(\beta) \subseteq S$.
    Now denote $\Delta^{e,f,S} := \beta^{e,S} - \beta^{f,S}$, $\zeta^{e,S} := \widehat{\beta}^{e,S} - \beta^{e,S}$, and $D_1(S) := \sum_{e<f} \omega^{e,f} \left\| \zeta^{e,S} - \zeta^{f,S} \right\|_2^2$. Then
    \begin{align*}
    \begin{split}
    \left| \widehat{\mathsf{I}}_1(S) - \mathsf{I}_1(S) \right|
    &= \left| \sum_{e<f} \omega^{e,f} \left\| \Delta^{e,f,S} + (\zeta^{e,S} - \zeta^{f,S}) \right\|_2^2 - \sum_{e<f} \omega^{e,f} \left\| \Delta^{e,f,S} \right\|_2^2 \right| \\
    &= \left| 2 \sum_{e<f} \omega^{e,f} \langle \Delta^{e,f,S} , \zeta^{e,S} - \zeta^{f,S} \rangle + \sum_{e<f} \omega^{e,f} \left\| \zeta^{e,S} - \zeta^{f,S} \right\|_2^2 \right| \\
    &\leq 2 \left( \sum_{e<f} \omega^{e,f} \left\| \Delta^{e,f,S} \right\|_2^2 \right)^{1/2} \left( \sum_{e<f} \omega^{e,f} \left\| \zeta^{e,S} - \zeta^{f,S} \right\|_2^2 \right)^{1/2} + \sum_{e<f} \omega^{e,f} \left\| \zeta^{e,S} - \zeta^{f,S} \right\|_2^2 \\
    &\leq 2 \sqrt{\mathsf{I}_1(S)} \sqrt{D_1(S)} + D_1 (S)  
    \leq \epsilon / (2\gamma) \cdot \mathsf{I}_1(S) + (2\gamma /\epsilon + 1) \cdot D_1(S) .
    \end{split}
    \end{align*}
    Note that
    $$
    D_1(S) = \sum_{e<f} \omega^{e,f} \left\| (\widehat{\beta}^{e,S}- \beta^{e,S}) - (\widehat{\beta}^{f,S}- \beta^{f,S} )\right\|_2^2
    \leq 4 \max_e \left\| \widehat{\beta}^{e,S} - \beta^{e,S} \right\|_2^2 
    \leq 8 r_n^2 + 8 \left\|\beta - \beta^\star \right\|_2^2 r_n^2,
    $$
    where the last inequality follows from~\eqref{eq:diff-in-beta}.
    Moreover,
    $$
    \widehat{\mathsf{I}}_1(S^\star) - \mathsf{I}_1(S^\star) 
    = \widehat{\mathsf{I}}_1(S^\star)
    = D_1(S^\star)  \leq 8r_n ^2 ,
    $$
    by taking $\beta = \beta^\star$ under the true support $S^\star$.
    Putting them together, we have
    \begin{align*}
    \begin{split}
    & \left| \mathsf{I}_1(S) - \mathsf{I}_1(S^\star) - \widehat{\mathsf{I}}_1(S) + \widehat{\mathsf{I}}_1(S^\star) \right| \\
    \leq{}& \epsilon/(2\gamma) \cdot \mathsf{I}_1(S) + (2\gamma /\epsilon + 1) D_1(S)  + D_1(S^\star) 
    \leq \epsilon/(2\gamma) \cdot \mathsf{I}_1(S) +  (16\gamma/\epsilon + 8) (r_n^2 + \left\| \beta - \beta^\star \right\|_2^2 r_n^2) + 8 r_n^2 ,
    \end{split}
    \end{align*}
    for any $\beta\in\mathbb{R}^p$ with $\mathrm{supp}(\beta) = S$.

    \vspace{0.1in}
    \textit{\textbf{Step II.}} Bound $\left| \mathsf{I}_2(S) - \mathsf{I}_2(S^\star) - \widehat{\mathsf{I}}_2(S) + \widehat{\mathsf{I}}_2(S^\star) \right|$. For any $e\in\mathcal{E}$, the sample variances of the OLS residuals satisfy
    $$
    \widehat{\nu}^{ e, S} - \nu^{e,S} = \left[ \frac{1}{n^e} \sum_{i=1}^{n^e} (r_i^{e,S})^2 - \mathbb{E}\left[r^{e,S}\right]^2 \right]
    - \left( \frac{1}{n^e} \sum_{i=1}^{n^e} X^{e}_{S,i} r^{e,S}_i \right)^\top \left( \widehat{\Sigma}^{e}_S \right)^{-1} \left( \frac{1}{n^e} \sum_{i=1}^{n^e} X^{e}_{S,i} r^{e,S}_i \right),
    $$
    and under $\mathcal{T}_2 \cap \mathcal{T}_3 \cap \mathcal{T}_4$, 
    \begin{align} \label{eq:diff-in-sigma}
        \left| \widehat{\nu}^{e, S} - \nu^{e,S} \right| 
    \le (1 + \|\beta - \beta^\star\|_2)^2 \cdot r_n
    \end{align}
    for any $\beta$ with $\mathrm{supp}(\beta) \subseteq S$.
    Now denote $\widetilde{\Delta}^{e,f,S} := \nu^{e,S} - \nu^{f,S}$, $\widetilde{\zeta}^{e,S} := \widehat{\nu}^{e,S} - \nu^{e,S}$, and $D_2(S) := \sum_{e<f} \omega^{e,f} \left| \widetilde{\zeta}^{e,S} - \widetilde{\zeta}^{f,S} \right|^2$. Then as in Step I,
    \begin{align*}
    \begin{split}
    \left| \widehat{\mathsf{I}}_2(S) - \mathsf{I}_2(S)  \right|
    \leq 2 \sqrt{\mathsf{I}_2(S)} \sqrt{D_2(S)} + D_2 (S)  
    \leq \epsilon/(2\rho\gamma) \cdot \mathsf{I}_2(S) + (2\rho\gamma/ \epsilon + 1) \cdot D_2(S) .
    \end{split}
    \end{align*}
    Note that
    $$
    D_2(S) = \sum_{e<f} \omega^{e,f} \left| (\widehat{\nu}^{e,S}- \nu^{e,S}) - (\widehat{\nu}^{f,S}- \nu^{f,S} )\right|^2
    \leq 4 \max_e \left| \widehat{\nu}^{e,S} - \nu^{e,S} \right|^2 
    \leq 32 r_n^2 + 32 \left\|\beta - \beta^\star \right\|_2^4 r_n^2
    $$
    where the last inequality follows from~\eqref{eq:diff-in-sigma}. Moreover,
    $$
    \widehat{\mathsf{I}}_2(S^\star) - \mathsf{I}_2(S^\star) 
    = \widehat{\mathsf{I}}_2(S^\star) 
    = D_2(S^\star) \leq 32 r_n^2,
    $$
    by taking $\beta = \beta^\star$ under the true support $S^\star$.
    Putting them together, we have
    \begin{align*}
    \begin{split}
    & \left| \mathsf{I}_2(S) - \mathsf{I}_2(S^\star) - \widehat{\mathsf{I}}_2(S) + \widehat{\mathsf{I}}_2(S^\star) \right| \\
    \leq{}& \epsilon/(2\rho\gamma) \cdot \mathsf{I}_2(S) + (2\rho\gamma/ \epsilon + 1) D_2(S)  + D_2(S^\star)
    \leq \epsilon / (2\rho\gamma) \cdot \mathsf{I}_2(S) + (64\rho\gamma/\epsilon + 32) (r_n^2 + \left\| \beta - \beta^\star \right\|_2^4 r_n^2) + 32 r_n^2 .
    \end{split}
    \end{align*}

    \vspace{0.1in}
    \textit{\textbf{Step III.}} Combine the results. We combine the bounds obtained in the previous steps and conclude that, for any $S$ with $|S|\leq s_n$,
    \begin{align*}
    \begin{split}
    &    \left|  \mathsf{I}(S) - \mathsf{I}(S^\star) - \widehat{\mathsf{I}}(S) + \widehat{\mathsf{I}}(S^\star) \right| \\
    \leq{}& \left| \mathsf{I}_1(S) - \mathsf{I}_1(S^\star) - \widehat{\mathsf{I}}_1(S) + \widehat{\mathsf{I}}_1(S^\star) \right| + \rho \left| \mathsf{I}_2(S) - \mathsf{I}_2(S^\star) - \widehat{\mathsf{I}}_2(S) + \widehat{\mathsf{I}}_2(S^\star) \right| \\
    \leq{}& \epsilon / (2\gamma) \cdot \left( \mathsf{I}_1(S) + \mathsf{I}_2(S) \right) + ((16+64\rho^2)\gamma/\epsilon + 16 + 64\rho) r_n^2 \\
    &+(16\gamma/\epsilon + 8) r_n^2 \cdot \|\beta - \beta^\star \|_2^2 
    + (64\rho^2\gamma/\epsilon + 32\rho) r_n^2 \cdot \|\beta - \beta^\star\|_2^4.
    \end{split}
    \end{align*}
    
\end{proof}

\vspace{0.2in}
\begin{lemma} \label{lem:OLS}
    Under the event $\mathcal{T}_2 \cap \mathcal{T}_4$, for any $S$ with $|S| \leq s_n$, the corresponding OLS coefficients  
    $$
    \beta^{{\rm OLS}, S} := \argmin_{\mathrm{supp}(\beta) \subseteq S } \sum_{e\in \mathcal{E}} \omega^e \left\| Y^e - X^e \beta \right\|^2_{n^e} 
    $$
    satisfy
    $$
    \left\| \beta^{{\rm OLS}, S} - \beta^\star  \right\|_2^2 \lesssim 1 + \left\| \beta^\star \right\|_2^2.
    $$
\end{lemma}
\begin{proof}
    Note that
    \begin{align*}
    \begin{split}
        \sum_e \omega^e\lambda_{\min}(\widehat{\Sigma}^e_{S}) \cdot \left\| \beta^{{\rm OLS},S}_S \right\|_2^2  \le \sum_e \omega^e \left( \left\| X_S^{e\top} \beta^{{\rm OLS}, S} \right\|_{n^e}^2 + \left\|Y^e - X_S^{e\top} \beta^{{\rm OLS},S}_S \right\|_{n^e}^2 \right) = \sum_e \omega^e \|Y^e\|_{n_e}^2,
    \end{split}
    \end{align*}
    where the last equality follows from the normal equations.
    On the one hand, under the event $\mathcal{T}_2$,
    we have $\lambda_{\rm min}(\widehat{\Sigma}^e) \asymp 1$. On the other hand, under the event $\mathcal{T}_4$ and by taking $S=S^\star$ and $\beta = \beta^\star$, we have $\|\varepsilon^e\|_{n^e}^2  \lesssim  r_n +\sigma^2_\varepsilon \asymp 1 $. As a result,  
    $$
    \left\| Y^e \right\|_{n^e} \leq \left\| X^{e\top}_{S^\star} \beta^\star_{S^\star} \right\|_{n^e} + \left\| \varepsilon^e \right\|_{n^e} \lesssim  \left\| \beta^\star \right\|_2 + 1.
    $$
    Consequently, 
    $$
    \left\| \beta^{{\rm OLS}, S} - \beta^\star  \right\|_2^2 \leq 2\left\| \beta^{{\rm OLS},S} \right\|_2^2 + 2\left\|\beta^\star \right\|_2^2 \lesssim 1 + \left\| \beta^\star \right\|_2^2.
    $$
    
\end{proof}

\vspace{0.2in}
\begin{lemma} \label{lem:lower-bound}
    Suppose all conditions of Lemma~\ref{lem:identification2} hold.
    Under the event $\mathcal{T}_1 \cap \mathcal{T}_2 \cap \mathcal{T}_3 \cap \mathcal{T}_4$, 
    $$
    \widehat{\mathsf{F}}(\beta) - \widehat{\mathsf{F}}(\beta^\star)  \gtrsim \left\| \beta - \beta^\star \right\|_2^2 - \frac{s_n \log p + \log |\mathcal{E}| }{ \min_{e\in\mathcal{E}} n^e } - 2\lambda_n s_n,
    $$
    where $\beta$ is the OLS coefficients $\beta^{{\rm OLS}, S}$, and $|S|\leq s_n$.
\end{lemma}
\begin{proof}
Consider the decomposition:
\begin{align*}
\begin{split}
    \widehat{\mathsf{F}}(\beta) - \widehat{\mathsf{F}}(\beta^\star)
    = -\gamma\mathsf{D}_{\mathsf{I}}(S, S^\star) - \mathsf{D}_{\mathsf{R}}(\beta, \beta^\star) + \mathsf{F}(\beta) - \mathsf{F}(\beta^\star) + \lambda_n(|S| - s^\star),
\end{split}
\end{align*}
where
$$
\mathsf{D}_{\mathsf{I}}(S, S^\star) := \mathsf{I}(S) - \mathsf{I}(S^\star) - \widehat{\mathsf{I}}(S) + \widehat{\mathsf{I}}(S^\star),
$$
which can be bounded using Lemma~\ref{lem:D-I}; 
and
$$
\mathsf{D}_{\mathsf{R}}(\beta, \beta^\star) := \mathsf{R}(\beta) - \mathsf{R}(\beta^\star) - \widehat{\mathsf{R}}(\beta) + \widehat{\mathsf{R}}(\beta^\star),
$$
which can be bounded using Lemma~\ref{lem:D-R}. Moreover, Lemma~\ref{lem:identification2} provides a bound for $\mathsf{F}(\beta) - \mathsf{F}(\beta^\star)$. Putting them together, we obtain:
\begin{align*}
\begin{split}
    \widehat{\mathsf{F}}(\beta) - \widehat{\mathsf{F}}(\beta^\star) 
    \geq&  - \frac{\epsilon}{2} \cdot \left( \mathsf{I}_1(S) + \mathsf{I}_2(S) \right)
    - \gamma\left( (16+64\rho^2)\gamma/\epsilon + 16 + 64\rho \right) r_n^2 \\
    &\quad - \gamma \left( 16 \gamma/\epsilon + 8 \right) r_n^2 \cdot \|\beta - \beta^\star \|_2^2  
    - \gamma \left( 64\rho^2\gamma/\epsilon + 32\rho \right) r_n^2 \cdot \|\beta - \beta^\star\|_2^4. \\
    &\quad- (C_0 \delta_n + \frac{\kappa_L}{4}) \left\| \beta - \beta^\star \right\|_2^2 - \frac{ (C_0' \delta_n)^2 }{\kappa_L}
    + \frac{1}{2} \kappa_L \left\| \beta - \beta^\star \right\|_2^2 + \epsilon\cdot \left( \mathsf{I}_1(S) + \mathsf{I}_2(S) \right) \\
    &\quad + \lambda_n(|S|- s^\star) \\
    &\gtrsim \left\| \beta - \beta^\star \right\|_2^2 - r_n^2 - \delta_n^2 - r_n^2 \left\| \beta - \beta^\star \right\|_2^4 - 2 \lambda_n s_n.
\end{split}
\end{align*}
Further note that $\delta_n \lesssim r_n$, and by Lemma~\ref{lem:OLS},
$$
r_n^2 \left\| \beta - \beta^\star \right\|_2^4 
\lesssim r_n^2 (1 + \left\|\beta^\star \right\|_2^2) \cdot \left\| \beta - \beta^\star \right\|_2^2 \lesssim (r_n^2+ r_n^{\delta}) \cdot  \left\| \beta - \beta^\star \right\|_2^2 
\lesssim r_n^{\delta} \left\| \beta - \beta^\star \right\|_2^2,
$$
where the second to the last inequality follows from Assumption~\ref{assump:sparse-moderate}.
As a result,
$$
 \widehat{\mathsf{F}}(\beta) - \widehat{\mathsf{F}}(\beta^\star)  
 \gtrsim \left\| \beta - \beta^\star \right\|_2^2 - r_n^2 - 2\lambda_n s_n.
$$

\end{proof}

\vspace{0.1in}
\subsection{Proof of Theorem~\ref{thm:convergence}}
\label{sec:app:proof:convergence} 

\begin{proof}
    We first provide the proof for the optimal solution $\widehat{\beta}$.
    Appealing to the basic inequality and Lemma~\ref{lem:lower-bound}:
    $$
    \left\| \widehat{\beta} - \beta^\star \right\|_2^2 - \frac{s_n \log p + \log |\mathcal{E}| }{ \min_{e\in\mathcal{E}} n^e } - 2\lambda_n s_n
    \lesssim \widehat{\mathsf{F}}(\widehat{\beta}) - \widehat{\mathsf{F}}(\beta^\star) 
    \leq 0 .
    $$
    The result hence follows.
    
    \vspace{0.05in}
    For the early-stopping solution $\widehat{\beta}^{\rm early}$, we have
    $$
    \widehat{\mathsf{F}}(\widehat{\beta}^{\rm early}) - \widehat{\mathsf{F}}(\beta^\star) 
    = \left( \widehat{\mathsf{F}}(\widehat{\beta}^{\rm early}) - \widehat{\mathsf{F}}(\widehat{\beta}) \right) + \left( \widehat{\mathsf{F}}(\widehat{\beta}) - \widehat{\mathsf{F}}(\beta^\star)  \right) 
    \leq \tau^{\rm early} +0 = \tau^{\rm early} .
    $$
    Then combining with the results in Lemma~\ref{lem:lower-bound}, we have
     $$
    \left\| \widehat{\beta}^{\rm early} - \beta^\star \right\|_2^2 - \frac{s_n \log p + \log |\mathcal{E}| }{ \min_{e\in\mathcal{E}} n^e } - 2\lambda_n s_n
    \lesssim \widehat{\mathsf{F}}(\widehat{\beta}^{\rm early}) - \widehat{\mathsf{F}}(\beta^\star) 
    \leq \tau^{\rm early} .
    $$
    The result hence follows.
\end{proof}

\vspace{0.1in}
\subsection{Proof of Theorem~\ref{thm:consistency}}
\label{sec:app:proof:consistency}

\subsubsection{Proof of part (1).}

\begin{proof}
    This proof proceeds in three steps.  In each step, we
solve an optimization problem that is more constrained than the original one.
\begin{itemize}
    \item In step I, we restrict the feasible region to 
    $$
    \mathcal{R}_1 = \left\{ \beta\in\mathbb{R}^p: \mathrm{supp}({\beta}) = \mathrm{supp}(\beta^\star) \right\},
    $$ and derive an upper bound for $\min_{\beta\in\mathcal{R}_1} \widehat{\mathsf{F}}(\beta) - \widehat{\mathsf{F}}(\beta^\star)$. We denote this upper bound by $u$.

    \vspace{0.05in}
    \item In step II, we restrict the feasible region to 
    $$
    \mathcal{R}_2 = \left\{\beta\in\mathbb{R}^p: |\mathrm{supp}({\beta})| \leq s^\star,\ \mathrm{supp}({\beta})\neq\mathrm{supp}(\beta^\star) \right\},
    $$ 
    and derive a lower bound for $\min_{\beta\in\mathcal{R}_2} \widehat{\mathsf{F}}(\beta) - \widehat{\mathsf{F}}(\beta^\star)$. We denote this lower bound by $l_1$, and verify that $u<l_1$ when $\min_{e\in\mathcal{E}} n^e$ is sufficiently large.

    \vspace{0.05in}
    \item In step III, we restrict the feasible region to 
    $$
    \mathcal{R}_3 = \left\{\beta\in\mathbb{R}^p: |\mathrm{supp}({\beta})| > s^\star \right\},
    $$
    and derive a lower bound for $\min_{\beta\in\mathcal{R}_3} \widehat{\mathsf{F}}(\beta) - \widehat{\mathsf{F}}(\beta^\star)$. We denote this lower bound by $l_2$, and verify that $u<l_2$ when $\min_{e\in\mathcal{E}} n^e$ is sufficiently large.
\end{itemize}

\vspace{0.05in}
The three steps above guarantee that the solution to the original problem, $\argmin_\beta \widehat{\mathsf{F}}(\beta) $, falls inside $\mathcal{R}_1$ when $\min_{e\in\mathcal{E}} n^e$
is sufficiently large. This
proof relies on the event $\mathcal{T}_1 \cap \mathcal{T}_2 \cap \mathcal{T}_3 \cap \mathcal{T}_4 $, which holds with probability at least $1 - 7/p$.

\vspace{0.1in}
\textit{\textbf{Step I.}}
In this step, we use $\bar\beta$ to denote the solution to $\min_{\beta \in \mathcal{R}_1} \widehat{\mathsf{F}}(\beta) - \widehat{\mathsf{F}}(\beta^\star) $.
Note that for any $\beta\in\mathcal{R}_1$, we have
$
  \widehat{\mathsf{F}}(\beta) - \widehat{\mathsf{F}}(\beta^\star)
= \widehat{\mathsf{R}}(\beta) - \widehat{\mathsf{R}}(\beta^\star).
$
Therefore, $\bar{\beta}$ is the pooled OLS coefficients, and
\begin{align*}
\begin{split}
    \bar{\beta}_{S^\star}
&:= \left( \sum_e \omega^e\frac{\boldsymbol{X}^{e\top}_{S^\star} \boldsymbol{X}^e_{S^\star}}{n^e} \right)^{-1} \left( \sum_e \omega^e\frac{\boldsymbol{X}^{e\top}_{S^\star} \boldsymbol{Y}^e}{n^e} \right)  
= \beta^\star + \left( \sum_e \omega^e\frac{\boldsymbol{X}^{e\top}_{S^\star} \boldsymbol{X}^e_{S^\star}}{n^e} \right)^{-1} \left( \sum_e \omega^e\frac{\boldsymbol{X}^{e\top}_{S^\star} \boldsymbol{\varepsilon}^e}{n^e} \right) .
\end{split}
\end{align*}
Let $A:= \sum_e \omega^e\boldsymbol{X}^{e\top}_{S^\star} \boldsymbol{X}^e_{S^\star} / n^e $ and $b:= \sum_e \omega^e\boldsymbol{X}^{e\top}_{S^\star} \boldsymbol{\varepsilon}^e / n^e $. Then we can rewrite $\bar\beta_{S^\star} - \beta^{\star}_{S^\star} =  A^{-1} b$, and hence
\begin{align*}
\begin{split}
    \widehat{\mathsf{R}}(\bar{\beta}) - \widehat{\mathsf{R}}(\beta^\star)
&=  \sum_e \omega^e\left( \left\| \boldsymbol{Y}^e - \boldsymbol{X}_{S^\star}^e \bar{\beta}_{S^\star} \right\|_{n^e}^2 - \left\| \boldsymbol{Y}^e - \boldsymbol{X}^e_{S^\star} \beta^\star_{S^\star} \right\|_{n^e}^2  \right) \\
&= \sum_e \omega^e\left( \left\| \boldsymbol{\varepsilon}^e - \boldsymbol{X}^e_{S^\star} (\bar\beta_{S^\star} - \beta^\star_{S^\star} ) \right\|_{n^e}^2 - \left\| \boldsymbol{\varepsilon}^e \right\|_{n^e}^2  \right) \\
&=  (\bar\beta_{S^\star} - \beta^\star_{S^\star} )^\top A (\bar\beta_{S^\star} - \beta^\star_{S^\star} ) - 2 (\bar\beta_{S^\star} - \beta^\star_{S^\star} )^\top b \\
&= (A^{-1} b)^\top A (A^{-1} b) - 2 (A^{-1} b)^\top b
= -b^{\top} A^{-1} b \\
&= - \left(  \sum_e \omega^e \frac{ \boldsymbol{X}^{e\top}_{S^\star} \boldsymbol{\varepsilon}^e }{n^e} \right)^{\top}
\left( \sum_e \omega^e\frac{\boldsymbol{X}^{e\top}_{S^\star} \boldsymbol{X}^e_{S^\star}}{n^e} \right)^{-1}
\left(  \sum_e \omega^e \frac{ \boldsymbol{X}^{e\top}_{S^\star} \boldsymbol{\varepsilon}^e }{n^e} \right).
\end{split}
\end{align*}
Now denote 
$
\Sigma_{S} :=  \sum_e \omega^e\Sigma^{e}_S. 
$
and
$
\widehat{\Sigma}_S :=  \sum_e \omega^e\widehat{\Sigma}^{e}_S.
$
Appealing to Weyl's inequality, 
\begin{align*}
\begin{split}
    \lambda_{\min} \left( \widehat{\Sigma}_{S} \right) 
&\geq  \lambda_{\min} \left( \Sigma_S \right) - \left\| \Sigma_S - \widehat{\Sigma}_S \right\|_{\rm op} 
\geq \lambda_{\min} \left( \Sigma_S \right)  - \sum_e \omega^e\left\|\Sigma^{e}_S - \widehat{\Sigma}^{e}_S \right\|_{\rm op} \\
&\geq \lambda_{\min} \left( \Sigma_S \right) - \sqrt{ \frac{s_n \log p + \log |\mathcal{E}|}{\min_{e\in\mathcal{E}} n^e} }
\geq \kappa_L  - \sqrt{ \frac{s_n \log p + \log |\mathcal{E}|}{\min_{e\in\mathcal{E}} n^e} }.
\end{split}
\end{align*}
When $\min_{e\in\mathcal{E}} n^e \gtrsim s_n \log p + \log |\mathcal{E}| $, we obtain $\lambda_{\rm min}(\widehat{\Sigma}_S) > 0$ for all $S$. As a result, 
$$
\left( \sum_e \omega^e\frac{\boldsymbol{X}^{e\top}_{S^\star} \boldsymbol{X}^e_{S^\star}}{n^e} \right) \succ 0,
$$
and $\widehat{\mathsf{R}}(\bar{\beta}) - \widehat{\mathsf{R}}(\beta^\star) < 0$ when $\min_{e\in\mathcal{E}} n^e$ is sufficiently large. We set the strict upper bound of $\min_{\beta \in \mathcal{R}_1} \widehat{\mathsf{F}}(\beta) - \widehat{\mathsf{F}}(\beta^\star)$ as $u = -b^\top A^{-1}b / 2 < 0$.

\vspace{0.1in}
\textit{\textbf{Step II.}} In this step, we overload the notation $\bar{\beta}$ to denote the solution to $\min_{\beta \in \mathcal{R}_2}\ \left\{ \widehat{\mathsf{F}}(\beta) - \widehat{\mathsf{F}}(\beta^\star) \right\} $. We also denote $\bar{S} := \mathrm{supp}(\bar\beta)$. Recall from Lemma~\ref{lem:lower-bound} that for any feasible $\beta$, 
 $$
\widehat{\mathsf{F}}(\beta) - \widehat{\mathsf{F}}(\beta^\star) 
\gtrsim \left\| \beta - \beta^\star \right\|_2^2 - \frac{s_n \log p + \log |\mathcal{E}| }{ \min_{e\in\mathcal{E}} n^e } + \lambda_n (|\bar{S}| - s^\star).
$$
Now we attempt to connect this lower bound with the number of missed causal parents $\eta_1:= |S^\star \setminus \bar{S}|$. We
also define $\eta_2:= |\bar{S} \setminus S^\star|$. For any $\bar{\beta}\in\mathcal{R}_2$, we consider two cases.

\vspace{0.05in}
\textit{\underline{Case 1:}} If $|\bar{S}| = s^\star$, then $\eta_1 = \eta_2 \geq 1$. As a result, 
$$
\left\| \bar{\beta} - \beta^\star \right\|_2^2 
\geq \eta_1 \min_{j\in S^\star} \left| \beta^\star_j \right|^2
\geq \min_{j\in S^\star} \left| \beta^\star_j \right|^2,
$$
and consequently,
$$
\widehat{\mathsf{F}}(\bar \beta) - \widehat{\mathsf{F}}(\beta^\star)  \gtrsim \min_{j\in S^\star} \left| \beta^\star_j \right|^2 - \frac{s_n \log p + \log |\mathcal{E}| }{ \min_{e\in\mathcal{E}} n^e }.
$$
We call this lower bound $l_1^{(1)}$. Appealing to the beta-min condition~\eqref{eq:beta-min}, we have $l_1^{(1)} \geq 0 > u$ when $\min_{e\in\mathcal{E}} n^e$ is sufficiently large.

\vspace{0.05in}
\textit{\underline{Case 2:}} If $|\bar{S}| < s^\star$, then $\eta_2 \geq 0$ and $\eta_1 = s^\star - |\bar{S}| + \eta_2$. As a result, 
$$
\left\| \bar{\beta} - \beta^\star \right\|_2^2 
\geq \eta_1 \min_{j\in S^\star} \left| \beta^\star_j \right|^2
\geq (s^\star - |\bar{S}| ) \min_{j\in S^\star} \left| \beta^\star_j \right|^2,
$$
and consequently,
\begin{align*}
\begin{split}
    \widehat{\mathsf{F}}(\bar\beta) - \widehat{\mathsf{F}}(\beta^\star)  
&\gtrsim (s^\star - |\bar{S}|) \min_{j\in S^\star} \left| \beta^\star_j \right|^2 - \frac{s_n \log p + \log |\mathcal{E}| }{ \min_{e\in\mathcal{E}} n^e } + \lambda_n (|\bar{S}| - s^\star) \\
& \geq (s^\star - |\bar{S}|) \left( 
\min_{j\in S^\star} |\beta^\star_j|^2 - \lambda_n - \frac{s_n \log p + \log |\mathcal{E}| }{ \min_{e\in\mathcal{E}} n^e }
\right) \\
&\geq 
\min_{j\in S^\star} |\beta^\star_j|^2 - \lambda_n - \frac{s_n \log p + \log |\mathcal{E}| }{ \min_{e\in\mathcal{E}} n^e }.
\end{split}
\end{align*}
We call this lower bound $l_1^{(2)}$. Again, appealing to the beta-min condition~\eqref{eq:beta-min}, we have $l_1^{(2)} \geq 0 > u$ when $\min_{e\in\mathcal{E}} n^e$ is sufficiently large.

\vspace{0.2in}
\textit{\textbf{Step III.}} In this step, we overload the notation $\bar{\beta}$ again to denote the solution to the optimization $\arg\min_{\beta \in \mathcal{R}_3}\ \left\{ \widehat{\mathsf{F}}(\beta) - \widehat{\mathsf{F}}(\beta^\star) \right\} $. We also overload $\bar{S}:=\mathrm{supp}(\bar\beta)$. Similarly as in Step II, we have
\begin{align*}
\begin{split}
    \widehat{\mathsf{F}}(\bar\beta) - \widehat{\mathsf{F}}(\beta^\star)  
&\gtrsim  \left\| \bar{\beta} - \beta^\star \right\|_2^2 - \frac{s_n \log p + \log |\mathcal{E}| }{ \min_{e\in\mathcal{E}} n^e } + \lambda_n (|\bar{S}| - s^\star) 
\geq  \lambda_n - \frac{s_n \log p + \log |\mathcal{E}| }{ \min_{e\in\mathcal{E}} n^e }.
\end{split}
\end{align*}
We call this lower bound $l_2$. With the choice $\lambda_n \gtrsim  (s_n \log p + \log |\mathcal{E}|) / \min_{e\in\mathcal{E}} n^e  $ and when $\min_{e\in\mathcal{E}} n^e$ is sufficiently large, we have $l_2 \geq 0 > u$.

\end{proof}

\vspace{0.2in}
\subsubsection{Proof of part (2)}
We use the same proof technique of part (1). From the basic inequality,
$$
\widehat{\mathsf{F}}(\widehat{\beta}^{\rm early}) - \widehat{\mathsf{F}}(\beta^\star) \leq \widehat{\mathsf{F}}(\widehat{\beta}) - \widehat{\mathsf{F }}(\beta^\star) + \tau^{\rm early} .
$$
Consequently, the upper bound of Step I becomes $u = -b^\top A^{-1}b/2 + \tau^{\rm early} $. When $\tau^{\rm early} \lesssim (s_n \log p + \log |\mathcal{E}|) / \min_{e\in\mathcal{E}} n^e$, we still have $l_1^{(1)} > u$, $l_1^{(2)} > u$, and $l_2 > u$ when $\min_{e\in\mathcal{E}} n^e$ is sufficiently large. The result hence follows.

\section{Proofs of results in Section~\ref{sec:task2}}
In this section, we present the proofs of results in Section~\ref{sec:task2} and Appendix~\ref{sec:app:bootstrap}. Without loss of generality, we focus on the case of two environments, $\mathcal{E}=\{e,f\}$. 
We rewrite $\gamma_1 = \gamma$ and $\gamma_2 = \gamma\rho$ in this section. 
We denote 
$$
U_n(k) := \min_{S\notin \mathcal{S}_k}\ n^\star \left\{ \gamma_1 \left\| \widehat{\beta}^{e,S} - \widehat{\beta}^{f,S}  \right\|_2^2 + \gamma_2 \left| \widehat{\nu}^{e,S} - \widehat{\nu}^{f,S} \right|^2 \right\},
$$
$$
T_n(k) := \max_{S\in \mathcal{S}_k}\ n^\star \left\{ \gamma_1 \left\| \widehat{\beta}^{e,S} - \widehat{\beta}^{f,S}  \right\|_2^2 + \gamma_2 \left| \widehat{\nu}^{e,S} - \widehat{\nu}^{f,S} \right|^2 \right\},
$$
and
$$
T_n^\circ(k) := \max_{S\in \mathcal{S}_k^\circ}\ n^\star \left\{ \gamma_1 \left\| \widehat{\beta}^{e,S} - \widehat{\beta}^{f,S}  \right\|_2^2 + \gamma_2 \left| \widehat{\nu}^{e,S} - \widehat{\nu}^{f,S} \right|^2 \right\}.
$$
We also adopt the notation introduced in Appendix~\ref{sec:app:proof:task1}.

\vspace{0.1in}
\subsection{Proof of Lemma~\ref{lem:bootstrap-D1}}
\label{sec:app:proof:bootstrap-D1} 
\begin{proof}
We suppress the parentheses involving $k$ in the notation whenever there is no ambiguity. 
It suffices to show that $\mathbb{P}_{H_{0,k}}(T_n < U_n) \rightarrow 1$, since it implies that
$$
1 - o(1) \leq \mathbb{P}_{H_{0,k}}(T_n < U_n) 
\leq \mathbb{P}_{{H_{0,k}}} \left(
\left\{ \mathcal{S}_k \subseteq \mathcal{S}_k^\circ \right\} \cap \left\{  \widehat{S} \in \mathcal{S}_k \right\} \right)
\leq 1.
$$

Now consider the working model 
$Y^e \mid X^e_S \sim \mathcal{N}(X^{e\top} \beta^{e,S}, \nu^{e,S})$,
where $\beta^{e,S}$ is the population OLS solution of $(X_S^e, Y^e)$. Let $r^{e,S} = Y^e - X^{e\top}_S {\beta}^{e,S}_S  $, we have
$$
\mathbb{E}\left[ X_S^e r^{e,S}  \right] = 0, \quad
\mathbb{E}\left[ (r^{e,S})^2 \right] = \nu^{e,S}, \quad \forall e\in\mathcal{E.}
$$
As a result, for the sample OLS solutions,
\begin{align*}
\begin{split}
    \widehat{\beta}^{e,S}_S - \widehat{\beta}^{f,S}_S 
    &= \beta^{e,S}_S - \beta^{f,S}_S + \left( \frac{ \boldsymbol{X}_S^{e\top}\boldsymbol{X}^e_S }{ n^e} \right)^{-1} \left( \frac{ \boldsymbol{X}^{e\top}_S \boldsymbol{r}^{e,S} }{n^e} \right)
    - \left( \frac{ \boldsymbol{X}_S^{f\top} \boldsymbol{X}^f_S }{ n^f} \right)^{-1} \left( \frac{ \boldsymbol{X}^{f\top}_S \boldsymbol{r}^{f,S} }{n^f} \right)  \\
    &= \beta_S^{e,S} - \beta_S^{f,S} + o_p(1),
\end{split}
\end{align*}
and
\begin{align*}
\begin{split}
    \widehat{\nu}^{e,S} - \widehat{\nu}^{f,S}
&= \frac{1}{n^e}\sum_{i=1}^{n^e} \left( Y^e_i - (\widehat{\beta}^{e,S}_S)^\top X^{e}_{S,i} \right)^2 - \frac{1}{n^f}\sum_{i=1}^{n^f} \left( Y^f_i - (\widehat{\beta}_S^{f,S})^\top X^{f}_{S,i} \right)^2
= \nu^{e,S} - \nu^{f,S} + o_p(1) .
\end{split}
\end{align*}
Therefore, by continuous mapping theorem, we have
$$
T_n / n^\star = \max_{S \in \mathcal{S}_k } \left[ \gamma_1 \left\| {\beta}^{e,S} - {\beta}^{f,S} \right\|_2^2 + \gamma_2 \left| {\nu}^{e,S} - {\nu}^{f,S} \right|^2  \right] +o_p(1)
$$
and
$$
U_n/n^\star = \min_{S \notin \mathcal{S}_k } \left[ \gamma_1 \left\| {\beta}^{e,S} - {\beta}^{f,S} \right\|_2^2 + \gamma_2 \left| {\nu}^{e,S} - {\nu}^{f,S} \right|^2  \right] +o_p(1)
$$
For any $S \in\mathcal{S}_k$, we have $\beta^{e,S} = \beta^{f,S}$ and $\nu^{e,S} = \nu^{f,S}$. However, for any $S \notin\mathcal{S}_k$, either $\beta^{e,S} \neq \beta^{f,S}$ or $\nu^{e,S} \neq \nu^{f,S}$. These imply that
$$
(U_n - T_n ) / n^\star = C + o_p(1)
$$
for some constant $C>0$, 
and the result follows.

\end{proof}

\vspace{0.1in}
\subsection{Intermediate results}
In this section, we present several intermediate results that will be used later to prove Theorem~\ref{thm:type-1}. Recall that we focus on the case of two environments, $\mathcal{E}=\{e,f\}$. We further denote the random variables constructed from bootstrap (line 8 in Algorithm~\ref{algo:bootstrap-threshold}) by
$$
T_n^\star (k) :=
\max_{S\in\mathcal{S}_k}
\left\{ \gamma_1 \|L_n^{\star,e,f,S}\|_2^2
+ \gamma_2 |V_n^{\star,e,f,S}|^2
\right\},
$$
and
$$
T_n^{\circ,\star} (k) :=
\max_{S\in\mathcal{S}_k^\circ}
\left\{ \gamma_1 \|L_n^{\star,e,f,S}\|_2^2
+ \gamma_2 |V_n^{\star,e,f,S}|^2
\right\}.
$$
For simplicity, we suppress the parentheses in the notation whenever there is no ambiguity.
Note that for every bootstrap realization, we have $T_n^\star \leq T_n^{\circ,\star}$. We also define
$$
F_{n,k}(t) := \mathbb{P}_{H_{0,k}}\left( T_n \leq t \right), \quad
F^\star_{n,k}(t) := \mathbb{P}_{H_{0,k}}\left( T^\star_n \leq t \mid \mathcal{D}_2 \right), \quad
G^\star_{n,k}(t) := \mathbb{P}_{H_{0,k}}\left( T^{\circ,\star}_n \leq t \mid \mathcal{D}_2 \right).
$$
With $B_n$ i.i.d. bootstrap draws $\{T^{\circ,\star}_{n,b}\}_{b=1}^{B_n}$ from the bootstrap distribution $T_n^{\circ, \star}$, we define the empirical bootstrap CDF 
$$
\widehat{G}^{\star}_{n,k}(t) := \frac{1}{B_n} \sum_{b=1}^{B_n} \mathds{1}\left( T^{\circ,\star}_{n,b} \leq t \right).
$$

\vspace{0.2in}
\begin{lemma} 
\label{lem:F-Fstar}
Suppose that $n^e / n^\star \rightarrow\rho_e >0$ and $n^f / n^\star \rightarrow\rho_f >0$. Then the bootstrap distribution $F_{n,k}^\star$ consistently approximates $F_{n,k}$:
$$
\sup_t \left| F_{n,k}(t) - F^\star_{n,k}(t) \right| = o_p(1).
$$  
\end{lemma}
\begin{proof}

For notational simplicity, we present the proof for $\mathcal{S}_k = \{S, \bar S\}$; the same argument applies to any fixed finite $\mathcal S_k$ by stacking the corresponding components. 
For each $A\in\mathcal{S}_k$, define
$$
L_n^A := \sqrt{n^\star}\left(\widehat{\beta}^{e,A}-\widehat{\beta}^{f,A}\right), \qquad
V_n^A := \sqrt{n^\star}\left(\widehat{\nu}^{e,A}-\widehat{\nu}^{f,A}\right),
$$
Under $H_{0,k}$, both $S$ and $\bar S$ are invariant. By the asymptotic linearity of the OLS coefficient and residual variance,
we have for any $A \in \mathcal{S}_k$,
\begin{align*}
\begin{split}
    L_n^{A}& 
= \sqrt{\frac{n^\star}{n^e}} \left( \frac{ \boldsymbol{X}_A^{e\top} \boldsymbol{X}^e_A }{n^e} \right)^{-1} \frac{ \boldsymbol{X}_A^{e\top} {\boldsymbol{r}}^{e,A} }{\sqrt{n^e}} - \sqrt{\frac{n^\star}{n^f}} \left( \frac{ \boldsymbol{X}_A^{f\top} \boldsymbol{X}^f_A }{n^f} \right)^{-1} \frac{ \boldsymbol{X}_A^{f\top} {\boldsymbol{r}}^{f, A} } {\sqrt{n^f}},
\end{split}
\end{align*}
and
$$
V^A_n = \sqrt{\frac{n^\star}{n^e}} \frac{1}{\sqrt{n^e}} \sum_{i=1}^{n^e} ({r}_i^{e,A})^2 - \sqrt{\frac{n^\star}{n^f}} \frac{1}{\sqrt{n^f}} \sum_{i=1}^{n^f} ({r}_i^{f,A})^2 + o_p(1),
$$
where ${r}^{e,A} = Y^e - X^{e\top}_A {\beta}^{e,A}_A$.
Now stack these quantities as $W_n = \left[ L_n^S, L_n^{\bar S},  V_n^S,  V_n^{\bar S} \right]^\top$; then we obtain
$$
W_n =
\sqrt{\frac{n^\star}{n^e}} \frac{1}{\sqrt{n^e}} \sum_{i=1}^{n^e}\psi_i^e -
\sqrt{\frac{n^\star}{n^f}} \frac{1}{\sqrt{n^f}} \sum_{i=1}^{n^f}\psi_i^f + o_p(1),
$$
where for $g\in\{e,f\}$,
$$
\psi_i^g
:=
\begin{bmatrix}
({\Sigma}^{g}_S)^{-1}X_{S,i}^{g} \cdot r_i^{g, S}\\
({\Sigma}^{g}_{\bar S})^{-1} X^{g}_{\bar S, i} \cdot r_i^{g,\bar S}\\
(r_i^{g,S})^2-\nu^{g,S}\\
(r_i^{g,\bar S})^2-\nu^{g,\bar S}
\end{bmatrix}, \quad \text{and} 
\quad \mathbb{E}\left[ \psi^g_i \right] = 0.
$$
Since $n^e/n^\star\to\rho_e>0$ and $n^f/n^\star\to\rho_f>0$, the multivariate central limit theorem and Slutsky's theorem imply
$$
W_n \rightsquigarrow W\sim\mathcal{N}(0,\Gamma), \quad \text{where}\quad 
\Gamma= \frac{1}{\rho_e}\operatorname{Var}(\psi_i^e) +
\frac{1}{\rho_f}\operatorname{Var}(\psi_i^f).
$$

Now denote the corresponding bootstrap vector by
$
W_n^\star := \left[ L_n^{\star, S}, L_n^{\star, \bar S}, V_n^{\star, S}, V_n^{\star, \bar S} \right]^\top. 
$
Note that
$$
W_n^\star = \sqrt{\frac{n^\star}{n^e}} \frac{1}{\sqrt{n^e}} \sum_{i=1}^{n^e} \xi^e_i \widehat{\psi}_i^e -  \sqrt{\frac{n^\star}{n^f}} \frac{1}{\sqrt{n^f}} \sum_{i=1}^{n^f} \xi^f_i \widehat{\psi}_i^f, \quad
\text{where} \quad 
\widehat{\psi}_i^g
:=
\begin{bmatrix}
(\widehat{\Sigma}^{g}_S)^{-1}X_{S,i}^{g} \cdot \widehat{r}_i^{g, S}\\
(\widehat{\Sigma}^{g}_{\bar S})^{-1} X^{g}_{\bar S, i} \cdot \widehat{r}_i^{g,\bar S}\\
(\widehat{r}_i^{g,S})^2- \widehat{\nu}^{g,S}\\
(\widehat{r}_i^{g,\bar S})^2- \widehat{\nu}^{g,\bar S}
\end{bmatrix}.
$$
Conditional on $\mathcal{D}_2$, we have
$W_n^\star\mid\mathcal{D}_2 \sim \mathcal{N}(0,\widehat{\Gamma}_n)$, where
$$
\widehat{\Gamma}_n = \frac{n^\star}{(n^e)^2} \sum_{i=1}^{n^e} \widehat{\psi}_i^e \widehat{\psi}_i^{e\top} + \frac{n^\star}{(n^f)^2} \sum_{i=1}^{n^f} \widehat{\psi}_i^f \widehat{\psi}_i^{f\top}, 
$$
since $\mathbb{E}[\xi^g_i \xi^g_j] = \mathds{1}(i =j)$ for all $g\in\{e,f\}$. Furthermore, by OLS consistency, the law of large numbers, and Slutsky's theorem, we have
\begin{align*}
\begin{split}
    \widehat{\Gamma}_n 
&= \frac{1}{\rho_e} \frac{1}{n^e} \sum_{i=1}^{n^e} {\psi}^e_i {\psi}^{e\top}_i + \frac{1}{\rho_f} \frac{1}{n^f} \sum_{i=1}^{n^f} {\psi}^f_i {\psi}^{f\top}_i + o_p(1) \\
&=\frac{1}{\rho_e}  \mathbb{E}\left[  \psi_i^e \psi_i^{e\top} \right] + \frac{1}{\rho_f} \mathbb{E} \left[ \psi_i^f \psi_i^{f\top} \right] + o_p(1) = \Gamma +o_p(1),
\end{split}
\end{align*}

Finally,  define the continuous function
$$
g(w) := \max_{A\in\mathcal{S}_k} \left\{ \gamma_1\|w_L^A\|_2^2 + \gamma_2|w_V^A|^2 \right\}.
$$
Then $T_n=g(W_n)$ and $ T_n^\star=g(W_n^\star)$.
By the continuous mapping theorem, we have $ T_n\rightsquigarrow g(W)$; conditional on $\mathcal{D}_2$, we have $T_n^\star \rightsquigarrow g(W)$ in probability.
Provided that the distribution function $F(t)$ of $g(W)$ is continuous, P\'{o}lya's theorem yields
$$
\sup_t \left|F_{n,k}(t)-F(t) \right|\to 0, \quad
\sup_t \left| F_{n,k}^\star(t) - F(t) \right|= o_p(1).
$$
As a result,
$$
\sup_t \left| F_{n,k}(t)-F_{n,k}^\star(t) \right|
\leq 
\sup_t \left|F_{n,k}(t)-F(t) \right| + \sup_t \left| F_{n,k}^\star(t) - F(t) \right| =
o_p(1).
$$

\end{proof}

\vspace{0.2in}
\begin{lemma}
\label{lem:Gstarhat-Gstar}
The empirical bootstrap distribution $\widehat{G}_{n,k}^\star$ uniformly approximates $G_{n,k}^\star$ as the bootstrap sample size $B_n\rightarrow\infty$:
$$
\sup_t \left| \widehat{G}^\star_{n,k}(t) - G^\star_{n,k}(t) \right| = o_p(1).
$$
\end{lemma}
\begin{proof}
By the Dvoretzky–Kiefer–Wolfowitz inequality, for every $\epsilon>0$,
$$
\mathbb{P} \left( \sup_t \left| \widehat{G}^\star_{n,k}(t) - G^\star_{n,k}(t) \right| > \epsilon \mid \mathcal{D} \right) \leq  2 \exp(-2B_n\epsilon^2).
$$
The result hence follows.

\end{proof}

\vspace{0.1in}
\subsection{Proof of Theorem~\ref{thm:type-1}}
\label{sec:app:proof:type-1}
\begin{proof}
The proof proceeds in two steps. In the first step, we show the type-I error control for each hypothesis $H_{0,k}$. In the second step, we show the asymptotic probabilistic guarantee for $\{ \widehat{S}(\mathcal{E}) \subseteq S^\star \}$.

\vspace{0.1in}
\textit{\textbf{Step I.}}
Define the sample quantile
$$
\widehat{t}_{n,k}^\circ := \inf\left\{ t: \widehat{G}^\star_{n,k}(t) \geq 1- \alpha/p\right\}, 
$$
and the distribution discrepancy
$$
\Delta_{1,n,k}:= \sup_t \left| F_{n,k}(t) - F^\star_{n,k}(t) \right|, \quad
\Delta_{2,n,k} := \sup_t \left| \widehat{G}^\star_{n,k}(t) - G^\star_{n,k}(t) \right|.
$$
By Lemma~\ref{lem:F-Fstar} and~\ref{lem:Gstarhat-Gstar}, we know $\Delta_{1,n,k} + \Delta_{2,n,k} = o_p(1)$, and hence we can choose a deterministic sequence $\epsilon_n \downarrow 0$ such that $\mathbb{P}_{H_{0,k}}(\Delta_{1,n,k} + \Delta_{2,n,k} > \epsilon_n) = o(1)$. Now consider the population quantile 
$$
q_{n,k}:= F^{-1}_{n,k}(1 - \alpha/p - \epsilon_n).
$$
On the event $A_{n,k}:=\{ \Delta_{1,n,k} + \Delta_{2,n,k} \leq \epsilon_n \}$, for all $t < q_{n,k}$, we have $F_{n,k}(t)< 1- \alpha/p - \epsilon_n$, and hence
$$
\widehat{G}^\star_{n,k}(t) \leq G^\star_{n,k}(t) + \Delta_{2n,k} \leq F^\star_{n,k}(t) + \Delta_{2n,k} \leq
F_{n,k}(t) + \Delta_{1n,k} + \Delta_{2n,k} < 1- \alpha / p.
$$
The second inequality follows from the fact that $T_n^\star \leq T_n^{\circ, \star}$.
This implies that the CDF $\widehat{G}^{\star}_{n,k}$ cannot reach $1- \alpha/p$ before $q_{n,k}$, and consequently $q_{n,k}\leq \widehat{t}^\circ_{n,k}$ on the event $A_{n,k}$. As a result, we have
\begin{align*}
\begin{split}
    \mathbb{P}_{H_{0,k}} \left( T_n > \widehat{t}^\circ_{n,k} \right)
    &\leq \mathbb{P}_{H_{0,k}} \left( \left\{ T_n > \widehat{t}^\circ_{n,k} \right\} \cap \left\{ A_{n,k} \right\} \right) + \mathbb{P}_{H_{0,k}} \left(\left\{ A_{n,k} \right\}^c \right) \\
    &\leq \mathbb{P}_{H_{0,k}} \left(  T_n > q_{n,k}  \right) + \mathbb{P}_{H_{0,k}}( \Delta_{1,n,k} + \Delta_{2,n,k} > \varepsilon_n) \\
    &\leq \alpha/p  + o(1).
\end{split}
\end{align*}
Finally, by the definition of $T_n(k)$ and the assumption that $\widehat{S} \in \mathcal{S}_k$, we have $\widehat{T}(k) \leq T_n(k)$. Furthermore, $\widehat{t}_{n,k}^\circ$ coincides with the bootstrap threshold $\widehat{\tau}_k(\alpha)$ defined in line 10 of Algorithm~\ref{algo:bootstrap-threshold}. As a result, 
$$
\mathbb{P}_{H_{0,k}} \left(\widehat{T}(k) > \widehat{\tau}_k(\alpha) \right) 
\leq \mathbb{P}_{H_{0,k}}\left( T_n(k) > \widehat{t}_{n,k}^\circ \right) \leq \alpha / p + o(1).
$$

\vspace{0.2in}
\textit{\textbf{Step II.}}
The result follows immediately that
$$
\mathbb{P} \left( \widehat{S}(\mathcal{E}) \not\subseteq S^\star \right) 
\leq \sum_{k=1}^p \mathbb{P}_{H_{0,k}} \left( \widehat{T}(k) > \widehat{\tau}_k(\alpha) \right) \leq p \cdot \frac{\alpha}{p} + o(1) = \alpha +o(1).
$$
\end{proof}

\section{Details on experiments}
\label{sec:app:experiments}

We first describe the synthetic data generation and intervention mechanisms we used. 
We generate data from linear Gaussian structural equation models over $p$ covariates $X=(X_1,\ldots,X_p)$ and a response $Y$. We first sample a common topological ordering of the $p+1$ variables and place $Y$ at a prespecified position in this ordering. The causal parent set $S^\star$ of $Y$ is sampled from the variables preceding $Y$, with $|S^\star|$ drawn between $s_{\min}$ and $s_{\max}$. For $j\in S^\star$, the invariant coefficients are sampled as
$$
\beta_j^\star = \xi_j U_j,\qquad
\xi_j\in\{-1,1\},\quad
U_j\sim \mathsf{Unif}(l_\star,u_\star),
$$
independently. Thus, in every environment,
$$
Y=\sum_{j\in S^\star}\beta_j^\star X_j+\varepsilon_Y,
\qquad
\varepsilon_Y\sim N(0,\sigma_Y^2),
$$
For the remaining structural equations, coefficients are sampled independently as a random sign multiplied by a uniform magnitude on $[l,u]$. The parameter $p_{\rm act}$ controls the sparsity of the covariate DAG: for each pair of variables consistent with the topological ordering, an edge is included with probability $p_{\rm act}$. Conditional on the resulting structural coefficients, observations are generated recursively according to the common topological ordering.

\vspace{0.05in}
Under \textbf{the soft-intervention mechanism}, the DAG structure is sampled once and shared across all environments. We first generate a baseline coefficient matrix $A^{(0)}$. The observational environment uses this baseline mechanism. For each subsequent environment $e>0$, every active edge into a covariate $X_j$ is independently selected with probability $p_{\rm resample}$; selected coefficients are then redrawn from the same signed uniform distribution, while the edge support remains unchanged. The noise standard deviations of the covariates are sampled once from $\mathsf{Unif}(\sigma_{X,\min},\sigma_{X,\max})$ and are also shared across environments. Hence, environments differ through perturbations of existing structural coefficients while preserving both the DAG support and the noise distributions.

\vspace{0.05in}
Under \textbf{the all-intervention mechanism}, the topological ordering and the response mechanism are still shared across environments, but the mechanisms generating covariates are resampled independently in every environment. Specifically, for each environment and for every admissible pair of covariates, an edge is included independently with probability $p_{\rm act}$, and its coefficient is sampled afresh from the signed uniform distribution. In addition, under the all-intervention mechanism, the covariate noise standard deviations are independently resampled in each environment from $\sigma_{X_j}^{(e)} \sim \mathsf{Unif}(\sigma_{X,\min},\sigma_{X,\max})$, whereas $\sigma_Y$ remains fixed.

\vspace{0.1in}
\subsection{Experiment of Figure~\ref{fig:vanilla-cd}}
\label{sec:app:figure-vanilla-cd}

We generate data using the soft-intervention mechanism described above. We vary the ambient dimension over $p\in\{10,20,30,40,50\}$ and generate 3 environments with $n^e=200$ observations per environment. For each value of $p$, we conduct 100 independent replications.  
The SEM is generated with $p_{\rm act}\in\{0.1,0.2\}$, $l=u=0.1$, $l_\star=0.5$, $u_\star = 1$, $\sigma_{X,\min}=0.5$, $\sigma_{X,\max}=1$, $\sigma_Y=0.5$, $s_{\min}=4$, $s_{\max}=5$, and $p_{\rm resample}=0.05$.

We compare the heuristic coordinate descent (CD) with the vanilla MIP formulation, both with and without a warm start. All three methods use the same objective parameters $\lambda=0.01$, $\gamma_1=5$, and $\rho=1$.
For coordinate descent, we initialize with the empty set. The MIP with a warm start is initialized using the support returned by coordinate descent, whereas the vanilla MIP is initialized from the empty set. For both MIP variants, we set \texttt{MIPFocus=2} and \texttt{Cuts=2}, and impose a time limit of 180 seconds. 

\vspace{0.1in}
\subsection{Experiment of Figure~\ref{fig:identification}}
\label{sec:app:figure-identification}
We use the data generation mechanism described in Example~\ref{ex:identification}. To avoid relying on the exact parameter configuration $(a_1,a_2)=(1,1/2)$ and $(v_1,v_2)=(2,3/2)$, which occurs with probability zero under a continuous parameter distribution, we independently perturb these parameters for each generated dataset according to
$$
a_e=(a_e)_0+\delta_{a,e},\qquad
v_e=(v_e)_0+\delta_{v,e},
$$
where $\delta_{a,e},\delta_{v,e}\overset{\mathrm{i.i.d.}}{\sim}\mathcal N(0,0.05^2)$, with $(a_1)_0=1,(a_2)_0=1/2$ and $(v_1)_0=2,(v_2)_0=3/2$.

\vspace{0.05in}
We compare \textit{MIP} with \textit{EILLS}, \textit{ICP}, \textit{Causal Dantzig}, and \textit{BIP}. We describe the implementation details:
\begin{itemize}
    \item \textit{MIP}: we set $\lambda_n=0.5$, $\gamma_1=50$, $\rho=1$, and $s_{\max}=10$; we use \texttt{MIPFocus=3} and impose a time limit of 60 seconds. 

    \vspace{0.05in}
    \item \textit{EILLS}: we set its penalty parameter to $\gamma=50$.

    \vspace{0.05in}
    \item \textit{ICP}: we use significance level $\alpha=0.05$. 

    \vspace{0.05in}
    \item \textit{Causal Dantzig}: we use the hidden variable implementation with $\alpha=0.1$.

    \vspace{0.05in}
    \item \textit{BIP}: we restrict the candidate models to size at most 10, use ridge stabilization $\epsilon=10^{-4}$, and report the maximum a posteriori model. No intercept is included for any of these methods.
\end{itemize}

\vspace{0.1in}
\subsection{Experiment of Figure~\ref{fig:E_n}}
\label{sec:app:figure-E_n}

\begin{itemize}
    \item \textbf{Left and middle panels.} 
    We generate data using the soft-intervention mechanism described above with ambient dimension $p=20$. We vary the number of environments over $|\mathcal E|\in\{2,3,5,10\}$ and the sample size per environment over $n^e\in\{50,75,\ldots,275\}$. For each combination of $(|\mathcal E|,n^e)$, we conduct 100 independent replications. The SEM is generated with $p_{\rm act}\in\{0.1,0.2\}$, $\ell=u=0.1$, $\ell_\star=0.5$, $u_\star=1$, $\sigma_{X,\min}=0.5$, $\sigma_{X,\max}=1$, $\sigma_Y=0.5$, $s_{\min}=4$, $s_{\max}=5$, and $p_{\rm resample}=0.05$.
    We apply the coordinate descent implementation of the invariant estimator with $\lambda_n=0.01$, $\gamma_1=5$, and $\rho=1$. Coordinate descent is initialized from the empty set. 

    \vspace{0.05in}
    \item \textbf{Right panel.} We use the same soft-intervention mechanism and DGP parameters as above, fixing $p=20$, $|\mathcal E|=10$, and $n^e=150$. We vary the tuning parameters over $\gamma\in\{1,5,10,15,20\}$ and $\lambda_n\in\{10^{-4},10^{-3},10^{-2},10^{-1},1\}$, while fixing $\rho=1$. For each combination of $(\gamma,\lambda_n)$, we conduct 50 independent replications. We compute the leave-one-environment-out (LOEO) validation criterion and apply coordinate descent with the corresponding tuning parameters, initialized from the empty set. Each cell in the heatmap reports the average LOEO criterion across replications, while red outlines indicate parameter combinations for which the proportion of perfect recovery of $S^\star$ exceeds $0.95$.

\end{itemize}

\vspace{0.1in}
\subsection{Experiment of Figure~\ref{fig:baseline_p}}
\label{sec:app:figure-baseline_p}
We generate data using the all-intervention mechanism described above. We vary the ambient dimension over $p\in\{5,10, \dots, 90,100\}$ and generate 10 environments with $n^e=500$ observations per environment. The number of causal parents is set to $|S^\star|=\lfloor0.4p\rfloor$ when $\lfloor0.4p\rfloor<16$, and $|S^\star|=16+\lfloor0.1p\rfloor$ otherwise. For each value of $p$, we conduct 200 independent replications. The SEM is generated with $p_{\rm act}\in\{0.1,0.15\}$, $\ell=0.2$, $u=0.5$, $\ell_\star=u_\star=0.5$, $\sigma_{X,\min}=0.5$, $\sigma_{X,\max}=1$, and $\sigma_Y=0.5$.

We compare our proposed method, \textit{MIP-CD}, with \textit{EILLS-s}, \textit{ICP-s,} \textit{Causal Dantzig}, \textit{BIP-s}, \textit{NegDRO}, \textit{CoCo}, and \textit{IRM}. For \textit{EILLS-s}, \textit{ICP-s}, and \textit{BIP-s}, we first screen the covariates using the magnitudes of ridge regression coefficients with regularization parameter $\lambda=1$, keeping the top 15 variables for \textit{EILLS-s} and the top 10 for \textit{ICP-s} and \textit{BIP-s}. We describe the implementation details below:
\begin{itemize}
    \item \textit{MIP-CD}: we set $\rho=1$, while $\lambda_n\in\{10^{-4},10^{-3},10^{-2},10^{-1},1\}$ and $\gamma\in\{1,5, 10, 15, 20\}$ are chosen based on LOEO criterion. Coordinate descent is initialized from the empty set.

    \vspace{0.05in}
    \item \textit{EILLS-s}: we set its penalty parameter to $\gamma=20$.

    \vspace{0.05in}
    \item \textit{ICP-s}: we use significance level $\alpha=0.05$.

    \vspace{0.05in}
    \item \textit{Causal Dantzig}: we use the hidden variable implementation with $\alpha=0.1$ and exclude the intercept.

    \vspace{0.05in}
    \item \textit{BIP-s}: we enumerate candidate models over the 10 screened variables, use ridge stabilization $\epsilon=10^{-4}$, and report the maximum a posteriori model without an intercept.

    \vspace{0.05in}
    \item \textit{NegDRO}: we set its penalty parameter to $\gamma=20$.

    \vspace{0.05in}
    \item \textit{CoCo}: we set the regularization parameter to $\lambda=10$, learning rate to $0.01$, and maximum number of iterations to $10{,}000$, using equal environment weights and excluding the intercept.

    \vspace{0.05in}
    \item \textit{IRM}: we set the penalty weight to $10$, penalty annealing iterations to $100$, $\ell_2$ regularization parameter to $10^{-3}$, learning rate to $10^{-3}$, and number of optimization steps to $500$.
    
\end{itemize}

\vspace{0.1in}
\subsection{Experiment of Figures~\ref{fig:demo-barplot}, \ref{fig:teaser1} and \ref{fig:teaser2}}
\label{sec:app:figure-teaser}
We use the structural equation model described in Example~\ref{ex:teaser}, with parameters $\gamma_3=\gamma_4=\beta_1^\star=\beta_3^\star=1$, $\eta_3=\eta_4=\sigma=1$, and intervention strength $a=3$. We additionally generate nine independent standard Gaussian covariates, yielding an ambient dimension of $p=13$. For each environment, we generate two independent datasets with sample sizes $n_1^e=300$ and $n_2^e\in\{400,800,\dots, 10000\}$. For each value of $n_2^e$, we conduct 200 independent replications. The true causal parent set is $S^\star=\{1,3\}$, while the plausible causal predictor set is $S(\mathcal E)=\{1\}$.

We compare \textit{MIP-test-2} and \textit{MIP-test-4} with \textit{ICP} and \textit{WVM}, as well as \textit{MIP}, \textit{EILLS}, \textit{Causal Dantzig}, and \textit{BIP-s}. For \textit{MIP-test-2} and \textit{MIP-test-4}, we use the first dataset to construct candidate supports and the second dataset for hypothesis testing. For the remaining methods, we combine the two datasets within each environment. We describe the implementation details below:
\begin{itemize}
    \item \textit{MIP-test-2} and \textit{MIP-test-4}: we set $\alpha=0.05$, $\overline M=5$, and use $B=300$ multiplier bootstrap replications. We set $\rho=1$ for \textit{MIP-test-2} and $(\gamma_1, \gamma_2, \gamma_3, \gamma_4)=(1,1,1,1)$ for \textit{MIP-test-4}. Both methods use coordinate descent to solve the underlying optimization problems, with candidate supports obtained through ridge screening.

    \vspace{0.05in}
    \item \textit{ICP}: we use significance level $\alpha=0.05$.

    \vspace{0.05in}
    \item \textit{WVM}: we use significance level $\alpha=0.05$ and $B=300$ bootstrap replications.

    \vspace{0.05in}
    \item \textit{MIP}: we use coordinate descent with $\lambda_n=0.01$, $\gamma_1=5$, and $\rho=1$, initialized from the empty set.

    \vspace{0.05in}
    \item \textit{EILLS}: we set its penalty parameter to $\gamma=20$ and exclude the intercept.

    \vspace{0.05in}
    \item \textit{Causal Dantzig}: we use the hidden variable implementation with $\alpha=0.05$ and exclude the intercept.

    \vspace{0.05in}
    \item \textit{BIP-s}: we first keep the top 10 variables using ridge screening, enumerate candidate subsets of these variables, and report the maximum a posteriori model. We use ridge stabilization $\epsilon=10^{-4}$ and exclude the intercept.
\end{itemize}

\vspace{0.1in}
\subsection{Experiment of Figure~\ref{fig:test-based}}
\label{sec:app:test-based}

We generate data using the all-intervention mechanism described above. We fix the ambient dimension at $p=10$ and the sample size at $n^e=500$ per environment, while varying the number of environments. For each setting, we conduct 200 independent replications. The SEM is generated with $p_{\rm act}\in\{0.1,0.15\}$, $\ell=0.2$, $u=0.5$, $\ell^\star=u^\star=0.5$, $\sigma_{X,\min}=0.5$, $\sigma_{X,\max}=1$, $\sigma_Y=0.5$, $s_{\min}=4$, and $s_{\max}=5$.

We compare \textit{MIP-test-2} and \textit{MIP-test-4} with \textit{ICP} and \textit{WVM}. For both MIP procedures, we randomly split the observations in each environment, using 20\% to construct candidate predictor sets and the remaining 80\% for hypothesis testing. For \textit{ICP} and \textit{WVM}, we use the full dataset. We describe the implementation details below:

\begin{itemize}
    \item \textit{MIP-test-2} and \textit{MIP-test-4}: we set $\alpha=0.05$, $\bar M=3$, and use $B=500$ multiplier bootstrap replications. We set $\rho=1$ for \textit{MIP-test-2} and $(\gamma_1, \gamma_2, \gamma_3, \gamma_4)=(1,1,1,1)$ for \textit{MIP-test-4}. Both methods use coordinate descent to solve the underlying optimization problems.

    \vspace{0.05in}
    \item \textit{ICP}: We use significance level $\alpha=0.05$.

    \vspace{0.05in}
    \item \textit{WVM}: We use significance level $\alpha=0.05$ and $B=300$ bootstrap replications.

\end{itemize}

\vspace{0.1in}
\subsection{Experiment of real data: the causal chambers in Section~\ref{sec:realdata}}
\label{sec:app:causal-chamber}

We consider the light tunnel dataset from Causal Chambers~\citep{gamella2025causal}, following the causal discovery setup of \citet{gu2025causality}. The response $Y=\widetilde I_3$ represents the infrared intensity measured after two polarizers. 
We consider 11 observed covariates: the RGB intensities $(R,G,B)$ of the light source, the two polarizer angles $(\theta_1,\theta_2)$, three measurements $(\widetilde V_1,\widetilde V_2,\widetilde V_3)$ of visible light intensity, two additional measurements $(\widetilde I_1,\widetilde I_2)$ of infrared intensity, and the measured electrical current $\widetilde C$. The known causal parent set is $S^\star= \{R,G,B,\theta_1,\theta_2\}$.
See Figure~\ref{fig:causal_chamber_dag} for the full ground truth.

\begin{figure}[!ht]
\centering
\begin{tikzpicture}[
    scale=0.7,
    >=Stealth,
    every node/.style={font=\small},
    every path/.style={line width=0.5pt},
    shorten >=2pt,
    shorten <=2pt
]
% Light source
\node (rgb) at (0,0) {$(R,G,B)$};
% Upper measurements (symmetric)
\node (v2) at (-3.0,2.2) {$\widetilde V_2$};
\node (i2) at (-1.0,2.2) {$\widetilde I_2$};
\node (i1) at ( 1.0,2.2) {$\widetilde I_1$};
\node (v1) at ( 3.0,2.2) {$\widetilde V_1$};
% Lower measurements
\node (c)  at (-3.0,-2.2) {$\widetilde C$};
\node (i3) at (-0.9,-2.2) {\textcolor{red}{$\widetilde I_3$}};
\node (v3) at ( 1.4,-2.2) {$\widetilde V_3$};
% Polarizer angles
\node (t1) at (4.0,-0.7) {$\theta_1$};
\node (t2) at (4.0,-2.2) {$\theta_2$};
% Effects of RGB intensities
\foreach \v in {v2,i2,i1,v1,c,i3,v3}{
    \draw[->] (rgb) -- (\v);
}
% Effects of polarizer angles
\draw[->] (t1) to[out=205,in=25] (i3);
\draw[->] (t1) -- (v3);
\draw[->] (t2) to[out=165,in=-25] (i3);
\draw[->] (t2) -- (v3);
% Effect of infrared intensity on visible measurement
\draw[->] (i3) -- (v3);
\end{tikzpicture}
\caption{\small 
Ground truth of the causal structure in the light tunnel chamber,
adapted from \citet{gu2025causality}.
The response is $Y=\widetilde I_3$, with causal parent set
$S^\star=\{R,G,B,\theta_1,\theta_2\}$.
}
\label{fig:causal_chamber_dag}
\end{figure}
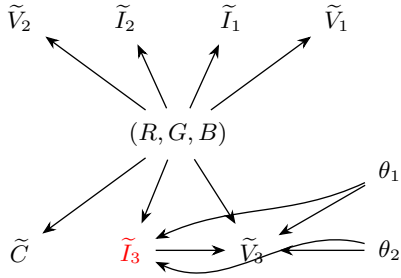

We consider one reference environment \texttt{uniform\_reference}, and 19 intervention environments corresponding to different manipulations of the physical system. 
We vary the number of environments over $|\mathcal E|\in \{3,4,\ldots,10\}$, with $n^e=1000$ observations per environment. Specifically, we conduct 100 independent trials. In each trial, we randomly sample 1000 observations from each environment and generate a random permutation of the 19 intervention environments. We then construct nested environment sets by always including the reference environment and sequentially adding intervention environments according to the random ordering. For example, when $|\mathcal E|=3$, we include the reference environment and the first two randomly selected intervention environments. When $|\mathcal E|=4$, we keep these three environments and add the next intervention environment. This procedure continues until $|\mathcal E|=10$, and the selected environments and their corresponding observations remain unchanged within each trial as additional environments are introduced. 
Finally, all covariates and the response are standardized using the empirical means and standard deviations computed from the reference environment, with the same standardization applied to all intervention environments.

We compare our proposed methods, \textit{MIP-CD} and \textit{MIP-test}, with \textit{EILLS}, \textit{ICP}, \textit{Causal Dantzig}, \textit{BIP}, and \textit{FAIR}. We describe the implementation details below:

\begin{itemize}
    \item \textit{MIP-CD}: we set $\lambda_n=0.01$, $\gamma_1=5$, and $\rho=1$. Coordinate descent is initialized from the empty set.

    \vspace{0.05in}
    \item \textit{MIP-test}: we first use 20\% of the observations in each environment to obtain a candidate support using \textit{MIP-CD}, with the same tuning parameters as above. We then restrict the subsequent variable selection to this candidate support, allowing \textit{MIP-test} to serve as a diagnostic procedure for \textit{MIP-CD}. The remaining observations are further divided into two subsets, containing 20\% and 80\% of the remaining data, respectively, for constructing $\mathcal{S}^\circ_k$ and conducting hypothesis tests. We use \textit{MIP-test-4} with $(\gamma_1,\gamma_2,\gamma_3,\gamma_4)=(1,1,1,1)$, significance level $\alpha=0.05$, $\overline M=5$, and $B=500$ multiplier bootstrap replications. The underlying optimization problems are solved using coordinate descent. 
    
    \vspace{0.05in}
    \item \textit{EILLS}: we set its penalty parameter to $\gamma=20$ and exclude the intercept.

    \vspace{0.05in}
    \item \textit{ICP}: we use significance level $\alpha=0.05$ and report the intersection of all accepted predictor sets.

    \vspace{0.05in}
    \item \textit{Causal Dantzig}: we use the hidden variable implementation with $\alpha=0.05$ and exclude the intercept.

    \vspace{0.05in}
    \item \textit{BIP}: we enumerate candidate models of size at most 10 among the 11 observed covariates, use ridge stabilization $\epsilon=10^{-4}$, and report the maximum a posteriori model without an intercept.

    \vspace{0.05in}
    \item \textit{FAIR}: we use the implementation based on neural networks with Gumbel variable selection. We set the invariance penalty to $\gamma=36$, the number of optimization iterations to 10,000, and the variable selection threshold to $0.9$. We use an accelerated Gumbel temperature annealing schedule to accommodate the reduced number of optimization iterations.
\end{itemize}

\vspace{0.05in}
For each randomized environment configuration, we evaluate the estimated causal parent set $\widehat S$ using the F1 score,
    $$
    \mathrm{F1}=\frac{2|\widehat S\cap S^\star|}
    {2|\widehat S\cap S^\star|+|\widehat S\setminus S^\star|+|S^\star\setminus\widehat S|},
    $$
    the number of false discoveries $|\widehat S\setminus S^\star|$, and runtime.

\vspace{0.1in}
\subsection{Experiment of Figure~\ref{fig:nonlinear}}
\label{sec:app:figure-nonlinear}
We generate data from the nonlinear structural equation model described in Example~\ref{ex:nonlinear}. We fix the ambient dimension at $p=6$ and the number of environments at $|\mathcal E|=5$, while varying the sample size per environment over $n^e\in \{100,150,\ldots,500\}$. The true causal parent set has size $|S^\star|=2$. For each sample size, we conduct 100 independent replications.
We use the all-intervention mechanism described above, replacing each linear structural relationship with the nonlinear function $\tanh(2\sin(x))$. The SEM is generated with $p_{\rm act}=0.2$, $\ell=0.1$, $u=0.5$, $\ell^\star=0.5$, $u^\star=1$, $\sigma_{X,\min}=0.5$, $\sigma_{X,\max}=1$, and $\sigma_Y=0.5$. Across environments, the structural coefficients and noise variances generating the covariates are independently resampled, while the nonlinear response mechanism remains invariant.

We compare our proposed \textit{nonlinear MIP} with \textit{nonlinear ICP}, \textit{WVM}, and \textit{FAIR}. We describe the implementation details below:

\begin{itemize}
    \item Nonlinear MIP: We expand each covariate using four Fourier basis functions: $\sin(x)$, $\cos(x)$, $\sin(2x)$ and $\cos(2x)$. We set $\lambda_n=0.01$, $\gamma=0.5$, and $\rho=1$. The resulting group $\ell_0$ optimization problem is solved using coordinate descent, initialized from the empty set.

    \vspace{0.05in}
    \item \textit{Nonlinear ICP}: We use the \texttt{nonlinearICP} package with the invariant residual distribution test and GAM fitting. We set the significance level to $\alpha=0.05$ and restrict the maximum candidate set size to two.

    \vspace{0.05in}
    \item \textit{WVM}: We use significance level $\alpha=0.05$ and $B=300$ bootstrap replications.

    \vspace{0.05in}
    \item \textit{FAIR}: We use the implementation based on neural networks with Gumbel variable selection. We set the invariance penalty to $\gamma=36$, the number of optimization iterations to 10,000, and the variable selection threshold to 0.9.
\end{itemize}

\vspace{0.1in}
\subsection{Experiment of Figure~\ref{fig:hidden}}
\label{sec:app:figure-hidden}
We generate data using the noise-intervention model with hidden confounding described in Example~\ref{ex:hidden}. We consider $|\mathcal{E}|=10$ environments and vary the sample size per environment over $n^e\in\{1000,1500,\ldots,5000\}$. For each sample size, we conduct 200 independent replications. We set $q=2$, $|S^\star|=4$, and independently sample $|A|,|D|\in\{2,3\}$, with eight additional independent standard Gaussian covariates. The structural coefficients are generated with $\beta_{S^\star}^\star=(1,1,1,1)^\top$, $\theta_d=0.1$ for all $d\in D$. The remaining coefficients for the relationships among $H$, $X_A$, $X_{S^\star}$, and $X_D$ are independently drawn from $\mathcal N(0,0.3^2)$. The coefficients for the direct effect of $H$ on $Y$ are independently drawn from $\mathcal N(0,1)$. We set the variance of the response noise to $\nu_Y=0.2^2$ and the variances of the covariate noise terms to $\nu_j^e=(0.2+0.05(e-1))^2$ for all $j\in A\cup S^\star\cup D$ and $e\in\{1,\ldots,10\}$. All structural coefficients remain fixed across environments. 

We compare our proposed method, \textit{MIP-hidden}, with \textit{EILLS-s}, \textit{Causal Dantzig}, \textit{BIP-s}, \textit{negDRO}, and \textit{CoCo}. For \textit{EILLS-s} and \textit{BIP-s}, we first screen the covariates using ridge regression, keeping the top 15 and 10 variables, respectively. We describe the implementation details below:
\begin{itemize}
    \item \textit{MIP-hidden}: we set $\lambda=0.5$, $\gamma=300$, and $s_{\max}=10$. We use \texttt{MIPFocus=3} and \texttt{Cuts=3}, and impose a time limit of 60 seconds.

    \vspace{0.05in}
    \item \textit{EILLS-s}: we set its penalty parameter to $\gamma=20$ and exclude the intercept.

    \vspace{0.05in}
    \item \textit{Causal Dantzig}: we use the hidden variable implementation with $\alpha=0.1$ and exclude the intercept.

    \vspace{0.05in}
    \item \textit{BIP-s}: we enumerate candidate models over the 10 screened variables, use ridge stabilization $\epsilon=10^{-4}$, and report the maximum a posteriori model without an intercept.

    \vspace{0.05in}
    \item \textit{negDRO}: we set its penalty parameter to $\gamma=20$.

    \vspace{0.05in}
    \item \textit{CoCo}: we set the regularization parameter to $\lambda=10$, learning rate to $0.01$, and maximum number of iterations to $10{,}000$, using equal environment weights and excluding the intercept.
\end{itemize}

\end{document}